\newlength{\textlarg}

\documentclass[12pt,english]{article}
\usepackage{subcaption} 
\usepackage[T1]{fontenc}
\usepackage[utf8]{inputenc}
\usepackage{geometry}
\usepackage{color}
\usepackage{babel}
\usepackage{float}
\usepackage{mathrsfs}
\usepackage{url}
\usepackage{amsmath}
\usepackage{amsthm}
\usepackage{amssymb}
\usepackage{graphicx}
\usepackage{setspace}
\usepackage{booktabs,threeparttable}
\usepackage[authoryear]{natbib}
\usepackage[unicode=true,pdfusetitle,
 hypertexnames=false,bookmarks=true,bookmarksnumbered=false,bookmarksopen=false,
 breaklinks=false,pdfborder={0 0 0},pdfborderstyle={},backref=false,colorlinks=true]
 {hyperref}
\hypersetup{
 citecolor=blue, linkcolor=blue, urlcolor=blue}

\usepackage[title]{appendix}

\makeatletter

\usepackage{paralist}
  \theoremstyle{definition}
  \newtheorem{assumption}{\protect\assumptionname}
  \theoremstyle{plain}
 \theoremstyle{definition}

\usepackage{bibunits}

  \usepackage{enumitem}

\usepackage{multirow}
\usepackage{adjustbox}
\usepackage{booktabs}
\usepackage{siunitx} 
\usepackage{times}
\usepackage{appendix}

\makeatother

  \providecommand{\assumptionname}{Assumption}

\newtheorem{proposition}{Proposition}

\usepackage{graphicx}
\usepackage{subcaption}
\usepackage{float}

\makeatletter
\renewenvironment{proof}[1][\proofname] {\par\pushQED{\qed}\normalfont\topsep6\p@\@plus6\p@\relax\trivlist\item[\hskip\labelsep\bfseries#1\@addpunct{.}]\ignorespaces}{\popQED\endtrivlist\@endpefalse}
\makeatother

\makeatletter
\let\SNPDnatlinkstart\hyper@natlinkstart
\def\hyper@natlinkstart#1{\SNPDnatlinkstart{\@bibunitname.#1}}
\let\SNPDnatlinkbreak\hyper@natlinkbreak
\def\hyper@natlinkbreak#1#2{\SNPDnatlinkbreak{#1}{\@bibunitname.#2}}
\let\SNPDnatanchorstart\hyper@natanchorstart
\def\hyper@natanchorstart#1{\SNPDnatanchorstart{\@bibunitname.#1}}
\makeatother

\begin{document}

\title{Scalable Nonparametric Demand Estimation in Differentiated Product Markets\thanks{I thank Steven Berry, Debopam Bhattacharya, Philippe Chon\'e, Federico Ciliberto, Christopher Conlon, Pierre Dubois, Olivier De Groote, Andr\'e de Palma, Xavier D'Haultfoeuille, Roxana Fernandez, Alfred Galichon, Gautam Gowrisankaran, Rachel Griffith, Laura Grigolon, Brett Hollenbeck, Alessandro Iaria, Ulrich Kaiser, Laura Lasio, Laurent Linnemer, Hugo Molina, Lars Nesheim, Ariel Pakes, Michelangelo Rossi, Pasquale Schiraldi, Senay Sokullu, Yutec Sun, Frank Verboven, Thibaud Verg\'e, and Daniel Wilhelm. I also thank seminar participants at Baleares Islands, Bristol, Chicago, CMA, Copenhagen, CREST, IMT Business School, INRAE, Mannheim, Sciences Po, Stockholm School of Economics, Telecom Paris, TSE, and Vienna, as well as workshop/conference participants at Workshop on Platforms, CREST-Telecom Paris E-commerce and Digital Economics 2019, TSE Conference on IO and the Food Industry 2019, EARIE 2019/2021, ES World Congress 2020, DSE Winter School 2020, ASSA 2021, RES 2021, ES Asian Meeting 2021, Leuven Summer Event 2023, CEPR Paris Symposium 2023, Bristol Empirical IO Workshop 2024, and BSE Summer Forum 2026. This paper supersedes an earlier draft that circulated under the title "Measuring Substitution Patterns with a Flexible Demand Model".}}

\author{Julien Monardo\thanks{School of Economics, University of Bristol; \protect\url{julien.monardo@bristol.ac.uk}. The Priory Road Complex, Priory Road, Clifton BS8 1TU, United Kingdom}}

\date{October 6, 2026}\newgeometry{tmargin=1in,bmargin=1in,lmargin=1.1in,rmargin=1.1in}
\maketitle

\vspace{-0.0cm}
\begin{abstract}

The answers to many economic questions depend on the slope and curvature of demand. Estimating demand for differentiated products entails a trade-off between flexibility and scalability. Existing nonparametric approaches face a curse of dimensionality in the number of products. I develop a scalable nonparametric approach using market-level data that overcomes this curse through economically motivated restrictions embedded in many standard models. Simulations and an application to the U.S. beer market demonstrate flexibility in demand slope and curvature, with consequential differences in counterfactual price responses. Estimation in the application takes seconds, versus several hours for a random coefficient nested logit model.

\vspace{0.25cm}

\end{abstract}
\restoregeometry

\newpage 

\begin{bibunit}

\section{Introduction}

A popular approach to answering many economic questions is to conduct counterfactual analysis using estimated models of supply and demand for differentiated products.\footnote{Another popular approach is the "treatment effects" approach (e.g., difference-in-differences analysis).} Prominent examples include classic industrial organization questions, such as merger evaluation and the measurement of market power, and applications in many other fields.\footnote{See Table 1 in \citet*{berryhaile2016} and Footnote 1 in \citet*{compiani2018} for several examples.} Given the assumed supply model, the answers to these questions depend on the slope (summarized by price elasticities or diversion ratios) and curvature (how the slope changes with prices) of demand implied by the estimated demand model. Credible counterfactual analysis thus requires both a demand model flexible enough to accommodate a wide range of slopes and curvatures and an estimation method that can recover these features of demand accurately from the data and scale to modern settings with many products and markets.

Estimating demand for differentiated products therefore entails a trade-off between flexibility and scalability: because demand for each product depends not only on its own price and characteristics but also on those of its competitors, the dimension of the problem grows with the number of products per market. Standard logit-type models achieve scalability by reducing the dimension of the problem through parametric functional form restrictions. The leading BLP approach of \citet*{BLP1995} relaxes these restrictions through a random coefficient logit (RCL) demand model whose flexibility depends on the flexibility of the random coefficient specification.\footnote{Flexibility requires sufficiently flexible utility and random coefficient specifications \citep*{mcfadden2000mixed}.} For tractability, empirical applications often impose strong distributional assumptions on random coefficients, such as normality, that can restrict substitution patterns. By contrast, \citet*{compiani2018} develops a nonparametric estimation approach that avoids arbitrary distributional and functional form assumptions but faces a severe curse of dimensionality limiting its scalability: because its polynomial approximation uses product-specific arguments, the number of polynomial coefficients grows rapidly with the number of products in the market.\footnote{His nonparametric approach builds on the insight of \citet*{berry2014identification} that most assumptions commonly imposed in demand estimation are unnecessary for nonparametric identification of demand.}

Recent papers underscore the need for flexibility in both the slope and curvature of demand. \citet*{miravete2023elasticity} show that similar elasticities can coexist with sharply different curvature, and identify the curse of dimensionality as an obstacle to nonparametric alternatives. \citet*{birchall2024estimating} show that, when price enters utility linearly, own-price elasticities in logit and RCL models scale quasi-linearly with price, and that a more flexible price functional form relaxes this scaling and adds flexibility in demand curvature. \citet*{goldszmidt2026price} show that commonly used RCL specifications can generate more flexible cross-price elasticities while leaving own-price elasticities close to quadratic functions of own price, noting that they remain functions of own price alone, though not necessarily quadratic, when price enters through a transformation such as the logarithm.\footnote{Their approximation requires small type-specific choice probabilities and independent price and non-price random coefficients. Allowing these coefficients to be correlated can relax the restriction.}

In this paper, I respond to this need by developing a nonparametric approach to estimating demand for differentiated products from market-level data that combines flexibility in demand with scalability in the number of products. Following the literature, I consider markets with products differentiated by observed and unobserved characteristics and address the resulting endogeneity using generalized method of moments (GMM) estimation with appropriate instruments. The approach achieves scalability by reducing the dimension of the problem through mild, economically motivated restrictions embedded in many standard demand models, rather than through parametric functional-form or distributional assumptions. The resulting model accommodates a wide range of substitution patterns, own-price responses, and demand curvature.

Specifically, I use the following strategy to achieve scalability. First, I impose invariance restrictions. In particular, permutation invariance requires the market share of a product to remain unchanged under any permutation of its competitors. Second, I show that, under these restrictions, the inverse market share of any product can be represented as the sum of the logit inverse term and a permutation-invariant function of competitor-relative quantities (market share ratios and differences in prices and characteristics). Third, I use the approximation result of \citet*{bachmayr2024polynomial} to approximate this component with a permutation-invariant polynomial expansion. By using competitor-pooled rather than competitor-specific arguments, the construction breaks the curse of dimensionality in the number of products: for a fixed polynomial degree, the number of polynomial coefficients no longer grows with the number of products.

\begin{table}[htbp]
\centering
\caption{Number of coefficients in degree-$3$ approximations}
\label{tabdim1}

\begin{threeparttable}
\small
\setlength{\tabcolsep}{3.5pt}
\renewcommand{\arraystretch}{1.15}

\begin{tabular*}{0.95\textwidth}
{@{\extracolsep{\fill}} r l rrr rrr @{}}
\toprule
& &
\multicolumn{3}{c}{Shares only ($d=1$)} &
\multicolumn{3}{c}{Shares and prices ($d=2$)} \\
\cmidrule(lr){3-5}
\cmidrule(lr){6-8}
$J$ & Application
& SNPD & NPD & Full
& SNPD & NPD & Full \\
\midrule

5
&
& 7
& 80
& 256
& 23
& 13{,}056
& 65{,}536 \\

7
&
& 7
& 224
& 4{,}096
& 23
& 248{,}064
& $1.68\times10^{7}$ \\

10
&
& 7
& 660
& 262{,}144
& 23
& $7.85\times10^{6}$
& $6.87\times10^{10}$ \\

\addlinespace[2pt]

24
& Nevo (2000)
& 7
& 9{,}200
& $7.04\times10^{13}$
& 23
& $1.50\times10^{11}$
& $4.95\times10^{27}$ \\

39
& MW (2017)
& 7
& 39{,}520
& $7.56\times10^{22}$
& 23
& $7.17\times10^{13}$
& $5.71\times10^{45}$ \\

150
& BLP (1995)
& 7
& $2.25\times10^{6}$
& $5.09\times10^{89}$
& 23
& $9.59\times10^{21}$
& $2.59\times10^{179}$ \\

\bottomrule
\end{tabular*}

\begin{tablenotes}[flushleft]
\footnotesize
\item \textit{Notes:}
Each entry is the number of coefficients in a degree-$3$ polynomial
approximation of the inverse market share function. SNPD uses total degree at most 3. NPD and Full use degree at most 3 in each coordinate. $J$ denotes the number of products, and $d$ denotes the number of relative
objects per competitor ($d=1$ includes market share ratios only and
$d=2$ includes market share ratios and price differences). SNPD is the proposed symmetric polynomial sieve, which is permutation invariant by construction.
NPD is the sieve of \citet*{compiani2018}, which imposes permutation
invariance through linear equality constraints.
Full is the unrestricted tensor-product sieve. MW (2017) stands for \citet*{miller2017understanding}.
\end{tablenotes}

\end{threeparttable}
\end{table}

Table~\ref{tabdim1} quantifies the scalability gain from the approximation used in the proposed approach (SNPD) relative to the tensor-product approximation of \citet*{compiani2018}, which imposes permutation invariance through linear equality restrictions (NPD), and an unrestricted tensor-product approximation (Full). Both NPD and Full become infeasible with even a moderate number of products. By contrast, provided the polynomial degree does not exceed the number of competitors of each product, which is typical in practice, the number of SNPD coefficients is independent of the number of products, preserving feasibility in the canonical applications shown in the table. Hence, the model remains estimable in markets with many products, and the polynomial degree can be chosen by cross-validation. For example, in my empirical application to the U.S. beer market, where markets contain up to 39 products, cross-validation selects a cubic sieve with only 53 nonconstant terms.

The proposed approach is computationally attractive: estimation reduces to a constrained convex GMM problem. I regularize estimation with theory-based shape restrictions \citep*{matzkin1994,chetverikov2018}, such as the law of demand. These restrictions improve finite-sample stability and discipline the implied substitution patterns. Because they are linear in the sieve coefficients, imposing them preserves the convexity of the GMM problem. Convexity eliminates local optima and allows efficient optimization using standard software.

The framework also generates two extensions. First, the same strategy extends to flexible universal logit models, which can approximate a large class of discrete choice demand models \citep*{mcfadden1984econometric}. Under permutation invariance across competitors, the symmetric sieve again makes the number of coefficients independent of the number of products. Second, the supplemental appendix develops a parametric counterpart to the nonparametric inverse market share model that extends the nested logit inverse market share model by assigning a separate nest to each pair of products with nesting parameters that vary with product proximity in characteristic space.

Directly modeling inverse market shares rather than specifying individual-level utilities offers three advantages. First, the inverse market share function, unlike the market share function itself, yields equations that can be estimated directly by regression methods \citep*{berry1994estimating}. Second, the inverse market share function and its derivatives provide the key inputs needed to answer many economic questions of interest.\footnote{See references in \citet*{compiani2018} and \citet*{miravete2023elasticity} for examples of questions.} Combined with a supply model, the estimated inverse market share can be used to recover markups and to conduct counterfactual analyses of mergers, product introductions, and changes in tax and trade policies. Third, directly estimating inverse market shares avoids imposing restrictions on individual consumer behavior and therefore accommodates a broad range of behaviors and substitution patterns, including multiple choices and complementarity in demand.

The focus on market-level data remains relevant even when researchers observe individual-level (micro) data: micro moments supplement rather than replace aggregate moments \citep*{petrin2002,berry2004differentiated,conlon2023incorporating}. Market-level methods therefore remain useful even in settings with richer data. Market-level data are also often easier to obtain and remain widely used in research and policy practice.\footnote{Recent academic applications include \citet*{dopper2025rising}, \citet*{starc2025does}, and \citet*{ale2026colluding}. Recent competition policy applications include the European Commission's \href{https://eur-lex.europa.eu/legal-content/EN/TXT/?uri=SWD\%3A2020\%3A363\%3AFIN}{Digital Markets Act} impact assessment (2020, EU), the European Commission's \href{https://eur-lex.europa.eu/legal-content/EN/TXT/PDF/?uri=CELEX:52022SC0319}{AI Liability Directive} impact assessment (2022, EU), \href{https://www.gov.uk/cma-cases/arcelik-slash-whirlpool-emea-merger-inquiry}{Ar\c{c}elik / Whirlpool EMEA} (2024, UK), and \href{https://law.justia.com/cases/federal/district-courts/california/candce/3\%3A2020cv02345/357628/635/}{In re JUUL Labs Antitrust Litigation} (2026, US).}

The proposed approach nevertheless faces two limitations. First, because it specifies inverse market shares directly rather than deriving them from an individual-level utility model, individual-level welfare and counterfactual analyses require further assumptions on preference heterogeneity. One promising direction is to combine micro and aggregate moments \citep*{petrin2002,berry2004differentiated,conlon2023incorporating}. Second, it still faces a curse of dimensionality in the number of product characteristics included in the inverse market share function. As in BLP, which cannot accommodate arbitrarily many random coefficients, this limits the number of characteristics that can be included in practice. Possible strategies include principal component analysis
\citep*{backus2021common}, characteristic selection based on IIA tests
\citep*{gandhihoude}, single-index specifications
\citep*{horowitz2012semiparametric}, and regularization of the
symmetric polynomial basis \citep*{hastie09elements}. 

I assess the finite-sample performance of the proposed approach in
Monte Carlo simulations by examining the estimated substitution
patterns, demand curvature, and the implied
markups under a static price
competition model among multiproduct firms. First, the proposed
approach performs well in recovering the substitution patterns, demand curvature, and markups generated by BLP using normally distributed random coefficients on
price and a product characteristic, even under substantial unobserved
preference heterogeneity. Second, the approach can outperform BLP
when BLP is misspecified in either its logit functional form or its
random coefficient distribution. In one design, the data are generated
by a parametric version of the nonparametric inverse market share
model. In the other, the true price coefficient follows a lognormal
distribution, while BLP estimation imposes normality. Overall, the simulations show that the proposed approach performs well across a range of data-generating processes without imposing parametric functional-form or distributional assumptions.

I further evaluate the proposed approach in an empirical application to the U.S. beer market studied by \citet*{miller2017understanding}.
Relative to their random coefficient nested logit (RCNL) specification,
the proposed approach implies greater substitution toward closer
products while generating similar average own- and cross-price
elasticities. Under SNPD, own-price elasticities vary more at a given
price, even without a price-income interaction, so they are not functions of own price alone, and their relationship with prices is nonmonotonic, in contrast
to the increasing relationship in absolute value under the estimated
RCNL model. SNPD also generates a wider range of demand curvature,
with values both below and above one. These results illustrate a practical response to the restrictions on own-price elasticities documented by \citet*{birchall2024estimating} and \citet*{goldszmidt2026price} and on demand curvature documented by \citet*{miravete2023elasticity}. These differences matter for counterfactual analysis: the two demand specifications imply substantially different responses to a cost shock.
Under the same static price competition model, a per-unit tax on
imported products equal to $10\%$ of their baseline prices raises
domestic prices by $6.92\%$ under SNPD, compared with $1.53\%$ under
RCNL.
Thus, similar average elasticities can conceal substantial differences in the
predicted transmission of cost shocks across products. This application
also highlights the scalability gain of the approach: with $2{,}808$
markets and up to $39$ products per market, SNPD takes approximately
$30$ seconds to estimate, compared with approximately $10$ hours and
$30$ minutes for the RCNL model.

\paragraph{Related literature.} 

This paper builds on the large literature on demand estimation in differentiated product markets pioneered by \citet*{berry1994estimating} and \citet*{BLP1995}, but departs from it by specifying an inverse market share model rather than an individual-level utility model.\footnote{A prominent alternative uses flexible functional forms to specify demand directly in observables, typically with additive errors \citep[see][and references therein]{barnettserletis2008}. Relative to \citet*{berry1994estimating} and \citet*{BLP1995}, this approach operates in product rather than characteristics space and therefore faces a curse of dimensionality in the number of products. It also restricts product-specific unobservables to affect only own-product demand.} The standard approach derives market shares from utility maximization and inverts them to obtain the estimating equations. Except for the logit and nested logit models, this inversion has no closed form, ruling out regression-based estimation and requiring more involved procedures.\footnote{See \citet*{conlon2018}. Estimation repeatedly performs numerical inversion within a nonlinear, nonconvex optimization problem and, for random coefficient models, simulates the market share function. This raises familiar concerns about local optima, starting values, and the accuracy of numerical inversion and simulation \citep[see][and references therein]{knittel2014estimation}.} I instead specify the inverse market share function directly and estimate it by regression.

Among approaches that estimate an inverse market share function, the closest is \citet*{compiani2018}, who develops a nonparametric approach to estimating demand for differentiated products. His approach, however, suffers from a curse of dimensionality in the number of products. I build on it to develop a nonparametric estimator that overcomes this curse. \citet*{brand2025quasi} extend \citet*{compiani2018} with a quasi-Bayesian approach to better enforce economic constraints but inherit the same dimensionality problem. Other closely related papers estimate inverse demand within more structured models: \citet*{Nevo2026} and \citet*{lu2023semi} build on the RCL model, whereas \citet*{fosgerau2018demand} and \citet*{Hortacsu2020} rely on nesting structures.

The paper also relates to approaches that seek to generate flexible substitution patterns. These include semiparametric and nonparametric methods for estimating random coefficients in the RCL model \citep*{lu2023semi,wang2021sieve}, as well as approaches that map demand parameters into characteristics space using the distance-metric approach of \citet*{pinkse2002spatial} \citep*{pinkse2004mergers,slade2004,davis2014flexible}. Like the present approach, distance-metric approaches allow regression-based estimation. However, they restrict product-specific unobservables to affect only own-product demand. \citet*{wang2025scalable}, in turn, estimate price elasticities nonparametrically without specifying a demand function.

Finally, the paper contributes to work that balances flexibility and scalability. This literature reduces dimensionality by summarizing substitution patterns with lower-dimensional objects \citep*{pinkse2004mergers,fosgerau2018demand,fally2026consumer}, simplifying estimation of the RCL model \citep*{salanie2022,lu2023semi}, or exploiting permutation invariance \citet*{compiani2018,Nevo2026}. This paper follows the third approach. \citet*{compiani2018} imposes permutation invariance through equality constraints on polynomial coefficients and calls for ``additional ways to tackle the curse of dimensionality''. Within random coefficient models that are linear in characteristics, \citet*{gandhihoude} and \citet*{Nevo2026} derive closely related representations, in which inverse demand involves a permutation-invariant function of competitors' market shares and characteristic differences. I obtain such a representation from invariance and regularity restrictions on the market share function, without a random utility structure. Like \citet*{Nevo2026}, I address the curse of dimensionality in the number of products by pooling information across competitors. Their framework accommodates linear sieves and Deep Sets neural networks \citep*{zaheer2017}. Their preferred specification uses share-weighted sums of transformed characteristic differences, followed by a second step that recovers the preference distribution in an RCL model. Building on \citet*{bachmayr2024polynomial}, I construct a symmetric polynomial sieve of share ratios and price and characteristic differences, without specifying a random coefficient distribution. Both linear-sieve formulations permit convex GMM estimation under linear economic shape restrictions. I further develop constraint reduction and global coefficient bounds to impose these restrictions efficiently. \citet*{singh2023choice} use Deep Sets to estimate the market share function directly with neural networks.

\paragraph{Organization of the paper.} The remainder of the paper is organized as follows. Section \ref{Section:FD} introduces the general setting and develops the scalable nonparametric demand estimation approach. Section \ref{Section:Simul} presents Monte Carlo evidence on its empirical performance. Section \ref{sec:empirical_applications} presents an application of the proposed approach to the U.S. beer market. Section \ref{Section:CCL} concludes. 

\section{Scalable Nonparametric Demand Estimation \label{Section:FD}}

\subsection{General Setting}

I adopt the general setting of \citet*{berry2014identification}. In market $t$, consumers choose among $J+1$ differentiated products in $\mathcal{J} = \mathcal{J}_0\cup\{0\}$, where $\mathcal{J}_0=\{1,\ldots,J\}$ contains the inside products and $j=0$ is the outside good. 

In each market $t$, let $\mathcal{X}_t=(\mathbf{x}_t,\mathbf{p}_t,\boldsymbol{\xi}_t)$, with support $\mathcal{X}$, collect product/market characteristics, where $\mathbf{x}_t=(\mathbf{x}_{0t},\ldots,\mathbf{x}_{Jt})$ is the matrix of observed characteristics, $\mathbf{p}_t=(p_{0t},\ldots,p_{Jt})$ is the vector of observed prices, and $\boldsymbol{\xi}_t=(\xi_{0t},\ldots,\xi_{Jt})$ is the vector of  unobserved characteristics, with $\mathbf{x}_{jt}\in\mathbb{R}^K$, $p_{jt}\in\mathbb{R}_+$, and $\xi_{jt}\in\mathbb{R}$ for each $j\in\mathcal{J}$. For the outside good, $\mathbf{x}_{0t}=\mathbf{0}$, $p_{0t}=0$, and $\xi_{0t}=0$. The unobservables $\xi_{jt}$ capture product differentiation not explained by observed characteristics. They are observed by consumers and firms but not by the modeler. As is standard, $\mathbf{x}_t$ is assumed exogenous, that is, uncorrelated with $\xi_{jt}$, whereas $\mathbf{p}_t$ is assumed endogenous.\footnote{For instance, in price competition models, firms set prices while taking unobserved product characteristics into account, which makes prices endogenous.}

Aggregating individual choices yields the vector of positive market shares $\mathbf{s}_t = (s_{0t},\ldots,s_{Jt}) \in \Delta_{J}^+$, where $\Delta_{J}^+= \left\{ \mathbf{s}_t \in\left(0,\infty \right)^{J+1}:\sum_{j\in\mathcal{J}} s_{jt}=1\right\}$ denotes the set of positive market share vectors.\footnote{Formally, $\Delta_{J}^+$ is the relative interior of the unit simplex in $\mathbb{R}^{J+1}$.} Aggregate consumer choice is represented by the market share function $\boldsymbol{\sigma}=(\sigma_0,\ldots,\sigma_J)$, which maps $\mathcal{X}_t \in \mathcal{X}$ into $\mathbf{s}_t \in \Delta_{J}^+$:
\begin{equation}\label{eq:GeneralD}
   s_{jt} = \sigma_j \left(\mathcal{X}_t\right) = \sigma_j \left(\mathbf{x}_{t},\mathbf{p}_{t}, \boldsymbol{\xi}_{t}\right),\quad j\in\mathcal{J}.
\end{equation}

I develop a scalable nonparametric approach to estimating \eqref{eq:GeneralD} from market-level data using mild, economically motivated restrictions embedded in many standard demand models. Proofs for this section are provided in Appendix \ref{Appendix:Proof}.  

\subsection{Nonparametric Demand Model \label{Sec:NPDModel}}

Three assumptions are imposed on the market share function in \eqref{eq:GeneralD}.\footnote{It also relies on two standard implicit assumptions: scalar unobservables $\xi_{jt}$ and income-independent demand.} The first assumption imposes a linear index structure that restricts how \(\mathcal{X}_t\) enters the market share function \citep*{berry1994estimating, BLP1995, berry2014identification}. Partition the observed characteristics $\mathbf{x}_t$ into $\mathbf{x}_t^{(1)} = (\mathbf{x}_{0t}^{(1)},\ldots,\mathbf{x}_{Jt}^{(1)})$ and $\mathbf{x}_t^{(2)} = (\mathbf{x}_{0t}^{(2)},\ldots,\mathbf{x}_{Jt}^{(2)})$, with each $\mathbf{x}_{jt}^{(1)} \in \mathbb{R}^{K_1}$ and $\mathbf{x}_{jt}^{(2)} \in \mathbb{R}^{K_2}$, where $K_1 + K_2 = K$. Formally:

\begin{assumption}[Index]\label{Ass1} For each market $t$, there exists a vector of linear indices  $\boldsymbol{\delta}_t = (\delta_{0t},\ldots,\delta_{Jt})$ with $\delta_{jt} =  \mathbf{x}_{jt}\boldsymbol{\beta} - \alpha p_{jt}+\xi_{jt}$ for all $j \in\mathcal{J}$ such that 
\begin{equation}\label{eq:GeneralDIndex}
   s_{jt} = \sigma_j \left(\mathcal{X}_t\right)  = \sigma_j (\boldsymbol{\delta}_t,\mathbf{p}_t,\mathbf{x}_t^{(2)}), \quad j\in\mathcal{J}.
\end{equation}
\end{assumption}

The index restriction requires $\mathbf{x}^{(1)}_t$ and $\boldsymbol{\xi}_t$ to enter the market share function only through $\boldsymbol{\delta}_t$, leaving unrestricted how $\mathbf{p}_t$ and $\mathbf{x}_t^{(2)}$ enter it and jointly affect market shares with $\boldsymbol{\delta}_t$. In a BLP-type random coefficient logit model, $\mathbf{x}_t^{(1)}$ and  $\mathbf{x}_t^{(2)}$ enter utility with fixed and random coefficients, respectively. The logit model does not include $\mathbf{x}_t^{(2)}$. In models with nesting structures, such as the nested logit model, $\mathbf{x}_t^{(2)}$ collects the characteristics used to form nests.  The outside good normalizations imply that $\delta_{0t}=0$.

The second assumption imposes economically motivated invariance restrictions on the market share function in \eqref{eq:GeneralDIndex}.\footnote{Each invariance restriction is imposed only for transformations that leave the arguments within the domain of $\boldsymbol\sigma$. Translation invariance is stated on the unnormalized index domain; $\delta_{0t}=0$ selects the normalized representation.} As shown below, these restrictions make the approach scalable.
\begin{assumption}[Invariance]\label{Ass5}\noindent
\begin{enumerate}
    \item \label{Ass5.1} (Product) There exists a function $\sigma$,
common to all inside products, such that, for every $j\in\mathcal J_0$,
\[\sigma_j(\boldsymbol{\delta}_t,\mathbf p_t,\mathbf x_t^{(2)})
=\sigma\Big( (\delta_{jt},p_{jt},\mathbf x_{jt}^{(2)});
(\delta_{0t},p_{0t},\mathbf x_{0t}^{(2)});
\big((\delta_{kt},p_{kt},\mathbf x_{kt}^{(2)})\big)_{k\in\mathcal J_0\setminus\{j\}}\Big).\]  
     \item \label{Ass5.2} (Permutation) For each $j\in\mathcal J_0$ and any permutation $\pi$ of the indices of product $j$'s
inside competitors $\mathcal J_0\setminus\{j\}$,
\[
\begin{aligned}
&\sigma_j((\delta_{jt},p_{jt},\mathbf x_{jt}^{(2)});
(\delta_{0t},p_{0t},\mathbf x_{0t}^{(2)});
((\delta_{kt},p_{kt},\mathbf x_{kt}^{(2)}))_{k\in\mathcal J_0\setminus\{j\}})\\
&\quad=\sigma_j((\delta_{jt},p_{jt},\mathbf x_{jt}^{(2)});
(\delta_{0t},p_{0t},\mathbf x_{0t}^{(2)});
((\delta_{\pi(k)t},p_{\pi(k)t},\mathbf x_{\pi(k)t}^{(2)}))_{k\in\mathcal J_0\setminus\{j\}}).
\end{aligned}
\]
\item \label{Ass5.3} (Translation) For any $c\in\mathbb R$ and every $j\in\mathcal J$,
\[
\sigma_j(\boldsymbol{\delta}_t+c\mathbf1,\mathbf p_t,\mathbf x_t^{(2)})
=\sigma_j(\boldsymbol{\delta}_t,\mathbf p_t,\mathbf x_t^{(2)}).
\]
\end{enumerate}
\end{assumption}

The outside good market share is determined by
\[
\sigma_0(\boldsymbol\delta_t,\mathbf p_t,\mathbf x_t^{(2)})
=1-\sum_{j\in\mathcal J_0}
\sigma_j(\boldsymbol\delta_t,\mathbf p_t,\mathbf x_t^{(2)}).
\]

Many standard demand models, including the RCL model, satisfy these restrictions. Product invariance assigns every product $j \in\mathcal{J}_0$ the same market share function, eliminating product-specific functional forms.\footnote{Product invariance is sometimes referred to as symmetry \citep[see, e.g.,][]{gandhihoude}.} Permutation invariance leaves the market share of any inside product \(j\) unchanged under any relabeling of its inside competitors in \(\mathcal{J}_0\setminus\{j\}\): competitors matter through their characteristics, not their identities.\footnote{Permutation invariance is also referred to as exchangeability \citep*[see, e.g.,][]{pakes1994computing,BLP1995,compiani2018} and symmetry \citep*[see, e.g.,][]{bachmayr2024polynomial,gandhihoude}.} Translation invariance in \(\boldsymbol{\delta}_t\) makes market shares depend only on index differences, thereby justifying the normalization \(\delta_{0t}=0\). Product and permutation invariance apply to inside products. They do not require the outside good share to follow the same functional form or impose invariance to permutations that exchange the outside good with an inside competitor. Translation invariance applies to all products.

The final assumption imposes regularity conditions that ensure the market share model is well behaved.
\begin{assumption}[Regularity]\label{AssReg}
    The market share function $\boldsymbol \sigma$, defined for all $\boldsymbol\delta_t\in\mathbb R^{J+1}$, is continuously differentiable and satisfies
$\lim_{\delta_{jt} \to +\infty} \sigma_j(\boldsymbol{\delta}_t,\mathbf{p}_t,\mathbf{x}^{(2)}_t) = 1$ 
for all $j \in \mathcal{J}$ and for any fixed finite competitor indices $\{\delta_{kt}\}_{k \ne j}$. Its substitution matrix $[\partial \sigma_j(\boldsymbol{\delta}_t,\mathbf{p}_t,\mathbf{x}^{(2)}_t)/\partial \delta_{kt}]_{0\leq j,k \leq J}$ is everywhere symmetric positive semidefinite  with the inside good block positive definite.
\end{assumption}

Every additive random utility model obeying mild regularity assumptions generates a market share function that satisfies these conditions \citep*{dalyzachary,anderson1992discrete,koning2003}. These include the logit, nested logit, and their mixtures, such as the RCL model and the random coefficients nested logit model. The conditions ensure that  \eqref{eq:GeneralDIndex} is invertible in $\boldsymbol{\delta}_t$ and, when prices do not enter $\sigma$ directly, consistent with consumer theory.

For each competitor $k$, define its market share ratio and price and characteristic
differences relative to product \(j\) as
\[
r_{jk,t}=\frac{s_{kt}}{s_{jt}}\in\mathbb{R}_{++},\qquad
d_{jk,t}^{p}= p_{kt}-p_{jt}\in\mathbb{R},\qquad
\mathbf{d}_{jk,t}^{(2)}
=\mathbf{x}_{kt}^{(2)}-\mathbf{x}_{jt}^{(2)}
\in\mathbb{R}^{K_2},
\]
where the \(l\)th element of \(\mathbf{d}_{jk,t}^{(2)}\) is $d_{jk,t}^{(2,l)}
= x_{kt}^{(2,l)}-x_{jt}^{(2,l)}
\in\mathbb{R}$, $l=1,\ldots,K_2$. 

As \citet*{berry1994estimating} and \citet*{berry2014identification} emphasize, estimation targets the inverse market share function rather than the market share function itself. Under Assumptions \ref{Ass1} -- \ref{AssReg}, the next proposition establishes invertibility of the market share function and characterizes its inverse.

\begin{proposition}[Inverse market share representation]\label{prop:Main}
Suppose that Assumptions \ref{Ass1} -- \ref{AssReg} hold. Then the market share function
$\boldsymbol{\sigma}$ is invertible in $\boldsymbol{\delta}_t$ with $\delta_{0t}=0$. Moreover, there exists a continuously differentiable function $G$, common to all inside products and invariant to permutations of its inside competitor arguments, such that, for every $j\in\mathcal J_0$,
\begin{equation}\label{eq:InvGeneralDIndexCorrected}
\begin{aligned}
\ln\!\left(\frac{s_{jt}}{s_{0t}}\right)
+G\!\left(
(r_{j0,t},d^p_{j0,t},\mathbf d^{(2)}_{j0,t});
\bigl(r_{jk,t},d^p_{jk,t},\mathbf d^{(2)}_{jk,t}\bigr)_{k\in\mathcal J_0\setminus\{j\}}
\right)=\mathbf x_{jt}\boldsymbol\beta-\alpha p_{jt}+\xi_{jt}.
\end{aligned}
\end{equation}

Moreover, if all inside prices are strictly positive, $G$ admits an equivalent continuously differentiable representation invariant under permutations of all competitors, including the outside good.

Furthermore, for a given market share function $\boldsymbol{\sigma}$ and the normalization $\delta_{0t}=0$, the function $G$ satisfying \eqref{eq:InvGeneralDIndexCorrected} is unique at every admissible value of its arguments.

Lastly, defining $G_{jt}$ as the value of $G$ in \eqref{eq:InvGeneralDIndexCorrected}, the derivative matrix 
\begin{equation*}
    \Bigg[\dfrac{1}{s_{0t}} + \dfrac{1\{i=j\}}{s_{jt}}+\,\dfrac{\partial G_{it}}{\partial s_{jt}}\,\Bigg]_{1\le i,j\le J}
\end{equation*}
is symmetric positive definite on $\Delta_J^+$.\footnote{Derivatives are total derivatives with respect to inside shares on the share simplex, holding prices and characteristics fixed, with $s_{0t}=1-\sum_{i\in\mathcal J_0}s_{it}$.}
\end{proposition}

The decomposition \eqref{eq:InvGeneralDIndexCorrected} forms the basis of the estimation approach: it combines the familiar logit inverse term $\ln(s_{jt}/s_{0t})$ and a permutation-invariant function $G$ of the relative variables $r_{jk,t}$, $d_{jk,t}^{p}$, and $\mathbf{d}_{jk,t}^{(2)}$, including those associated with the outside good. Setting \(G=0\) recovers the logit inverse, which maps shares one-to-one to indices via \(\delta_{jt}=\ln(s_{jt}/s_{0t})\): a larger market share \(s_{jt}\) relative to the outside-good share \(s_{0t}\) corresponds to a higher product index \(\delta_{jt}\). Beyond the logit case, product \(j\)'s market share may also be larger when it faces more distant substitutes in the \((\mathbf{p}_t,\mathbf{x}_t^{(2)})\) space. The function \(G\) captures this mechanism by allowing substitution patterns to vary flexibly with product proximity in that space.

The price and characteristic differences in Proposition~\ref{prop:Main} uniquely determine their levels under the outside-good normalizations. Since $p_{0t}=0$, we recover $p_{jt}=-d^p_{j0,t}$ and $p_{kt}=d^p_{jk,t}-d^p_{j0,t}$, with analogous identities for characteristics. Thus, the inverse market share function and $G$ can be expressed in terms of  share ratios and price and characteristic differences without imposing translation invariance in prices or characteristics: outside good arguments retain dependence on levels. $G$ also incorporates the normalization $\delta_{0t}=0$. If all inside prices are strictly positive, the outside good is uniquely identified by the smallest price difference. This yields an equivalent continuously differentiable representation invariant under permutations of all competitors.

Invertibility is only one implication of these assumptions. The regularity conditions rule out zero market shares and impose symmetry and positive semidefiniteness on the substitution matrix. The invariance restrictions replace the \(J+1\) product-specific functions with a common, permutation-invariant function \(G\) of relative objects. This representation reduces the number of arguments from \((2+K_2)\times(J+1)\) to \((2+K_2)\times J\), and exploiting permutation invariance yields a further reduction in dimensionality, as shown below.

When prices do not enter \(G\) and \(\alpha>0\), symmetry and positive semidefiniteness also ensure consistency with consumer theory. Symmetry guarantees the existence of a surplus function whose derivatives with respect to product indices equal market shares. Positive semidefiniteness makes this surplus function convex, implying the law of demand: each product's market share is nonincreasing in its own price. The model does not rule out complementarity in demand, defined as a negative cross-price derivative of market share.

Supplemental Appendix \ref*{Appendix:Connection} relates this model to the RCL model \citep*{lu2023semi,Nevo2026}, the semiparametric logit model \citep*{compiani2018}, and the universal logit model \citep*{mcfadden1984econometric}.

\paragraph{Outside good exclusion from $G$.} 
Proposition \ref{prop:Main} allows \(G\) to depend on the relative objects associated with every competitor, including the outside good. For identification and
estimation, I impose the stronger restriction that $G$ depends only on the relative objects associated with inside competitors. For each $j\in\mathcal{J}_0$, let $\mathbf V_{jt}^{I}
=
\{
r_{jk,t},d_{jk,t}^{p},\mathbf d_{jk,t}^{(2)}
\}_{k\in\mathcal J_0\setminus\{j\}}$ collect the relative objects for product \(j\)'s inside competitors, and impose
\begin{equation}\label{eq:OutsideExclusion}
G\!\left(r_{j0,t},d_{j0,t}^p,\mathbf d_{j0,t}^{(2)};\mathbf V_{jt}^I\right)
=G^I\!\left(\mathbf V_{jt}^I\right).
\end{equation}
where \(G^{I}\) is continuously differentiable and invariant to permutations of the inside competitors. This exclusion does not follow from the outside-good normalizations or from Proposition \ref{prop:Main}. Hereafter, I write \(G\) and \(\mathbf V_{jt}\) for \(G^{I}\) and
\(\mathbf V_{jt}^{I}\), respectively. 

Restriction \eqref{eq:OutsideExclusion} removes an exact overlap between the linear and nonparametric components. Under the outside good normalizations, $p_{jt}=-d^p_{j0,t}$, $\mathbf x^{(2)}_{jt}=-\mathbf d^{(2)}_{j0,t}$, and $\ln(s_{jt}/s_{0t})=-\ln r_{j0,t}$. An unrestricted $G$ containing these arguments could therefore absorb the dependent variable and the linear effects of price and $\mathbf x^{(2)}$. In particular, the separation condition in Assumption~\ref{Ass3}.\ref{Ass3.3} fails because price itself is a function of these arguments. Retaining the outside good arguments would require additional restrictions addressing both index scale and the overlap between the linear and nonparametric components. The outside good exclusion, together with the identification assumptions below, permits separate identification of $(\alpha,\boldsymbol\beta)$ and $G$.

Economically, \eqref{eq:OutsideExclusion} allows flexible substitution among inside products while constraining substitution with the outside good. Holding covariates fixed, a common increase $c$ in inside indices leaves conditional inside shares unchanged and multiplies participation odds $(1-s_{0t})/s_{0t}$ by $e^c$. The function $G$ captures departures from logit through inside share ratios and signed price and characteristic differences, without direct dependence on outside good arguments. The outside share enters only through $\ln(s_{jt}/s_{0t})$, while price and characteristic levels can still affect demand through the indices and resulting shares. The RCL model does not generally satisfy this exclusion, but simulations show that the model can nevertheless approximate its substitution patterns in the design considered.

\subsection{Nonparametric Identification}

Suppose the modeler observes market-level data \(\{\mathbf{s}_t,\mathbf{p}_t,\mathbf{x}_t,\mathbf{z}_t\}_{t=1}^T\), where \(\mathbf{z}_t\) are instruments for prices and market shares. Under the outside good exclusion  \eqref{eq:OutsideExclusion}, rearranging \eqref{eq:InvGeneralDIndexCorrected} yields a partially linear instrumental variable  model \citep*{ai2003efficient,florens2012instrumental}
\begin{equation}\label{eq:PLIV}
     \ln \left(\frac{s_{jt}}{s_{0t}}\right) = - G(\mathbf{V}_{jt}) + \mathbf{x}_{jt}\boldsymbol{\beta}-\alpha p_{jt}+\xi_{jt},\quad j \in\mathcal{J}_0,
\end{equation}
where both prices and market shares are endogenous.\footnote{Market shares are also endogenous because they depend jointly on prices and unobserved characteristics, and consumers may consider the unobserved characteristics when choosing products.}

The common inside good intercept is not separately identified from $G$, so its allocation is a normalization. I absorb it into $G$ and require that no linear combination of the regressors equals a nonzero constant. When the regressors span a constant, as with the product fixed effects in the empirical application, I instead omit the constant term of the sieve and let the fixed effects absorb the intercept. Either choice fixes the allocation without restricting the model.

Assume the following exclusion, completeness, and separation restrictions on \(\mathbf{z}_t\).\footnote{With endogenous characteristics, \citet*{borusyak2026nonparametric} impose a restriction they call faithfulness.}

\begin{assumption}[Instruments]\label{Ass3}For all $j\in\mathcal{J}_0$, the following restrictions hold:
    \begin{enumerate}
        \item \label{Ass3.1} (Exclusion) $ \mathbb{E}\left[ \xi_{jt}|\mathbf{x}_t,\mathbf{z}_t\right]=0$ almost surely.
        \item \label{Ass3.2} (Completeness for $G$)  For all functions $B(\mathbf{V}_{jt})$ with finite expectation, if $\mathbb{E}[B(\mathbf{V}_{jt})|\mathbf{x}_t,\mathbf{z}_t]=0$ almost surely, then $B(\mathbf{V}_{jt})=0$ almost surely.
        \item \label{Ass3.3} (Separation between $G$ and $(\alpha,\boldsymbol{\beta})$) For all functions $B(\mathbf{V}_{jt})$ with finite expectation and all $(\widetilde{\alpha},\widetilde{\boldsymbol{\beta}})$, if $\mathbb{E}[B(\mathbf{V}_{jt})|\mathbf{x}_t,\mathbf{z}_t]=\mathbf{x}_{jt}\widetilde{\boldsymbol{\beta}}-\widetilde{\alpha}\mathbb{E}[p_{jt}|\mathbf{x}_t,\mathbf{z}_t]$ almost surely, then $\widetilde{\alpha}=0$, $\widetilde{\boldsymbol{\beta}}=\mathbf{0}$, and $B(\mathbf{V}_{jt})=0$ almost surely.
    \end{enumerate}
\end{assumption}

Exclusion requires mean independence between the instruments $\mathbf{z}_t$ and demand shocks $\xi_{jt}$, so any variation the instruments generate in prices and market shares is exogenous. Completeness is the nonparametric counterpart to the rank condition in linear instrumental variable models: it requires the instruments to generate sufficiently rich variation in $\mathbf{V}_{jt}$ to identify $G$ uniquely \citep*{newey2003instrumental,berry2014identification}. Separation is the semiparametric counterpart to the same rank condition. It requires the characteristics \(\mathbf{x}_{jt}\) and predicted prices \(\mathbb{E}[p_{jt}|\mathbf{x}_t,\mathbf{z}_t]\) to vary in ways that no admissible function of \(\mathbf{V}_{jt}\) can reproduce \citep*{florens2012instrumental,newey2003instrumental}. Equivalently, no nonzero linear combination of \(\bigl(\mathbf{x}_{jt},-\mathbb{E}[p_{jt}|\mathbf{x}_t,\mathbf{z}_t]\bigr)\) can be represented as \(\mathbb{E}[B(\mathbf{V}_{jt})|\mathbf{x}_t,\mathbf{z}_t]\) for any admissible \(B\). This prevents \(G\) from absorbing the linear price and characteristics effects, allowing \(\alpha\) and \(\boldsymbol{\beta}\) to be identified separately from \(G\). When $\mathbf x_t^{(1)}$ is absent, equation~\ref{eq:PLIV}, together with exclusion and completeness, implies that all inside good demand shocks are equal within a market. Supplemental Appendix~\ref*{Appendix:PolynomialIdentification} establishes this implication and discusses identification within a fixed-degree polynomial specification. 

Intuitively, instruments help identify $G$ by generating sufficiently rich variation in the relative objects entering $\mathbf V_{jt}$. Separately identifying $(\alpha,\boldsymbol\beta)$ requires variation in predicted prices and characteristics that conditional expectations of functions of these relative objects cannot reproduce.\footnote{Let $\bar p_t=J^{-1}\sum_{j\in\mathcal{J}_0} p_{jt}$ and $\bar{\mathbf x}^{(2)}_t=J^{-1}\sum_{j\in\mathcal{J}_0}\mathbf x^{(2)}_{jt}$. The relative inputs determine $p_{jt}-\bar p_t$ and $\mathbf x^{(2)}_{jt}-\bar{\mathbf x}^{(2)}_t$. Hence, variation in these centered components alone does not generally distinguish their linear effects from $G$. In particular, if a characteristic entering $\mathbf x^{(2)}_{jt}$ is demeaned within every market, its linear coefficient cannot be separated from an unrestricted $G$ without an additional restriction.}
Market-wide cost or markup shifters may help identify $\alpha$ by generating variation in price levels. Examples include exogenous changes in input costs, taxes, or fees, as well as ownership changes that affect equilibrium markups. They need not leave price differences unchanged: a common cost shock can affect prices differently across products. Instruments that reveal substitution patterns by generating sufficiently rich variation in the relative objects $\mathbf V_{jt}$, for example through exogenous changes in the choice set, help identify $G$.\footnote{Candidate price instruments include cost shifters and their proxies, such as prices of the same good in other markets (Hausman instruments). Other candidates shift markups through competitors' characteristics (BLP-type instruments), exogenous changes in market structures such as ownership changes, or demographics in related markets (Waldfogel-Fan instruments), subject to the exclusion restriction \citep[see][Section 4.2]{berryhaile2021}.}
The same instruments can help identify both components if they generate both types of variation.

The following proposition establishes nonparametric identification of the market share function under Assumptions \ref{Ass1} -- \ref{Ass3} and the outside good exclusion \eqref{eq:OutsideExclusion}.

\begin{proposition}[Nonparametric demand identification]\label{PropBH} Under Assumptions \ref{Ass1} -- \ref{Ass3} and the outside good exclusion \eqref{eq:OutsideExclusion}, 
$(\alpha,\boldsymbol{\beta})$ are point identified, and $G$ is point 
identified on the support of $\mathbf{V}_{jt}$. It follows that the 
inverse market share function $\boldsymbol{\sigma}^{-1}$ is point 
identified on the support of 
$(\mathbf{s}_t,\mathbf{p}_t,\mathbf{x}_t^{(2)})$ and, by Proposition 
\ref{prop:Main}, that the market share function $\boldsymbol{\sigma}$ 
is point identified on $\mathcal{X}$.
\end{proposition}

I now turn from identification to estimation. I jointly estimate \((\alpha,\boldsymbol{\beta})\) and \(G\), approximating \(G\) with a polynomial sieve whose degree grows with the sample size. 

\subsection{Symmetric (Permutation-Invariant) Sieve Estimation \label{Subsec:EmpiricalAppli}}

\paragraph{Method of sieves.} Fix $J\geq 2$. Collect the objects relative to product $j$ in $\mathbf{V}_{jt}\in\mathbb{R}^{\widetilde v}$, where $\widetilde v=(2+K_2) \times (J-1)$, and normalize each component to lie in $[-1,1]^{\widetilde v}$. Let $\widetilde{\mathbf V}_{jt}=\mathbf h(\mathbf V_{jt})$ denote this normalization.\footnote{In practice, the market share ratios $r_{jk,t}$ and the difference terms
$d\in\{d_{jk,t}^{(2,l)},d_{jk,t}^{p}\}$ need not lie in $[-1,1]$ and
therefore must be transformed before constructing the Chebyshev basis. For the ratio $r_{jk,t}=s_{kt}/s_{jt}$, I use the relative difference
transformation $\widetilde r_{jk,t}
=
\frac{r_{jk,t}-1}{r_{jk,t}+1}
=
\frac{s_{kt}-s_{jt}}{s_{kt}+s_{jt}}$, which maps $(0,\infty)$ onto $(-1,1)$. For each difference term $d$, I choose the transformation according to its support. If $d$ has bounded support $[a,b]$, I use the linear rescaling 
$\widetilde d
=
2\frac{d-a}{b-a}-1$, which maps $[a,b]$ onto $[-1,1]$. If $d$ has unbounded support, I instead use the algebraic transformation $\widetilde d = \frac{d}{\sqrt{1+d^2}}$,
which maps $\mathbb{R}$ onto $(-1,1)$, following
\citet*{shen2014approximations}.} The map $\mathbf h$ uses fixed, strictly increasing transformations that may differ across coordinate types but are identical across competitors for each corresponding coordinate. Define $\widetilde G(\widetilde{\mathbf v})=G(\mathbf h^{-1}(\widetilde{\mathbf v}))$, so that $\widetilde G(\widetilde{\mathbf V}_{jt})=G(\mathbf V_{jt})$. For uniform approximation, consider a fixed compact domain $\widetilde{\mathcal V}$ of admissible normalized inputs that is unchanged by relabeling inside competitors.

For every $(j,t)$, the model and the exclusion restrictions imply 
\begin{equation}\label{eq:xistruct}
\mathbb{E}[\xi_{jt}|\mathbf{x}_t,\mathbf{z}_t]=0, \quad \xi_{jt}=\ln \left(\frac{s_{jt}}{s_{0t}}\right) + G(\mathbf{V}_{jt}) - \mathbf{x}_{jt}\boldsymbol{\beta}+\alpha p_{jt}.
\end{equation}

As usual, projecting the conditional moments onto a finite set of instrument functions yields unconditional moments. Since \(G(\mathbf{V}_{jt})=\widetilde{G}(\widetilde{\mathbf{V}}_{jt})\), the sieve approximates \(\widetilde{G}\) on \(\widetilde{\mathcal V}\).

I estimate \(\widetilde{G}\) nonparametrically by a sieve whose dimension grows with \(T\). First allow equation-specific approximations. The sieve implementation below imposes a common basis and coefficient vector. For each equation $j\in\mathcal{J}_0$, let $\boldsymbol{\phi}_{P_j}^{(j)}:[-1,1]^{\widetilde v}\to\mathbb{R}^{P_j}$, defined by $\boldsymbol{\phi}_{P_j}^{(j)}(\widetilde{\mathbf V}_{jt})=(\phi_{1,P_j}^{(j)}(\widetilde{\mathbf V}_{jt}),\ldots,\phi_{P_j,P_j}^{(j)}(\widetilde{\mathbf V}_{jt}))'$, denote a basis vector, for example a tensor Bernstein or Chebyshev polynomial basis. The linear sieve approximation to \(\widetilde{G}\) is given by
\begin{equation*}\label{eq:sieve}
\widetilde{G}(\widetilde{\mathbf{V}}_{jt};\boldsymbol{\theta}_j)
= (\boldsymbol{\phi}_{P_j}^{(j)}(\widetilde{\mathbf{V}}_{jt}))^\top\boldsymbol{\theta}_j,
\qquad \boldsymbol{\theta}_j \in \mathbb{R}^{P_j},
\end{equation*}
where \(P_j=P_j(T)\to\infty\) as \(T\to\infty\) with $J$ fixed. 

Let $\boldsymbol{\psi}_{K_j}^{(j)}: (\mathbf{x}_t,\mathbf{z}_t)\mapsto(\boldsymbol{\psi}_{1,K_j}^{(j)}(\mathbf{x}_t,\mathbf{z}_t),\ldots,\boldsymbol{\psi}_{K_j,K_j}^{(j)}(\mathbf{x}_t,\mathbf{z}_t))'\in \mathbb{R}^{K_j}$ denote the instrument basis for equation $j$, with $K_j = K_j(T) \rightarrow \infty$ as $T \rightarrow \infty$. Assume that $\mathbb E\!\big[\boldsymbol{\psi}_{K_j}^{(j)}(\mathbf{x}_t,\mathbf{z}_t)\,\boldsymbol{\phi}_{P_j}^{(j)}(\widetilde{\mathbf{V}}_{jt})^{\!\top}\big]\in\mathbb{R}^{K_j \times P_j}$ has full column rank $P_j$. This rank condition implies the order condition $K_j\geq P_j$. For joint identification of $(\alpha,\boldsymbol\beta,\boldsymbol\theta_D)$ in the correctly specified common sieve, impose full column rank for the moments of all regressors, as in Supplemental Appendix~\ref*{Appendix:ProofUniqueExist}. Stacking over markets gives the instrument matrix $\boldsymbol{\Psi}_{(j)}=(\psi_{K_j}^{(j)}(\mathbf{x}_1,\mathbf{z}_1),\ldots,\psi_{K_j}^{(j)}(\mathbf{x}_T,\mathbf{z}_T))' \in \mathbb{R}^{T \times K_j}$.

Define the moment vector for observation $(j,t)$ by
 \begin{equation*}
    g_{jt}(\alpha,\boldsymbol{\beta},\boldsymbol{\theta}_j)=\xi_{jt}(\alpha,\boldsymbol\beta,\boldsymbol\theta_j)\boldsymbol{\psi}_{K_j}^{(j)}(\mathbf{x}_t,\mathbf{z}_t)=(\ln (s_{jt}/s_{0t}) + \widetilde{G}(\widetilde{\mathbf{V}}_{jt};\boldsymbol{\theta}_j) - \mathbf{x}_{jt}\boldsymbol{\beta}+\alpha p_{jt})\boldsymbol{\psi}_{K_j}^{(j)}(\mathbf{x}_t,\mathbf{z}_t),
\end{equation*}
and stack the per-equation sample moments into the vector
\begin{equation*}
    \mathbf{g}_T(\boldsymbol{\theta}) = \left(\mathbf{g}_T^{(1)}(\boldsymbol{\theta})',\ldots,\mathbf{g}_T^{(J)}(\boldsymbol{\theta})'\right)'\in\mathbb{R}^{\sum_j K_j},
    \qquad
    \mathbf{g}_T^{(j)}(\boldsymbol{\theta})=\dfrac{1}{JT}\sum_{t=1}^T g_{jt}(\boldsymbol{\theta}),
\end{equation*}
where $\boldsymbol{\theta}=(\alpha,\boldsymbol{\beta},\boldsymbol{\theta}_1,\ldots,\boldsymbol{\theta}_J)$ stacks all parameters.\footnote{Here, $\xi_{jt}(\boldsymbol\theta)$ denotes the sieve residual. At the true linear coefficients, it coincides with the structural shock only if the sieve represents $G$ exactly on the relevant support.} When all equations use the same instrument functions, $\boldsymbol{\psi}^{(j)}_{K_j}$ may be replaced by observation-specific instruments $\mathbf z_{jt}\in\mathbb R^K$ (e.g., own-characteristic and competitor-based functions of $(\mathbf x_t,\mathbf z_t)$), and one may instead use the pooled moment $\mathbf g_T(\boldsymbol\theta)=\frac{1}{JT}\sum_{j,t}\xi_{jt}(\boldsymbol\theta)\mathbf z_{jt}$.

The sieve-GMM estimator $\widehat{\boldsymbol{\theta}}$ solves
\begin{equation*}\label{eq:MinQ}
    \min_{\boldsymbol{\theta}}   \mathbf{g}_T(\boldsymbol{\theta})'\mathbf{W}_T \mathbf{g}_T(\boldsymbol{\theta}),
\end{equation*}
where the weighting matrix is typically the 2SLS weighting matrix $\mathbf{W}_T= (\boldsymbol{\Psi}'\boldsymbol{\Psi}/(JT))^{-1}$ with $\boldsymbol{\Psi}=\mathrm{diag}(\boldsymbol{\Psi}_{(1)},\ldots,\boldsymbol{\Psi}_{(J)})$. Here, $\boldsymbol\Psi$ has full column rank and is block diagonal for stacked moments; for pooled moments, its rows are $\mathbf z_{jt}'$.

Sieve estimation presents two challenges. The first is ill-posedness. Projecting onto the bases $\phi^{(j)}_{P_j}(\widetilde{\mathbf{V}}_{jt})$ and $\psi^{(j)}_{K_j}(\mathbf{x},\mathbf{z})$ induces the finite-dimensional map
\begin{equation*}
T_{K_j,P_j}:\ \theta_j\in\mathbb{R}^{P_j}\ \longmapsto\ 
\mathbb{E}\big[\boldsymbol{\psi}^{(j)}_{K_j}(\mathbf{x}_t,\mathbf{z}_t)\,
\boldsymbol{\phi}^{(j)}_{P_j}(\widetilde{\mathbf{V}}_{jt})^{\!\top}\big]\theta_j
= \mathbf{M}_{K_j,P_j}\theta_j.
\end{equation*}

Recovering $\theta_j$ requires inverting this map. As $P_j$ and $K_j$ grow, the smallest singular value of $\mathbf{M}_{K_j,P_j}$ may approach zero. Small sampling errors in $\mathbf{M}_{K_j,P_j}\theta_j$ can then produce large errors in the sieve coefficients and recovered inverse market share function. The literature stabilizes this inversion by restricting the admissible function space through (i) sieve-dimension control or hard thresholding \citep*{horowitz2011applied,chen2012estimation}, (ii) Tikhonov (ridge) regularization \citep*{newey2003instrumental,chen2012estimation}, and (iii) theory-based shape restrictions such as monotonicity or convexity \citep*{matzkin1994,chetverikov2018,blundell2012measuring,chetverikov2017nonparametric}. 

The second challenge is the curse of dimensionality: in a standard tensor-product sieve, the number of basis terms grows exponentially in the dimension of $\widetilde{\mathbf{V}}_{jt}$. To see this, note that each coordinate of $\widetilde{\mathbf V}_{jt}$ lies in $[-1,1]$, and fix a univariate basis $\{\phi_k^{(\ell)}\}_{k=0}^D$ on $[-1,1]$ for each $\ell=1,\ldots,\widetilde v$. For multi-indices $\mathbf k=(k_1,\ldots,k_{\tilde v})\in\{0,\ldots,D\}^{\tilde v}$, the usual tensor-product approximation takes the form 
\[
\widetilde{G}(\widetilde{\mathbf{V}}_{jt};\theta_j)\;=\;\sum_{\mathbf k\in\{0,\ldots,D\}^{\tilde v}}\theta_{j,\mathbf k}\,\Phi_{\mathbf k}(\widetilde{\mathbf{V}}_{jt}),\quad \Phi_{\mathbf k}(\widetilde{\mathbf{V}}_{jt})=\prod_{\ell=1}^{\tilde v}\phi^{(\ell)}_{k_\ell}\!\big(\widetilde{V}_{jt,\ell}\big),
\]
and therefore requires $P_j=(D+1)^{\tilde v}$ coefficients $\{\theta_{j,\mathbf{k}}\}$. 

Because $\widetilde v=(2+K_2) \times (J-1)$, the number of sieve coefficients increases sharply with the number of products. Ill-posedness magnifies this dimensionality burden, leading to poor finite-sample performance and limited practical feasibility even when the dimension is moderate. For example, when $K_2=0$ and $D=2$, a tensor-product sieve already requires $4.78\times 10^6$ parameters with only $J=8$ products. Prior work addresses this curse through dimension-reduction strategies such as additive structures \citep*[][Section 5.1]{ichimura2007implementing}, permutation-invariant specifications \citep*{compiani2018}, and machine learning methods \citep[see][and references therein]{Bakhitov2025}.

I address ill-posedness through sieve-dimension control (i.e., selecting the polynomial degree $D$ through cross-validation) and theory-based shape restrictions, yielding a constrained sieve-GMM estimator.  I next address the curse of dimensionality.

\paragraph{Symmetric sieve construction.} Building on \citet*{bachmayr2024polynomial}, I approximate \(\widetilde{G}\) with a symmetric polynomial sieve that pools competitor-specific basis terms \(\phi_{\mathbf r}(\widetilde{\mathbf v}_{jk,t})\). Applying a common transformation to each competitor and aggregating the terms imposes permutation invariance by construction. Once \(J-1\geq D\), the number of sieve coefficients depends on the polynomial degree \(D\), which governs approximation accuracy, rather than on the number of products \(J\). For a fixed degree, this breaks the curse of dimensionality in products without sacrificing approximation power for permutation-invariant functions.

Fix a product \(j\), and let \(\widetilde{\mathbf v}_k\) denote the block of arguments in \(\widetilde{\mathbf V}_{jt}\) corresponding to inside competitor \(k\). Each block lies in \(\Omega=[-1,1]^d\), where \(d=2+K_2\), and product \(j\) has \(N=J-1\) inside competitors. The approximation target is the permutation-invariant function $\widetilde G:\widetilde{\mathcal{V}}\to\mathbb R$ of $(\widetilde{\mathbf v}_1,\ldots,\widetilde{\mathbf v}_N)$, where $\widetilde{\mathcal{V}}\subseteq\Omega^N=[-1,1]^{\widetilde v}$ is the compact domain of admissible normalized inputs, which is unchanged by relabeling inside competitors, and $\widetilde v=dN$.

Choose a one-body tensor-product basis $\Phi=\{\phi_{\mathbf r}:\mathbf r\in\mathbb N^d\}$ on $\Omega$. Here $\mathbb{N}$ includes zero. I use the Chebyshev polynomial basis of the first kind, for which the degree-$n$ polynomial is defined by $T_n(v)=\cos\!\big(n\arccos v\big)$ for $|v|\leq 1$. For a multi-index $\mathbf{r}=(r_1,\ldots,r_d)\in\mathbb N^d$ and an input $\widetilde{\mathbf v}=(\widetilde v_1,\ldots,\widetilde v_d)\in\Omega$, the associated tensor-product basis element and its total degree are 
\begin{equation*}
\phi_\mathbf{r}(\widetilde{\mathbf v})
=
\prod_{q=1}^d T_{r_q}(\widetilde v_q),
\qquad
\deg(\phi_\mathbf{r})=\sum_{q=1}^d r_q.
\end{equation*}

For an input $\widetilde{\mathbf{V}}=(\widetilde{\mathbf{v}}_1,\ldots,\widetilde{\mathbf{v}}_N)\in\Omega^N$, define the pooled one-body feature by
\begin{equation}   \label{eq:pooledfeature}
A_\mathbf{r}(\widetilde{\mathbf{V}})\ =\ \dfrac{1}{N} \sum_{n=1}^{N}\phi_\mathbf{r}(\widetilde{\mathbf{v}}_n), 
\end{equation}
and for a multi-index $\mathbf{u}=(\mathbf{u}_1,\ldots,\mathbf{u}_N)$, define the many-body symmetric feature by
\begin{equation}\label{eq:manypooledfeature}
    A_\mathbf{u}(\widetilde{\mathbf{V}})\ =\ \prod_{m=1}^{N} A_{\mathbf{u}_m}(\widetilde{\mathbf{V}}).
\end{equation}

Equation \eqref{eq:pooledfeature} implements the Deep Sets idea of \citet*{zaheer2017}, under which permutation-invariant functions are represented through pooled one-body features.\footnote{I pool by averaging, whereas \citet*{bachmayr2024polynomial} sum. Averaging rescales every basis function by a constant, so the span is unchanged and their approximation result continues to apply. It sets the constant feature to one rather than $N$, so factors with a zero multi-index leave \eqref{eq:manypooledfeature} unchanged and features of different body order share a common scale, which improves conditioning.} Their products in \eqref{eq:manypooledfeature} add higher-order symmetric interactions.

Let $\preceq_d$ denote the graded lexicographic order on $\mathbb N^d$: multi-indices are ordered first by total degree and, within each degree class, lexicographically. Restricting tuples of one-body multi-indices to a fixed nondecreasing order eliminates duplicate many-body basis functions. Specifically, restrict $\mathbf{u}=(\mathbf{u}_1,\ldots,\mathbf{u}_N)$ to $(\mathbb N^d)_{\mathrm{ord}}^N=\{\mathbf{u}\in(\mathbb N^d)^N:\mathbf{u}_1\preceq_d\cdots\preceq_d\mathbf{u}_N\}$ and define its total degree by $\deg(\mathbf{u})=\sum_{m=1}^N\deg(\phi_{\mathbf{u}_m})$. The degree-$D$ sieve then uses only the basis functions $\{A_{\mathbf{u}}:\mathbf{u}\in(\mathbb N^d)_{\mathrm{ord}}^N,\ \deg(\mathbf{u})\leq D\}$.

Let the degree-$D$ symmetric Chebyshev space be
\begin{equation*}
\mathcal{A}_D
=
\operatorname{span}
\left\{
A_{\mathbf{u}}:
\mathbf{u}\in(\mathbb N^d)_{\mathrm{ord}}^N,\ 
\deg(\mathbf{u})\leq D
\right\}.
\end{equation*}

The approximant is
\begin{equation*}
\widetilde G_D(\widetilde{\mathbf V};\boldsymbol\theta_D)
=
\sum_{\substack{
\mathbf u\in(\mathbb N^d)_{\mathrm{ord}}^N\\
\deg(\mathbf u)\leq D
}}
\theta_{\mathbf u}\,A_{\mathbf u}(\widetilde{\mathbf V}),
\qquad
\boldsymbol\theta_D=(\theta_{\mathbf u}).
\end{equation*}

Under the symmetric sieve construction, the inverse market share equations become
\begin{equation}\label{eq:xi}
    \xi_{jt}(\boldsymbol{\theta}) = \ln \left(\frac{s_{jt}}{s_{0t}}\right)+ \widetilde{G}_D(\widetilde{\mathbf{V}}_{jt};\boldsymbol{\theta}_D) - \mathbf{x}_{jt}\boldsymbol{\beta}+\alpha p_{jt},\quad j \in \mathcal{J}_0,
\end{equation}
where $\boldsymbol{\theta}=(\alpha,\boldsymbol{\beta},\boldsymbol{\theta}_D)'$.

The following uniform approximation result of \citet*{bachmayr2024polynomial} provides the theoretical foundation for the symmetric polynomial sieve.

\begin{proposition}[Approximation]\label{Prop:Approx}
Let $\widetilde{G}:\widetilde{\mathcal{V}}\to\mathbb R$ be a continuous, permutation-invariant function. Then, for every $\varepsilon>0$, there exist $D \in \mathbb{N}$ and $\widetilde{G}_D\in\mathcal A_D$ such that
\[
\sup_{\widetilde{\mathbf{V}}\in\widetilde{\mathcal{V}}}\big|\widetilde{G}(\widetilde{\mathbf{V}})- \widetilde{G}_D(\widetilde{\mathbf{V}})\big| < \varepsilon.
\]
\end{proposition}

Thus, pooled one-body features and their products uniformly approximate every continuous, permutation-invariant function on $\Omega^N$. 

Define the uniform best-approximation error as
\begin{equation*}
    e_D(\widetilde{G}) = \inf_{g \in \mathcal{A}_D}\sup_{\widetilde{\mathbf{V}}\in\widetilde{\mathcal{V}}}|\widetilde{G}(\widetilde{\mathbf{V}})-g(\widetilde{\mathbf{V}})|.
\end{equation*}

Because the sieve spaces are nested, Proposition \ref{Prop:Approx} implies that  $e_D(\widetilde{G}) \rightarrow 0$ as  $D\rightarrow \infty$.\footnote{Nestedness makes $e_D(\widetilde G)$ nonincreasing in $D$. Proposition~\ref{Prop:Approx} implies that, for every $\varepsilon>0$, there exists $D_\varepsilon$ such that $e_{D_\varepsilon}(\widetilde G)<\varepsilon$. Hence, $e_D(\widetilde G)\leq e_{D_\varepsilon}(\widetilde G)<\varepsilon$ for every $D\geq D_\varepsilon$, proving the result.}  This establishes the approximation condition for unconstrained nonparametric estimation using the symmetric sieve space $\mathcal A_{D_T}$, whose degree $D_T$ increases with the sample size.\footnote{Supplemental Appendix~\ref*{Appendix:SupNormConsistency} states the additional conditions for sup-norm consistency, following \citet[Section 3.1]{Chen2007}.} See Supplemental Appendix \ref*{Appendix:BodyApprox} for another useful approximation result.

The key scalability result is that, once $J-1\geq D$, the number of sieve coefficients $P=P(J,D,d)$ depends only on $(D,d)$ and no longer on $J$. A term of total degree at most $D$ involves at most $D$ nonconstant pooled factors. Additional factor slots contribute only constants, giving alternative representations of existing many-body symmetric features. Increasing $J-1$ beyond $D$ therefore leaves the basis unchanged.

This result delivers scalable nonparametric demand estimation: the sieve approximates any continuous, permutation-invariant component arbitrarily well, while its coefficient count is independent of $J$ once $J -1 \geq D$. I select $D$ by cross-validation. The gain over an unrestricted total-degree-$D$ polynomial sieve is parsimony, not approximation power: for a permutation-invariant target, a best total-degree-$D$ polynomial approximation on $\widetilde{\mathcal{V}}$ can itself be chosen permutation invariant. The symmetric sieve achieves the same best approximation error as that sieve with fewer coefficients, potentially lowering variance and easing the bias-variance trade-off. Table \ref{tabdim1} reports coefficient counts relative to sieves of degree at most $D$ in each coordinate.

\subsection{Implementation \label{Sec:Implementation}}

\paragraph{Constrained GMM.} For a fixed degree $D$, estimation reduces to the estimation of the finite-dimensional parameter vector $\theta = (\alpha,\boldsymbol{\beta},\boldsymbol{\theta}_D)$, where $\boldsymbol{\theta}_D$ collects the coefficients of the symmetric polynomial sieve for $G$. Let $\mathbf{g}_T(\theta)$ denote the sample analogue of the unconditional moment conditions, with residual $\xi_{jt}(\theta)$ defined by \eqref{eq:xi}. Let $\Theta$ denote the parameter space induced by the theory-based shape restrictions described below, and let $\mathbf{W}_T$ denote a positive definite GMM weighting matrix.  The sieve-GMM estimator solves\footnote{In the Monte Carlo simulations, I estimate the models by constrained one-step GMM. I first orthonormalize the instruments using a QR decomposition, replacing the instrument matrix with an orthonormal basis for its column space. This transformation improves numerical conditioning without changing the instrument span. I form the GMM criterion using the transformed instruments. In the empirical application, I estimate the model by constrained two-step GMM: I form the first-step criterion using the transformed instruments and use the first-step residuals to construct a market-clustered second-step weighting matrix. When the second-step weighting matrix is ill-conditioned, I regularize it using the linear Ledoit-Wolf shrinkage estimator before inversion \citep{ledoit2004well}. I solve all GMM problems in MATLAB using MOSEK \citep*{mosek}.}
\[
\hat{\boldsymbol{\theta}} = \arg\min_{\boldsymbol{\theta} \in \Theta}  \mathbf{g}_T(\boldsymbol{\theta})' \mathbf{W}_T \mathbf{g}_T(\boldsymbol{\theta}).
\]

Because the sample moments are affine in $\theta$ and the restrictions below define a convex admissible set $\Theta$, the estimator solves a convex program. Supplemental Appendix~\ref*{Appendix:ConstrainedGMMUniqueness} gives conditions for existence and uniqueness of the constrained GMM solution. I select the sieve degree by minimizing a cross-validated GMM criterion, following \citet*{rivers2002model,komiyama2018cross}. Implementation details are provided in Supplemental Appendix \ref*{Appendix:Appli}.

\paragraph{Demand derivatives.} For each product $j$, define
\[
\widehat G_{jt}
=
\widetilde G_D
\bigl(
\widetilde{\mathbf V}_{jt};
\widehat{\boldsymbol{\theta}}_D
\bigr),
\qquad
\Gamma^{s,p}_{jk,t}
=
\frac{\partial\widehat G_{jt}}{\partial s_{kt}},
\qquad
\Gamma^p_{jk,t}
=
\frac{\partial\widehat G_{jt}}{\partial d^p_{jk,t}},
\quad j\neq k.
\]

The derivative matrices of the estimated inverse market share  with respect to inside good shares and prices are
\[
\mathbf A_t
=
\frac{1}{s_{0t}}\mathbf 1\mathbf 1'
+
\operatorname{diag}(\mathbf s_t)^{-1}
+
\left[
\Gamma^{s,p}_{jk,t}
\right]_{j,k=1}^{J},
\]
and
\[
[\mathbf B_t]_{jk}
=
\begin{cases}
\displaystyle
\widehat{\alpha}
-
\sum_{\ell\neq j}
\Gamma^p_{j\ell,t},
& j=k,\\[0.8em]
\displaystyle
\Gamma^p_{jk,t},
& j\neq k.
\end{cases}
\]

Let $\mathbf C_t
=
\operatorname{diag}(\mathbf s_t)^{-1}
+
\left[
\Gamma^{s,p}_{jk,t}
\right]_{j,k=1}^{J}$ so that $\mathbf A_t=\mathbf C_t+s_{0t}^{-1}\mathbf 1\mathbf 1'$, where $\mathbf 1$ denotes the $J$-dimensional vector of ones.
At a candidate coefficient vector, I write $\mathbf A_t(\boldsymbol{\theta}_D)$,
$\mathbf B_t(\alpha,\boldsymbol{\theta}_D)$ and $\mathbf C_t(\boldsymbol{\theta}_D)$.

Provided $\mathbf A_t$ is nonsingular, the matrix of demand derivatives with respect to prices is\footnote{Supplemental Appendix~\ref*{app:shape_restrictions} provides the exact constraint formulas and the restrictions used to ensure that $\mathbf A_t$ is nonsingular.}
\[
\mathbf{J}_{\boldsymbol{\sigma},t}^{\mathbf{p}}
=
-\mathbf A_t^{-1}\mathbf B_t,
\]
from which we can compute own- and cross-price elasticities, diversion ratios, and demand curvature. For products $j$ and $k$ in market $t$, writing $\partial\sigma_{kt}/\partial p_{jt}$ for the $(k,j)$ entry of $\widehat{\mathbf{J}}_{\boldsymbol{\sigma},t}^{\mathbf{p}}$, the price elasticity of product $k$ with respect to the price of product $j$ and the diversion ratio from product $j$ to product $k\neq j$ are
\[
\varepsilon_{kj,t}
=\frac{p_{jt}}{\sigma_{kt}}\frac{\partial \sigma_{kt}}{\partial p_{jt}},
\qquad
\mathcal D_{jk,t}
=-\frac{\partial\sigma_{kt}/\partial p_{jt}}{\partial\sigma_{jt}/\partial p_{jt}},
\]
and the diversion ratio to the outside good is $\mathcal D_{j0,t}=1-\sum_{k\neq j}\mathcal D_{jk,t}$. The own-price demand curvature of product $j$ is
\[
\rho_{jt}
=\frac{\sigma_{jt}\,\partial^2\sigma_{jt}/\partial p_{jt}^2}
      {\bigl(\partial\sigma_{jt}/\partial p_{jt}\bigr)^2},
\]
holding other prices and demand shifters fixed \citep*{mrazova2017,miravete2023pass,miravete2023elasticity}. Curvature requires $\boldsymbol\sigma$ to be twice differentiable in prices, which Assumption~\ref{AssReg} does not impose, and depends on second derivatives of $\widehat G$.

\paragraph{Shape restrictions.} I use theory-based shape restrictions to improve finite-sample stability and discipline substitution patterns. Specifically, I require $\mathbf C_t(\boldsymbol{\theta}_D)$ to be symmetric and strictly row-diagonally dominant, which is sufficient for positive definiteness of $\mathbf C_t(\boldsymbol{\theta}_D)$ and hence of $\mathbf A_t(\boldsymbol{\theta}_D)$. I impose dominance on $\mathbf C_t$ rather than on $\mathbf A_t$ because the rank-one term $s_{0t}^{-1}\mathbf 1\mathbf 1'$ can prevent diagonal dominance without preventing positive definiteness. When prices do not enter \(G\), \(\widehat{\alpha}>0\), and this derivative matrix is positive definite, the restrictions imply local consistency with consumer theory at the observed data points. In particular, they deliver the law of demand and a local surplus function. Auxiliary variables for absolute-value terms make all restrictions linear in these variables and the sieve coefficients. I impose these restrictions at the observed data points in each market.

Following \citet*{compiani2018}, one can impose one restriction under substitutability, defined as $\partial\sigma_k(\boldsymbol{\delta}_t,\mathbf{p}_t,\mathbf{x}^{(2)}_t)/\partial\delta_{jt}\leq0$ for every $k\neq j$, and two further
restrictions under additional conditions. Under substitutability,
the inside good block of the substitution matrix is an $M$-matrix, so its inverse
$\mathbf A_t(\boldsymbol{\theta}_D)$ has nonnegative entries, which adds $TJ^2$ inequalities. If
the own effect of each demand index also dominates the sum of its cross effects, each diagonal
entry of $\mathbf A_t(\boldsymbol{\theta}_D)$ is at least as large as the off-diagonal entries in
its row, which adds $TJ(J-1)$ inequalities. If demand in addition is decreasing in own price, that is, the diagonal entries of $\mathbf{J}_{\boldsymbol{\sigma}}^{\mathbf{p}}$ are negative, and each inverse share function is weakly increasing in every competitor's price, the diagonal entries of
$\mathbf B_t(\alpha,\boldsymbol{\theta}_D)$ are nonnegative, which adds $TJ$ inequalities. All three are linear in the sieve coefficients and require no auxiliary variables.

Direct implementation generates $TJ(J-1)/2$ symmetry equalities, $TJ(J-1)$ auxiliary variables, and $TJ(2J-1)$ diagonal-dominance inequalities, for a total of $TJ(5J-3)/2$ constraint rows. Because the symmetry equalities restrict the same $K_D$ sieve coefficients, their rank cannot exceed $\min\{K_D,TJ(J-1)/2\}$. A low-memory QR construction retains $r_{\mathrm{sym}}\leq \min\{K_D,TJ(J-1)/2\}$ orthonormal equality restrictions, reducing memory use and improving conditioning while preserving the feasible coefficient set. This reduces the total number of local constraint rows from $TJ(5J-3)/2$ to $r_{\mathrm{sym}}+TJ(2J-1)$, eliminating $TJ(J-1)/2-r_{\mathrm{sym}}$ redundant equality rows. Supplemental Appendix~\ref*{app:shape_restrictions} provides the exact constraint formulas. 

The restrictions above hold only at the observed market states. I also derive a coefficient-only sufficient condition making the symmetric part of $\mathbf A_t(\boldsymbol{\theta}_D)$ uniformly positive definite over a prespecified share and price domain. This condition is conservative, but the number of constraints it requires does not depend on the number of markets. If $K_s$ sieve coefficients depend on market shares, combining the global restriction with the low-memory symmetry representation requires $r_{\mathrm{sym}}+2K_s+1$ constraint rows and $K_s$ auxiliary variables, rather
than $r_{\mathrm{sym}}+TJ(2J-1)$ rows and $TJ(J-1)$ auxiliary variables under the local implementation. The global restriction therefore removes $TJ(2J-1)-(2K_s+1)$ constraint rows and $TJ(J-1)-K_s$ auxiliary variables. Because $r_{\mathrm{sym}}\leq K_D$, the total number of constraint rows is bounded by $K_D+2K_s+1$, independently of $T$ for a fixed sieve dimension.

\section{Monte Carlo Simulations  \label{Section:Simul}} 

In this section, I use Monte Carlo simulations to evaluate the empirical performance of the scalable nonparametric demand (SNPD for short) approach. I conduct three experiments. The first uses a correctly specified BLP benchmark and asks whether SNPD can recover the resulting substitution patterns, demand curvature, and markups. The remaining two experiments introduce BLP misspecification: misspecification of the random coefficient distribution and misspecification of the logit-type functional form. These two experiments ask whether SNPD can improve on BLP when BLP is misspecified. In all three experiments, the supply side is a static oligopolistic price competition model. 

For each experiment, I generate 100 Monte Carlo replications. Each replication contains \(T\in\{100,500\}\) independent markets with \(J\in\{25,50\}\) products. Products are owned by $J/5$ multiproduct firms, each producing five products. For each replication, I simulate data from a model of supply and demand by generating product characteristics, cost shifters, and unobserved demand and cost shocks, and then solving for equilibrium prices and market shares. The observed demand characteristic \(x_{jt}^{(2)}\) and cost shifters \(w_{jt}^1,w^2_{jt}\), which enter only marginal costs, are independently drawn from \(\mathcal{U}(0,1)\). The demand unobservable is \(\xi_{jt}=u_{1jt}+u_{2jt}\), and the cost unobservable is \(\omega_{jt}=u_{1jt}+u_{3jt}\), where \(u_{1jt}\), \(u_{2jt}\), and \(u_{3jt}\) are independently drawn from \(\mathcal{U}(-0.25,0.25)\). The supply side is a static oligopoly price competition model with marginal cost $c_{jt}=2+x_{jt}^{(2)}+w_{jt}^1+ w^2_{jt}+\omega_{jt}$. The demand model varies across experiments, as described below. Its parameters are calibrated so that the outside good share ranges from 0.8 to 0.9, consistent with \citet*{conlon2018}. Prices \(p_{jt}\) and market shares \(s_{jt}\) are determined by solving for the model equilibrium.\footnote{In each market \(t\), firm \(f\) chooses the prices of its products
\(\mathcal{J}_f\) to maximize profits $\sum_{j\in\mathcal{J}_f}(p_{jt}-c_{jt})\sigma_j(\mathcal{X}_t)$. Assuming a pure-strategy Nash equilibrium in prices, I use the implied first-order conditions to recover marginal costs \(\widehat c_{jt}\) and markups \(100(p_{jt}-\widehat c_{jt})/p_{jt}\).} 

I report three sets of statistics. First, I summarize substitution patterns using average own-price elasticities, cross-price elasticities, and diversion ratios to competing products and the outside good. For cross-product substitution, I report the average across all competing products and separate averages for the five closest products and all remaining products in price-characteristic space. For diversion ratios, I report both the average diversion to an individual competing product and the total diversion to each group of products. Second, I report average own-price demand curvature. Third, I report average implied markups. I weight diversion ratios by the market share of the source product and markups by own market shares; elasticities and curvature are unweighted.

I implement the SNPD approach following the procedure described in Subsection \ref{Sec:Implementation}. I simulate the RCL models using the PyBLP package developed by \citet*{conlon2018}, integrating the market share functions using a level-9 Gauss-Hermite product rule. I estimate the RCL models using PyBLP's implementation of the nested fixed-point algorithm of \citet*{BLP1995}, following the best practices recommended by \citet*{conlon2018}.\footnote{I integrate the market share functions using a level-9 Gauss-Hermite product rule, recover mean utilities with SQUAREM using a  tolerance of \(10^{-14}\), and minimize the GMM objective with Knitro 13.1's Interior/Direct algorithm. Estimation proceeds via two-step GMM with heteroskedasticity-robust weighting matrices. The initial instrument set includes cost shifters and differentiation instruments proposed by \citet*{gandhihoude}, based on sums of squared differences in $x^{(2)}_{jt}$ over products owned by the same and rival firms, respectively. Following \citet*{gandhihoude}, I construct analogous instruments using prices predicted from a regression on $w_{jt}^1, w_{jt}^2, x^{(2)}_{jt}$, and the differentiation instruments based on $x^{(2)}_{jt}$. I then construct approximate feasible optimal instruments and reestimate the model.}

\subsection{RCL Model with Normal Price and Characteristic Coefficients}

In this experiment, I generate data from an RCL model with normally distributed random coefficients on both price and an exogenous product characteristic. Consumer \(n\)'s conditional indirect utility from choosing product \(j \in \mathcal{J}_0\) in market \(t\) is 
\begin{equation}\label{eq:uRCL}
    u_{njt} = \beta_0 - \alpha_n p_{jt} +  \beta_n x_{jt}^{(2)}  + \xi_{jt} + \varepsilon_{njt}.
\end{equation}
For the outside good, \(u_{n0t}=\varepsilon_{n0t}\). The random utility shocks \(\varepsilon_{njt}\) are i.i.d. according to a type I extreme value distribution across consumers, products, and markets. I set the intercept to $\beta_0=-4$ and draw the random coefficients according to $\alpha_n \sim \mathcal{N}(0.5,0.15^2)$ and $\beta_n \sim \mathcal{N}(2,3^2)$, allowing for substantial unobserved heterogeneity in preferences.

Table~\ref{tab:mc_normal_j50} reports the results for $J=50$.
The results for $J=25$ are reported in Table~\ref{tab:mc_normal_j25}
in Supplemental Appendix~\ref*{app:mc_j25}. I compare SNPD with a correctly specified RCL model and a logit model. Because BLP imposes the correct functional form
and random coefficient distribution, it provides a benchmark for measuring
the finite-sample cost of relaxing the correct parametric structure. The logit model provides another useful misspecified benchmark: its poor performance confirms that the DGP is meaningfully different from the logit model and that the Monte Carlo exercise is nontrivial.

As expected, BLP performs very well: it yields estimates close to
the true values, with small RMSEs. Conversely, the logit model
performs substantially worse, understating own- and cross-price
elasticities and overstating diversion to the outside good and markups. By construction, it also fails to reproduce the stronger substitution toward nearby products in characteristic space. These results show that the DGP is sufficiently different from the
logit model that recovering the relevant economic objects requires a richer
demand model.

SNPD also performs well for both $J=25$ and $J=50$. It closely recovers average own- and cross-price elasticities and markups and captures the stronger substitution toward the five closest products. However, it slightly attenuates the difference between close and more distant products. The same pattern appears in diversion ratios: SNPD assigns less total diversion to inside products and more to the outside good, a difference that does not decline with $T$. Moreover, as expected, its RMSEs are larger than those of the correctly specified RCL model and, except for own-price curvature when $T = 100$, substantially smaller than those of the misspecified logit model. The RMSEs decline as the number of markets increases from \(100\) to \(500\), both when \(J=25\) and when \(J=50\). Overall, the finite-sample cost of flexibility is modest.

\begin{table}[H]
\centering
\caption{RCL model with normal price and characteristic coefficients: $J=50$}
\label{tab:mc_normal_j50}
\begingroup
\fontsize{10.25}{10.5}\selectfont
\setlength{\tabcolsep}{0.8pt}
\renewcommand{\arraystretch}{0.92}
\begin{tabular*}{\textwidth}{
@{\extracolsep{\fill}}l c c
r@{\hspace{4pt}}r
r@{\hspace{4pt}}r
r@{\hspace{4pt}}r@{}}
\toprule
Statistic & $T$ & True
& \multicolumn{2}{c}{\makebox[0pt][c]{\shortstack{Correctly\\[-1pt] specified RCL}}}
& \multicolumn{2}{c}{SNPD}
& \multicolumn{2}{c}{Logit} \\
\cmidrule(lr){4-5}
\cmidrule(lr){6-7}
\cmidrule(lr){8-9}
& & & Mean & RMSE & Mean & RMSE & Mean & RMSE \\
\midrule

\multicolumn{9}{@{}l}{\textit{Price elasticities}} \\
\quad Own
& 100 & -2.508 & -2.504 & 0.039
& -2.497 & 0.042 & -2.392 & 0.127 \\
& 500 & -2.508 & -2.510 & 0.017
& -2.501 & 0.020 & -2.393 & 0.118 \\[0.5pt]

\quad Cross, all
& 100 & 0.022 & 0.022 & 0.000
& 0.023 & 0.005 & 0.013 & 0.009 \\
& 500 & 0.022 & 0.022 & 0.000
& 0.023 & 0.002 & 0.013 & 0.009 \\[0.5pt]

\quad Cross, five closest
& 100 & 0.032 & 0.032 & 0.001
& 0.029 & 0.008 & 0.013 & 0.020 \\
& 500 & 0.032 & 0.032 & 0.000
& 0.029 & 0.004 & 0.013 & 0.020 \\[0.5pt]

\quad Cross, others
& 100 & 0.021 & 0.021 & 0.000
& 0.023 & 0.005 & 0.013 & 0.008 \\
& 500 & 0.021 & 0.021 & 0.000
& 0.023 & 0.003 & 0.013 & 0.008 \\

\addlinespace[2pt]
\multicolumn{9}{@{}l}{\textit{Mean diversion ratios}} \\
\quad All
& 100 & 0.011 & 0.011 & 0.000
& 0.009 & 0.002 & 0.005 & 0.006 \\
& 500 & 0.011 & 0.011 & 0.000
& 0.009 & 0.002 & 0.005 & 0.006 \\[0.5pt]

\quad Five closest
& 100 & 0.025 & 0.025 & 0.000
& 0.020 & 0.006 & 0.009 & 0.016 \\
& 500 & 0.025 & 0.025 & 0.000
& 0.020 & 0.005 & 0.009 & 0.016 \\[0.5pt]

\quad Other products
& 100 & 0.009 & 0.009 & 0.000
& 0.008 & 0.002 & 0.005 & 0.004 \\
& 500 & 0.009 & 0.009 & 0.000
& 0.008 & 0.001 & 0.005 & 0.004 \\

\addlinespace[2pt]
\multicolumn{9}{@{}l}{\textit{Total diversion ratios}} \\
\quad Five closest
& 100 & 0.125 & 0.125 & 0.001
& 0.101 & 0.029 & 0.046 & 0.079 \\
& 500 & 0.125 & 0.125 & 0.000
& 0.101 & 0.025 & 0.046 & 0.079 \\[0.5pt]

\quad Other inside
& 100 & 0.399 & 0.400 & 0.004
& 0.348 & 0.089 & 0.206 & 0.194 \\
& 500 & 0.400 & 0.400 & 0.001
& 0.346 & 0.061 & 0.206 & 0.194 \\[0.5pt]

\quad Outside good
& 100 & 0.476 & 0.475 & 0.006
& 0.551 & 0.115 & 0.749 & 0.273 \\
& 500 & 0.476 & 0.476 & 0.002
& 0.554 & 0.085 & 0.749 & 0.273 \\

\addlinespace[2pt]
\multicolumn{9}{@{}l}{\textit{Price curvature}} \\
\quad Own
& 100 & 1.103 & 1.108 & 0.076
& 1.083 & 0.122 & 0.995 & 0.108 \\
& 500 & 1.103 & 1.106 & 0.020
& 1.087 & 0.056 & 0.995 & 0.108 \\

\addlinespace[2pt]
\multicolumn{9}{@{}l}{\textit{Markups}} \\
\quad Markup (\%)
& 100 & 40.29 & 40.38 & 0.941
& 40.38 & 0.995 & 42.16 & 2.076 \\
& 500 & 40.29 & 40.21 & 0.337
& 40.18 & 0.430 & 42.13 & 1.894 \\
\bottomrule
\end{tabular*}
\endgroup

\par\vspace{1pt}
\begin{minipage}{\textwidth}
\fontsize{8.5}{9.2}\selectfont
\textit{Notes:}
Results are based on 100 Monte Carlo replications with \(J=50\) products
per market. \(T\) denotes the number of markets. ``True'' and ``Mean'' are
averages across replications; RMSE is the square root of the average squared
difference between estimated and true statistics. Price elasticities and
own-price curvature are unweighted. Diversion ratios weight source products
by market share. Mean diversion ratios average across individual
destinations, whereas total diversion ratios sum within destination groups.
Markups are market-share-weighted price-cost margins expressed as a
percentage of price. The five closest products minimize Euclidean distance
in price-characteristic space after rescaling both variables to \([0,1]\)
within each market.
\end{minipage}
\end{table}

\subsection{RCL Model with Lognormal Price and Normal Characteristic Coefficients}

In this experiment, I generate data from an RCL model with a lognormally
distributed random coefficient on price and a normally distributed random
coefficient on an exogenous product characteristic. Consumer \(n\)'s
conditional indirect utility from choosing product
\(j\in\mathcal{J}_0\) in market \(t\) is given by \eqref{eq:uRCL}. I set
\(\beta_0=-3\) and draw the random coefficients according to
\(\alpha_n\sim\log\mathcal{N}(0,1)\) and
\(\beta_n\sim\mathcal{N}(2,3^2)\).

Table~\ref{tab:mc_snpd_j50} reports the results for $J=50$.
The results for $J=25$ are reported in Table~\ref{tab:mc_snpd_j25}
in Supplemental Appendix~\ref*{app:mc_j25}. I compare SNPD with two RCL models and a logit model. The first RCL model imposes the correct functional form and random-coefficient distributions. Its performance relative to SNPD measures the finite-sample cost of flexibility. The second RCL model assumes normally distributed price and characteristic coefficients and is therefore misspecified. Its performance relative to SNPD measures the gain from flexibility when the functional form is correct but the assumed parametric distribution of the random coefficients is incorrect. As before, the poor performance of the logit model confirms that the Monte Carlo exercise is nontrivial.

\begin{table}[H]
\centering
\caption{RCL model with lognormal price and normal characteristic coefficients: $J=50$}
\label{tab:mc_snpd_j50}
\begingroup
\fontsize{10.25}{10.5}\selectfont
\setlength{\tabcolsep}{0.4pt}
\renewcommand{\arraystretch}{0.92}
\begin{tabular*}{\textwidth}{
@{\extracolsep{\fill}}l c c
r@{\hspace{3pt}}r
r@{\hspace{3pt}}r
r@{\hspace{3pt}}r
r@{\hspace{3pt}}r@{}}
\toprule
Statistic & $T$ & True
& \multicolumn{2}{c}{\makebox[0pt][c]{\shortstack{Correctly\\[-1pt] specified RCL}}}
& \multicolumn{2}{c}{\makebox[0pt][c]{\shortstack{Misspecified\\[-1pt] RCL}}}
& \multicolumn{2}{c}{SNPD}
& \multicolumn{2}{c}{Logit} \\
\cmidrule(lr){4-5}
\cmidrule(lr){6-7}
\cmidrule(lr){8-9}
\cmidrule(lr){10-11}
& & & Mean & RMSE & Mean & RMSE & Mean & RMSE & Mean & RMSE \\
\midrule

\multicolumn{11}{@{}l}{\textit{Price elasticities}} \\
\quad Own
& 100 & -2.285 & -2.282 & 0.029 & -2.259 & 0.040
& -2.261 & 0.040 & -2.314 & 0.050 \\
& 500 & -2.285 & -2.287 & 0.013 & -2.264 & 0.025
& -2.263 & 0.026 & -2.314 & 0.036 \\[0.5pt]

\quad Cross, all
& 100 & 0.023 & 0.023 & 0.001 & 0.017 & 0.006
& 0.022 & 0.004 & 0.009 & 0.014 \\
& 500 & 0.023 & 0.023 & 0.000 & 0.017 & 0.006
& 0.022 & 0.002 & 0.009 & 0.014 \\[0.5pt]

\quad Cross, five closest
& 100 & 0.035 & 0.034 & 0.001 & 0.028 & 0.007
& 0.031 & 0.008 & 0.009 & 0.026 \\
& 500 & 0.035 & 0.035 & 0.000 & 0.028 & 0.007
& 0.031 & 0.004 & 0.009 & 0.026 \\[0.5pt]

\quad Cross, others
& 100 & 0.021 & 0.021 & 0.001 & 0.016 & 0.006
& 0.021 & 0.003 & 0.009 & 0.013 \\
& 500 & 0.021 & 0.021 & 0.000 & 0.016 & 0.006
& 0.021 & 0.001 & 0.009 & 0.013 \\

\addlinespace[2pt]
\multicolumn{11}{@{}l}{\textit{Mean diversion ratios}} \\
\quad All
& 100 & 0.011 & 0.011 & 0.000 & 0.009 & 0.002
& 0.011 & 0.002 & 0.004 & 0.008 \\
& 500 & 0.011 & 0.011 & 0.000 & 0.009 & 0.002
& 0.011 & 0.001 & 0.004 & 0.008 \\[0.5pt]

\quad Five closest
& 100 & 0.026 & 0.026 & 0.000 & 0.022 & 0.004
& 0.023 & 0.004 & 0.006 & 0.019 \\
& 500 & 0.026 & 0.026 & 0.000 & 0.022 & 0.004
& 0.023 & 0.003 & 0.006 & 0.019 \\[0.5pt]

\quad Other products
& 100 & 0.010 & 0.010 & 0.000 & 0.008 & 0.002
& 0.009 & 0.002 & 0.003 & 0.006 \\
& 500 & 0.010 & 0.010 & 0.000 & 0.008 & 0.002
& 0.009 & 0.001 & 0.003 & 0.006 \\

\addlinespace[2pt]
\multicolumn{11}{@{}l}{\textit{Total diversion ratios}} \\
\quad Five closest
& 100 & 0.129 & 0.128 & 0.002 & 0.109 & 0.019
& 0.117 & 0.020 & 0.032 & 0.097 \\
& 500 & 0.129 & 0.129 & 0.001 & 0.109 & 0.019
& 0.117 & 0.013 & 0.032 & 0.097 \\[0.5pt]

\quad Other inside
& 100 & 0.426 & 0.425 & 0.007 & 0.345 & 0.081
& 0.413 & 0.068 & 0.150 & 0.276 \\
& 500 & 0.426 & 0.425 & 0.003 & 0.345 & 0.081
& 0.413 & 0.031 & 0.150 & 0.276 \\[0.5pt]

\quad Outside good
& 100 & 0.446 & 0.447 & 0.009 & 0.546 & 0.100
& 0.470 & 0.083 & 0.818 & 0.373 \\
& 500 & 0.446 & 0.446 & 0.003 & 0.546 & 0.100
& 0.470 & 0.040 & 0.818 & 0.373 \\

\addlinespace[2pt]
\multicolumn{11}{@{}l}{\textit{Price curvature}} \\
\quad Own
& 100 & 1.361 & 1.359 & 0.009 & 1.563 & 0.227
& 1.395 & 0.105 & 0.996 & 0.364 \\
& 500 & 1.361 & 1.360 & 0.003 & 1.547 & 0.191
& 1.395 & 0.057 & 0.996 & 0.364 \\

\addlinespace[2pt]
\multicolumn{11}{@{}l}{\textit{Markups}} \\
\quad Markup (\%)
& 100 & 42.44 & 42.52 & 0.503 & 41.94 & 0.803
& 44.31 & 2.114 & 46.04 & 3.681 \\
& 500 & 42.43 & 42.40 & 0.221 & 41.86 & 0.626
& 44.14 & 1.753 & 46.01 & 3.602 \\
\bottomrule
\end{tabular*}
\endgroup

\par\vspace{1pt}
\begin{minipage}{\textwidth}
\fontsize{8.5}{9.2}\selectfont
\textit{Notes:}
Results are based on 100 Monte Carlo replications with \(J=50\) products
per market. \(T\) denotes the number of markets. ``True'' and ``Mean'' are
averages across replications; RMSE is the square root of the average squared
difference between estimated and true statistics. Price elasticities and
own-price curvature are unweighted. Diversion ratios weight source products
by market share. Mean diversion ratios average across individual
destinations, whereas total diversion ratios sum within destination groups.
Markups are market-share-weighted price-cost margins expressed as a
percentage of price. The five closest products minimize Euclidean distance
in price-characteristic space after rescaling both variables to \([0,1]\)
within each market.
\end{minipage}
\end{table}

The misspecified RCL model recovers own-price elasticities reasonably well but systematically understates cross-price elasticities, both for the five closest products and for the remaining products. It therefore understates substitution between inside products and assigns too much of the lost demand to the outside good. The logit model performs substantially worse, especially for cross-price elasticities and diversion ratios. The RMSEs of the misspecified RCL model are generally larger than those of the correctly specified RCL model but smaller than those of the logit model. These results show that correctly specifying the distribution of the random coefficients matters for recovering substitution patterns.

SNPD performs well for both $J=25$ and $J=50$. It closely recovers average own- and cross-price elasticities and captures the stronger substitution toward the five closest products. Relative to the misspecified RCL model, SNPD more accurately recovers substitution between products and the allocation of diversion between inside products and the outside good, even though its nonparametric component excludes terms involving the outside good. These results show that the exclusion does not prevent SNPD from improving on the misspecified RCL model in these experiments. The RMSEs for cross-price elasticities and diversion ratios are generally smaller for SNPD than for the misspecified RCL model and substantially smaller than for the logit model. For markups, SNPD does not consistently outperform the misspecified RCL model, but its estimates remain reasonably close to the true values. Its RMSEs generally decline as the number of markets increases from \(100\) to \(500\), both when \(J=25\) and when \(J=50\). These results suggest that SNPD can be useful in applications where the distribution of consumer heterogeneity is unknown: it can recover substitution patterns without requiring the researcher to specify that distribution.

\subsection{Parametric SNPD}

In this experiment, I generate data directly from a parametric version of the inverse market share model. For each inside product $j\in\mathcal J_0$, market shares satisfy
\begin{equation}
\ln\!\left(\frac{s_{jt}}{s_{0t}}\right)
-\sum_{k\in\mathcal J_0\setminus\{j\}}
\mu_{jk,t}
\ln\!\left(\frac{s_{jt}}{s_{jt}+s_{kt}}\right)
=
\delta_{jt},
\label{eq:parametric_snpd_dgp}
\end{equation}
where $\delta_{jt}=
-2+2x^{(2)}_{jt}
-0.75p_{jt}+\xi_{jt}$ and
\begin{equation*}
\mu_{jk,t}
=
\frac{24}{J-1}
\left(
0.02
-0.005\lvert \widetilde{d}^p_{jk,t}\rvert
-0.005\lvert \widetilde{d}^{(2)}_{jk,t}\rvert
\right).
\end{equation*}

Because $\mu_{jk,t}$ decreases
with the distance between products in price and characteristic space,
closer products tend to be stronger substitutes. The factor $24/(J-1)$ keeps
the overall magnitude of the pairwise substitution terms comparable
between the designs with $J=25$ and $J=50$.

Table~\ref{tab:mc_parametric_j50} reports the results for $J=50$.
The results for $J=25$ are reported in Table~\ref{tab:mc_parametric_j25}
in Supplemental Appendix~\ref*{app:mc_j25}. I compare SNPD with an RCL model and a logit model. In this experiment, the RCL model is misspecified in its functional form. Its performance relative to SNPD measures the gain from flexibility when the DGP is specified directly through inverse market shares, complementing the preceding experiment on misspecified random coefficient distributions. As before, the poor performance of the logit model in recovering cross-price elasticities and diversion ratios confirms that the Monte Carlo exercise is nontrivial.

All three models recover average own-price elasticities reasonably well.
The misspecified RCL model systematically understates cross-price
elasticities and diversion between inside products and assigns too much
of the lost demand to the outside good. These discrepancies persist as
the number of markets increases from 100 to 500, both when $J=25$ and
when $J=50$. The logit model performs substantially worse, especially
for cross-price elasticities and diversion ratios.

\begin{table}[H]
\centering
\caption{Parametric inverse market share model: $J=50$}
\label{tab:mc_parametric_j50}
\begingroup
\fontsize{10.25}{10.5}\selectfont
\setlength{\tabcolsep}{0.8pt}
\renewcommand{\arraystretch}{0.92}
\begin{tabular*}{\textwidth}{
@{\extracolsep{\fill}}l c c
r@{\hspace{4pt}}r
r@{\hspace{4pt}}r
r@{\hspace{4pt}}r@{}}
\toprule
Statistic & $T$ & True
& \multicolumn{2}{c}{\makebox[0pt][c]{\shortstack{Misspecified\\[-1pt] RCL}}}
& \multicolumn{2}{c}{SNPD}
& \multicolumn{2}{c}{Logit} \\
\cmidrule(lr){4-5}
\cmidrule(lr){6-7}
\cmidrule(lr){8-9}
& & & Mean & RMSE & Mean & RMSE & Mean & RMSE \\
\midrule

\multicolumn{9}{@{}l}{\textit{Price elasticities}} \\
\quad Own
& 100 & -4.323 & -4.321 & 0.041 & -4.317 & 0.041 & -4.329 & 0.041 \\
& 500 & -4.324 & -4.325 & 0.019 & -4.320 & 0.019 & -4.333 & 0.020 \\[0.5pt]

\quad Cross, all
& 100 & 0.040 & 0.036 & 0.004 & 0.040 & 0.004 & 0.028 & 0.012 \\
& 500 & 0.040 & 0.036 & 0.004 & 0.040 & 0.002 & 0.028 & 0.012 \\[0.5pt]

\quad Cross, five closest
& 100 & 0.043 & 0.037 & 0.006 & 0.045 & 0.007 & 0.028 & 0.015 \\
& 500 & 0.043 & 0.037 & 0.006 & 0.045 & 0.004 & 0.028 & 0.015 \\[0.5pt]

\quad Cross, others
& 100 & 0.040 & 0.036 & 0.004 & 0.039 & 0.003 & 0.028 & 0.012 \\
& 500 & 0.040 & 0.036 & 0.004 & 0.039 & 0.002 & 0.028 & 0.012 \\

\addlinespace[2pt]
\multicolumn{9}{@{}l}{\textit{Mean diversion ratios}} \\
\quad All
& 100 & 0.009 & 0.008 & 0.001 & 0.009 & 0.001 & 0.007 & 0.002 \\
& 500 & 0.009 & 0.008 & 0.001 & 0.009 & 0.000 & 0.007 & 0.002 \\[0.5pt]

\quad Five closest
& 100 & 0.011 & 0.010 & 0.001 & 0.011 & 0.002 & 0.008 & 0.003 \\
& 500 & 0.011 & 0.010 & 0.001 & 0.011 & 0.001 & 0.008 & 0.003 \\[0.5pt]

\quad Other products
& 100 & 0.009 & 0.008 & 0.001 & 0.009 & 0.001 & 0.006 & 0.002 \\
& 500 & 0.009 & 0.008 & 0.001 & 0.009 & 0.000 & 0.006 & 0.002 \\

\addlinespace[2pt]
\multicolumn{9}{@{}l}{\textit{Total diversion ratios}} \\
\quad Five closest
& 100 & 0.054 & 0.050 & 0.005 & 0.057 & 0.008 & 0.038 & 0.016 \\
& 500 & 0.054 & 0.050 & 0.005 & 0.057 & 0.004 & 0.038 & 0.016 \\[0.5pt]

\quad Other inside
& 100 & 0.379 & 0.356 & 0.026 & 0.377 & 0.034 & 0.280 & 0.099 \\
& 500 & 0.379 & 0.358 & 0.022 & 0.375 & 0.015 & 0.280 & 0.099 \\[0.5pt]

\quad Outside good
& 100 & 0.567 & 0.594 & 0.031 & 0.566 & 0.041 & 0.682 & 0.115 \\
& 500 & 0.567 & 0.592 & 0.026 & 0.569 & 0.018 & 0.682 & 0.115 \\

\addlinespace[2pt]
\multicolumn{9}{@{}l}{\textit{Price curvature}} \\
\quad Own
& 100 & 0.817 & 1.051 & 0.234 & 0.938 & 0.128 & 0.994 & 0.176 \\
& 500 & 0.817 & 1.053 & 0.236 & 0.945 & 0.129 & 0.994 & 0.176 \\

\addlinespace[2pt]
\multicolumn{9}{@{}l}{\textit{Markups}} \\
\quad Markup ($\%$)
& 100 & 26.36 & 24.26 & 2.120 & 25.25 & 1.159 & 24.21 & 2.171 \\
& 500 & 26.36 & 24.23 & 2.130 & 25.18 & 1.196 & 24.18 & 2.179 \\
\bottomrule
\end{tabular*}
\endgroup

\par\vspace{1pt}
\begin{minipage}{\textwidth}
\fontsize{8.5}{9.2}\selectfont
\textit{Notes:}
Results are based on 100 Monte Carlo replications with \(J=50\) products
per market. \(T\) denotes the number of markets. ``True'' and ``Mean'' are
averages across replications; RMSE is the square root of the average squared
difference between estimated and true statistics. Price elasticities and
price curvature are unweighted. Diversion ratios weight source products by
market share. Mean diversion ratios average across individual destinations,
whereas total diversion ratios sum within destination groups. Markups are
market-share-weighted price-cost margins expressed as a percentage of price.
The five closest products minimize Euclidean distance in
price-characteristic space after rescaling both variables to \([0,1]\)
within each market.
\end{minipage}
\end{table}

SNPD performs well. It closely recovers average cross-price elasticities
and more accurately recovers the allocation of diversion between inside
products and the outside good. It also captures the stronger substitution
toward the five closest products, although it overstates their average
cross-price elasticities in both designs. For markups, SNPD has smaller
RMSEs than the RCL and logit models in all four configurations. Its
estimates are close to the true values when $J=25$, but remain below
them when $J=50$. The RMSE comparisons depend on the number of markets. When $T=100$,
SNPD has larger RMSEs than the misspecified RCL model for total diversion
to the five closest products, the remaining inside products, and the
outside good, despite recovering their average levels more accurately.
When $T=500$, its RMSEs are smaller for all three measures in both
designs. SNPD's RMSEs generally decline as the number of markets
increases, although its markup  and own-price curvature estimates do not improve when $J=50$.
These results show that SNPD can improve the recovery of substitution
patterns when the RCL functional form is misspecified, but that its
flexibility can come at a cost in smaller samples.

\section{Empirical Application: the U.S. Beer Market}
\label{sec:empirical_applications}

I apply my approach to the U.S. beer market studied by \citet*{miller2017understanding} and compare its estimates with their random coefficient nested logit model. With 2,808 markets and up to 39 products per market, this application assesses the flexibility and scalability of the proposed approach in an empirical setting. Details are provided in Supplemental Appendix \ref*{Appendix:Appli}.

\subsection{Data and Demand Models}

\paragraph{Data.} I use the same sample from the IRI marketing data set \citep*{bronnenberg2008database} as \citet*{miller2017understanding}. The sample covers 39 geographic regions over 2005 -- 2011, excluding the first year following the MillerCoors joint venture. It contains $J=39$ products across $T=2{,}808$ markets, yielding $N=94{,}656$ product-market observations. Markets are region-month combinations. Products are brand-size combinations in 6-pack, 12-pack, and combined 24/30-pack formats. Product availability varies across markets. Market shares are computed from sales measured in 144-ounce equivalent units, and prices are revenue per equivalent unit. The Public Use Microdata Sample provides the regional income distributions used in demand estimation. The dataset also includes 12 excluded instruments: an indicator for AB InBev and MillerCoors products after the merger; brewery-to-market distance interacted with the diesel price index; the number of products and total distance across products, together with their interactions with AB InBev and MillerCoors indicators; and mean income and its interactions with import status, calories, and package size.

\paragraph{Demand models.} I compare SNPD with the RCNL model of \citet*{miller2017understanding}. Both specifications include product and date fixed effects. The RCNL model uses regional income distributions to generate preference heterogeneity and places all inside goods in a single nest. SNPD includes mean income and its interactions with price and calories, together with a permutation-invariant Chebyshev sieve of transformed share ratios and price and calorie differences. The same function $\widetilde G_D$ is used for markets with different numbers of products. Cross-validation over $D\in\{3,4,5,6\}$ selects $D=3$, yielding 53 nonconstant sieve terms. I estimate RCNL using PyBLP \citep*{conlon2018} with 12 excluded instruments, and SNPD by constrained two-step GMM using these instruments, their symmetric transformations, and the included exogenous regressors.

\subsection{Demand Estimates}

Table~\ref{tab:demand_parameter_estimates} reports the estimates. Both specifications have a negative price coefficient and a positive income-price interaction. In the RCNL model, the latter allows higher-income consumers to be less price sensitive; in the SNPD model, it allows the linear price coefficient to vary with mean market income. In the RCNL model, the nesting parameter is $0.889$, indicating a strong tendency for consumers to substitute among inside goods rather than toward the outside good.

\begin{table}[H]
\caption{Demand estimates}
\label{tab:demand_parameter_estimates}
\centering
\def\demandtablescale{0.95} \begin{threeparttable}
\setlength{\tabcolsep}{6pt}

\scalebox{\demandtablescale}{\begin{tabular}{lcccccc}
\toprule
& & \multicolumn{2}{c}{RCNL}
& \multicolumn{3}{c}{SNPD} \\
\cmidrule(lr){3-4}\cmidrule(lr){5-7}
Variable & Parameter
& Estimate & SE
& Estimate & SE & 95\% CI \\
\midrule
Price
& $-\alpha$
& $-0.082$ & $0.005$
& $-3.785$ & $0.181$ & $[-3.954,\,-3.259]$ \\

Nesting parameter
& $\mu$
& $0.889$ & $0.006$
& -- & -- & -- \\

\addlinespace

Income $\times$ price
& $\pi_p;\gamma_{yp}$
& $4.20\times10^{-4}$ & $3.00\times10^{-5}$
& $0.651$ & $0.025$ & $[0.592,\,0.691]$ \\

Income $\times$ constant
& $\pi_0;\gamma_y$
& $0.017$ & $0.001$
& $-0.352$ & $0.027$ & $[-0.399,\,-0.292]$ \\

Income $\times$ calories
& $\pi_c;\gamma_{yc}$
& $0.002$ & $2.40\times10^{-4}$
& $0.247$ & $0.013$ & $[0.217,\,0.269]$ \\

\addlinespace

GMM objective
&
& \multicolumn{2}{c}{$0.004$}
& \multicolumn{3}{c}{$1.66\times10^{-7}$} \\
Computation time
&
& \multicolumn{2}{c}{$\approx 10\text{h}30\text{min}$}
& \multicolumn{3}{c}{$\approx 30\text{s}$} \\
\bottomrule
\end{tabular}}
\par\smallskip
\begin{minipage}{\textwidth}
\footnotesize \textit{Notes:} RCNL standard errors (SE) are clustered by market. SNPD standard errors (SE) and confidence intervals (CI) use the \citet{hong2020numerical} procedure. \end{minipage}
\end{threeparttable}
\end{table}

The coefficients enter different demand models, which use different normalizations. It is therefore better to compare the models through substitution patterns and counterfactual outcomes. The GMM objective is lower under SNPD. This comparison depends on the instrument sets and weighting matrices used. A more systematic assessment could use the model-selection test of \citet*{rivers2002model} or the cross-validation procedure of \citet*{komiyama2018cross}.

The SNPD model takes approximately 30 seconds to estimate, whereas the RCNL model takes around 10 hours and 30 minutes. SNPD estimates the inverse market share equation directly, avoiding repeated numerical share inversion. In addition, for a fixed polynomial degree, the number of sieve coefficients does not grow with the number of products. This application therefore illustrates the computational tractability of the proposed approach in a setting with many markets.

\subsection{Substitution Patterns}

Table~\ref{tab:empirical_demand_results} reports the implied substitution patterns and markups for the RCNL and SNPD models. Markups are computed under a static oligopolistic price competition model. The five closest products are defined by Euclidean distance in price and calories rescaled to $[0,1]$ within each market. 

\begin{table}[H]
\centering
\caption{Demand elasticities, diversion ratios, and markups}
\label{tab:empirical_demand_results}
\begingroup
\fontsize{10.25}{10.5}\selectfont
\setlength{\tabcolsep}{19.5pt}
\renewcommand{\arraystretch}{0.94}
\begin{tabular}{@{}lcccc@{}}
\toprule
Statistic
& \multicolumn{2}{c}{RCNL}
& \multicolumn{2}{c}{SNPD} \\
\cmidrule(lr){2-3}
\cmidrule(lr){4-5}
& Estimate & 95\% CI
& Estimate & 95\% CI \\
\midrule

\multicolumn{5}{@{}l}{\textit{Price elasticities}} \\
\quad Own
& $-5.596$ & $[-6.240,\,-5.070]$
& $-5.151$ & $[-5.582,\,-4.768]$ \\
\quad Cross, all
& $0.137$ & $[0.126,\,0.151]$
& $0.133$ & $[0.122,\,0.151]$ \\
\quad Cross, five closest
& $0.152$ & $[0.140,\,0.167]$
& $0.303$ & $[0.278,\,0.338]$ \\
\quad Cross, others
& $0.135$ & $[0.124,\,0.148]$
& $0.102$ & $[0.092,\,0.118]$ \\

\addlinespace[4.5pt]
\multicolumn{5}{@{}l}{\textit{Mean diversion ratios}} \\
\quad All

& $0.028$ & $[0.0276,\,0.0281]$
& $0.025$ & $[0.024,\,0.028]$ \\
\quad Five closest
& $0.047$ & $[0.044,\,0.052]$
& $0.054$ & $[0.051,\,0.059]$ \\
\quad Other products
& $0.024$ & $[0.024,\,0.025]$
& $0.020$ & $[0.019,\,0.023]$ \\

\addlinespace[4.5pt]
\multicolumn{5}{@{}l}{\textit{Total diversion ratios}} \\
\quad Five closest
& $0.235$ & $[0.220,\,0.258]$
& $0.272$ & $[0.257,\,0.295]$ \\
\quad Other inside
& $0.680$ & $[0.663,\,0.694]$
& $0.564$ & $[0.520,\,0.634]$ \\
\quad Outside good
& $0.085$ & $[0.075,\,0.094]$
& $0.164$ & $[0.090,\,0.202]$ \\

\addlinespace[4.5pt]
\multicolumn{5}{@{}l}{\textit{Markups}} \\
\quad Markup ($\%$)
& $32.623$ & $[30.836,\,34.556]$
& $37.821$ & $[33.005,\,40.839]$ \\
\bottomrule
\end{tabular}
\endgroup

\par\smallskip
\begin{minipage}{\textwidth}
\footnotesize
\textit{Notes:}
RCNL confidence intervals are based on market-clustered inference. SNPD
confidence intervals are obtained using the \cite{hong2020numerical} bootstrap procedure. The five closest products are defined by Euclidean distance in price-characteristic space, using price and calories rescaled to $[0,1]$ within each market. Price elasticities are unweighted. Diversion ratios are weighted by the market share of the source product. Mean diversion ratios average diversion to individual inside products, whereas total diversion ratios sum diversion within each destination group. Markups are market-share-weighted averages
of the price-cost margin, expressed as a percentage of price.
\end{minipage}
\end{table}

The SNPD and RCNL models generate similar average price
elasticities. Average own-price elasticities are $-5.151$ under SNPD
and $-5.596$ under RCNL, while average cross-price elasticities are
$0.133$ and $0.137$, respectively. The main difference concerns
substitution toward the five closest products. Their average cross-price
elasticity is $0.303$ under SNPD, compared with $0.152$ under RCNL.
For the remaining products, the corresponding estimates are $0.102$
and $0.135$. Diversion ratios show a similar pattern. The five closest products
receive $27.2\%$ of lost sales under SNPD, compared with $23.5\%$ under
RCNL. The remaining inside products receive $56.4\%$ and $68.0\%$,
respectively, while diversion to the outside good is $16.4\%$ and
$8.5\%$. Thus, although the models generate similar average elasticities,
SNPD assigns more diversion to nearby products and the outside good,
and less to more distant inside products. The implied weighted average markup is $37.8\%$ under SNPD and
$32.6\%$ under RCNL, with overlapping confidence intervals.

\paragraph{Own-price elasticities, prices, and demand curvature.} 
Figure~\ref{fig:elasticity_price_curvature} compares the relationship
between own-price elasticities and prices in Panel (a) and between
own-price elasticities and demand curvature in Panel (b). The first two rows report results for RCNL and SNPD, respectively. To assess whether these patterns depend on the price-income interaction, I also estimate an SNPD specification excluding this interaction.
The third row reports the results.

\newcommand{\BeerFigureScale}{0.7}

\begin{figure}[H]
\centering
\makebox[\textwidth][c]{\scalebox{\BeerFigureScale}{\begin{minipage}{\textwidth}

\caption{Own-price elasticities, prices, and demand curvature}
\label{fig:elasticity_price_curvature}

\medskip

\begin{minipage}[t]{0.48\textwidth}
    \centering
    (a) Own-price elasticities and prices
\end{minipage}
\hfill
\begin{minipage}[t]{0.48\textwidth}
    \centering
    (b) Own-price elasticities and curvature
\end{minipage}

\par\medskip

\begin{minipage}[t]{0.48\textwidth}
    \centering
    \includegraphics[width=\linewidth]{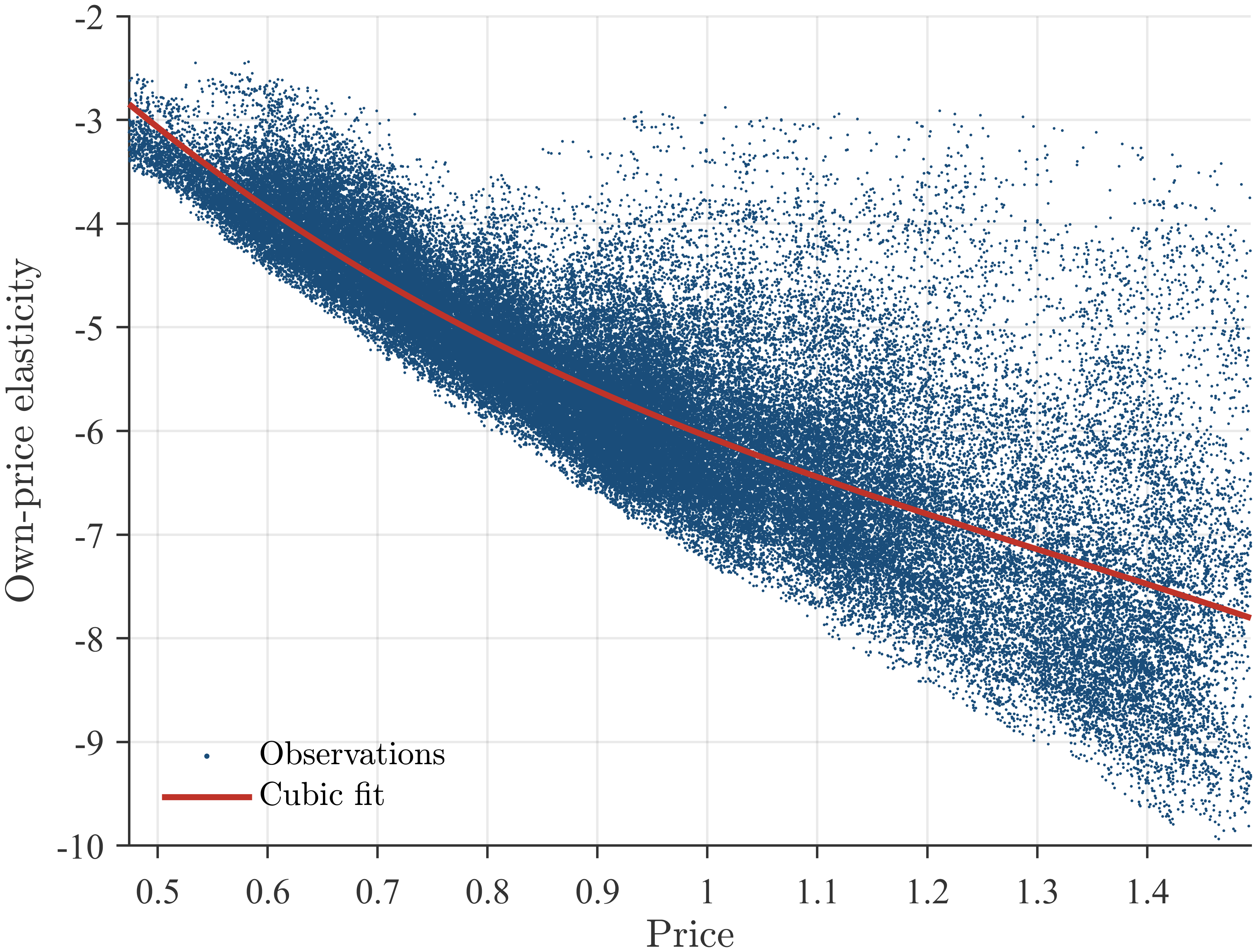}
    \par\smallskip
    RCNL
\end{minipage}
\hfill
\begin{minipage}[t]{0.48\textwidth}
    \centering
    \includegraphics[width=\linewidth]{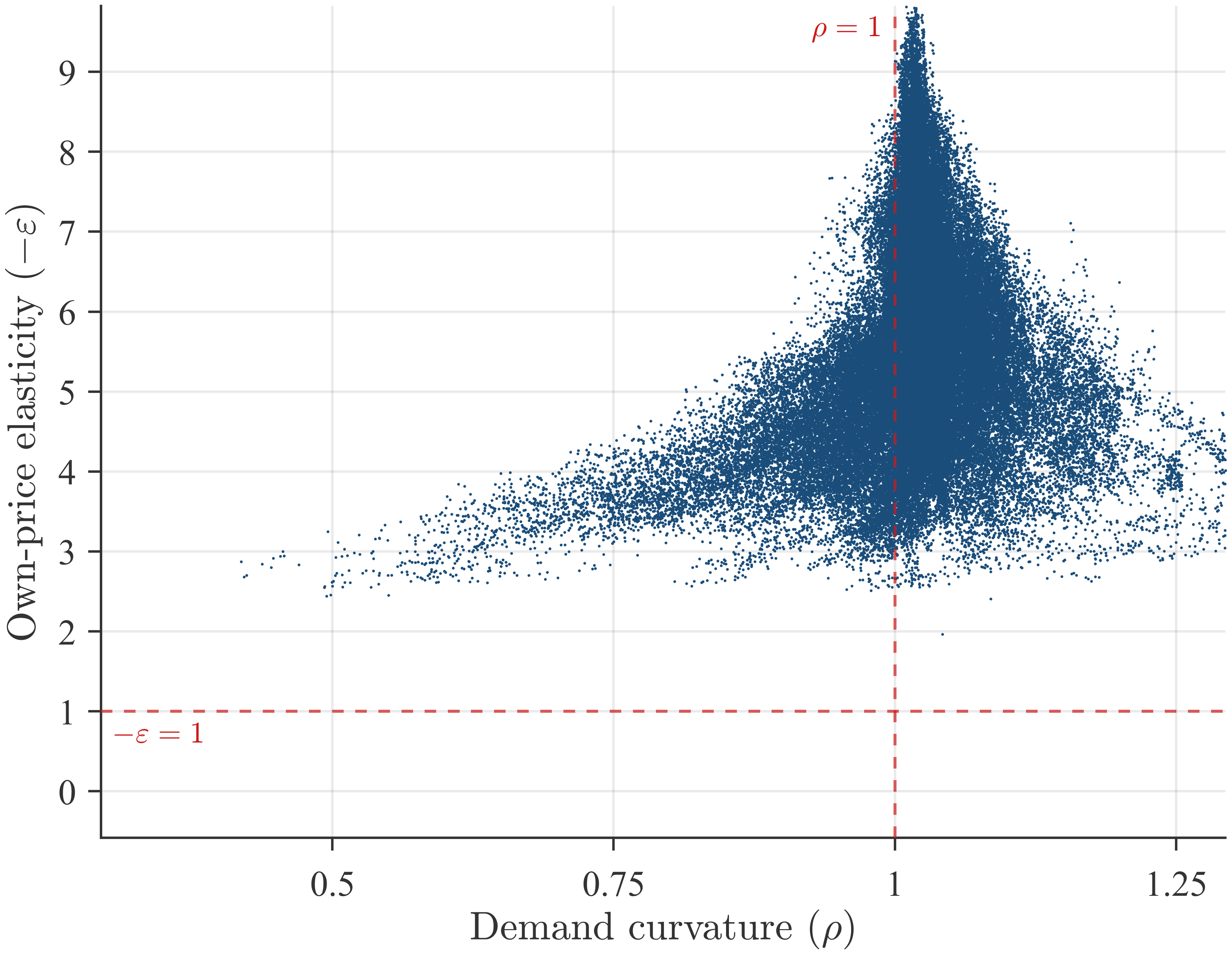}
    \par\smallskip
    RCNL
\end{minipage}

\par\medskip

\begin{minipage}[t]{0.48\textwidth}
    \centering
    \includegraphics[width=\linewidth]{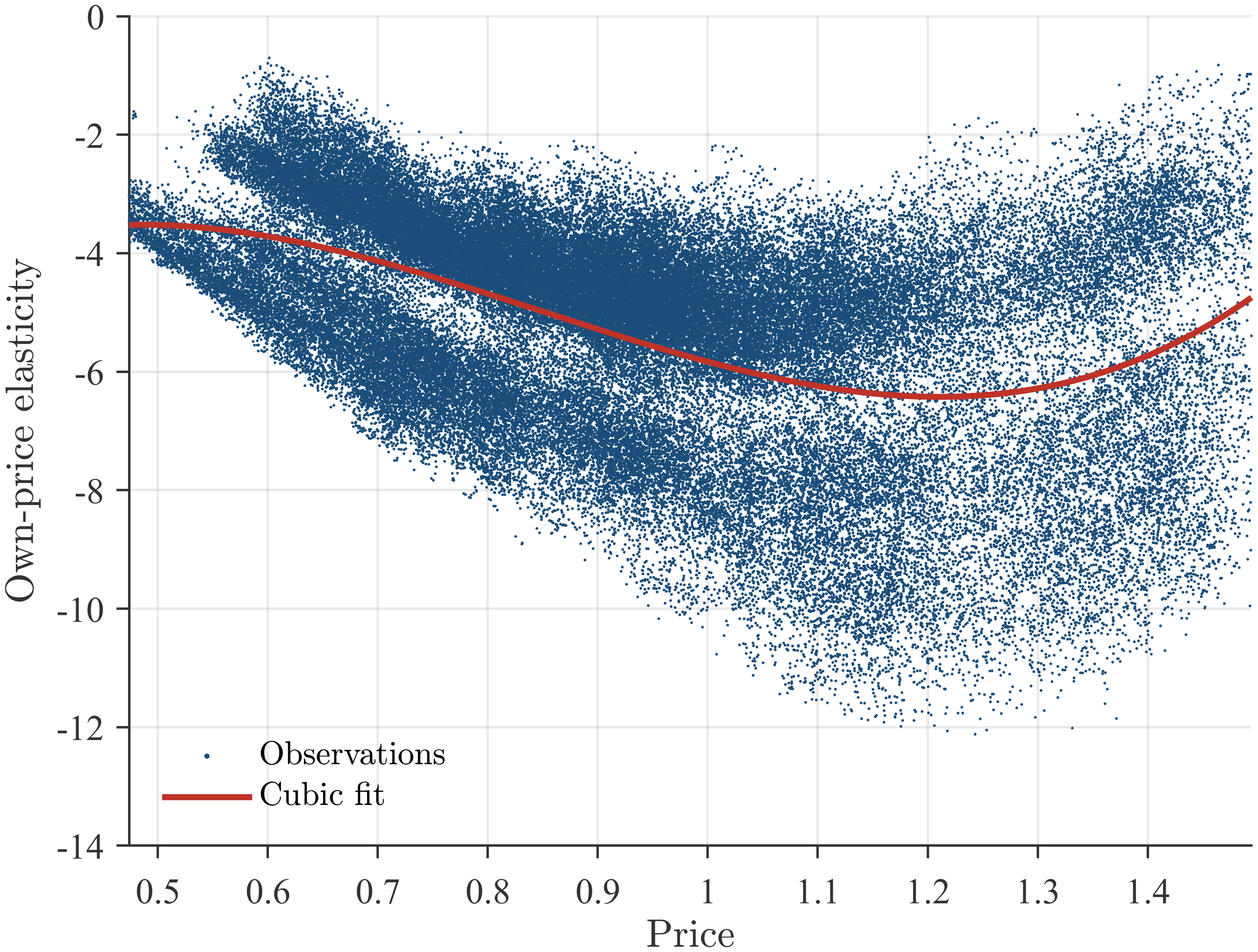}
    \par\smallskip
    SNPD
\end{minipage}
\hfill
\begin{minipage}[t]{0.48\textwidth}
    \centering
    \includegraphics[width=\linewidth]{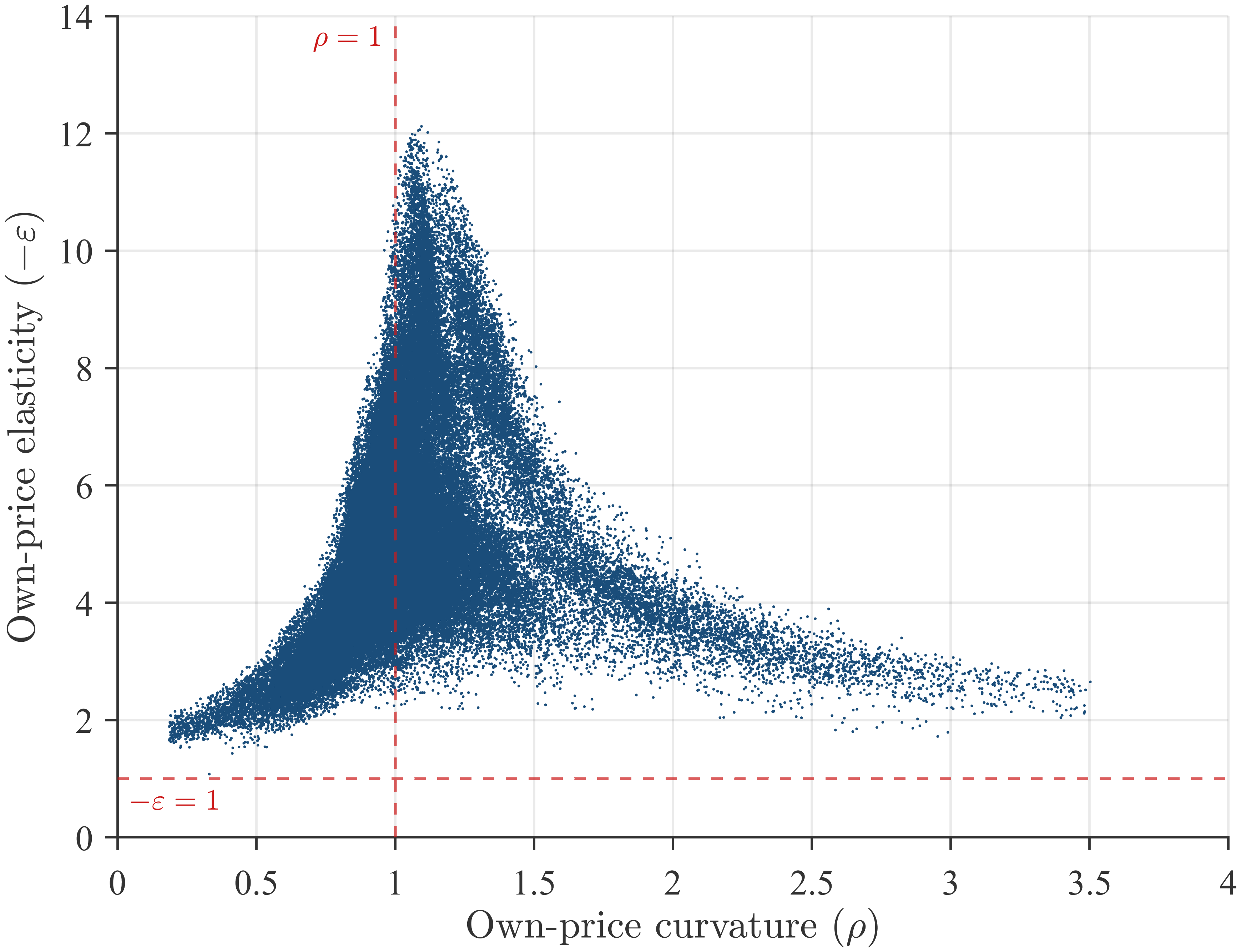}
    \par\smallskip
    SNPD
\end{minipage}

\par\medskip

\begin{minipage}[t]{0.48\textwidth}
    \centering
    \includegraphics[width=\linewidth]{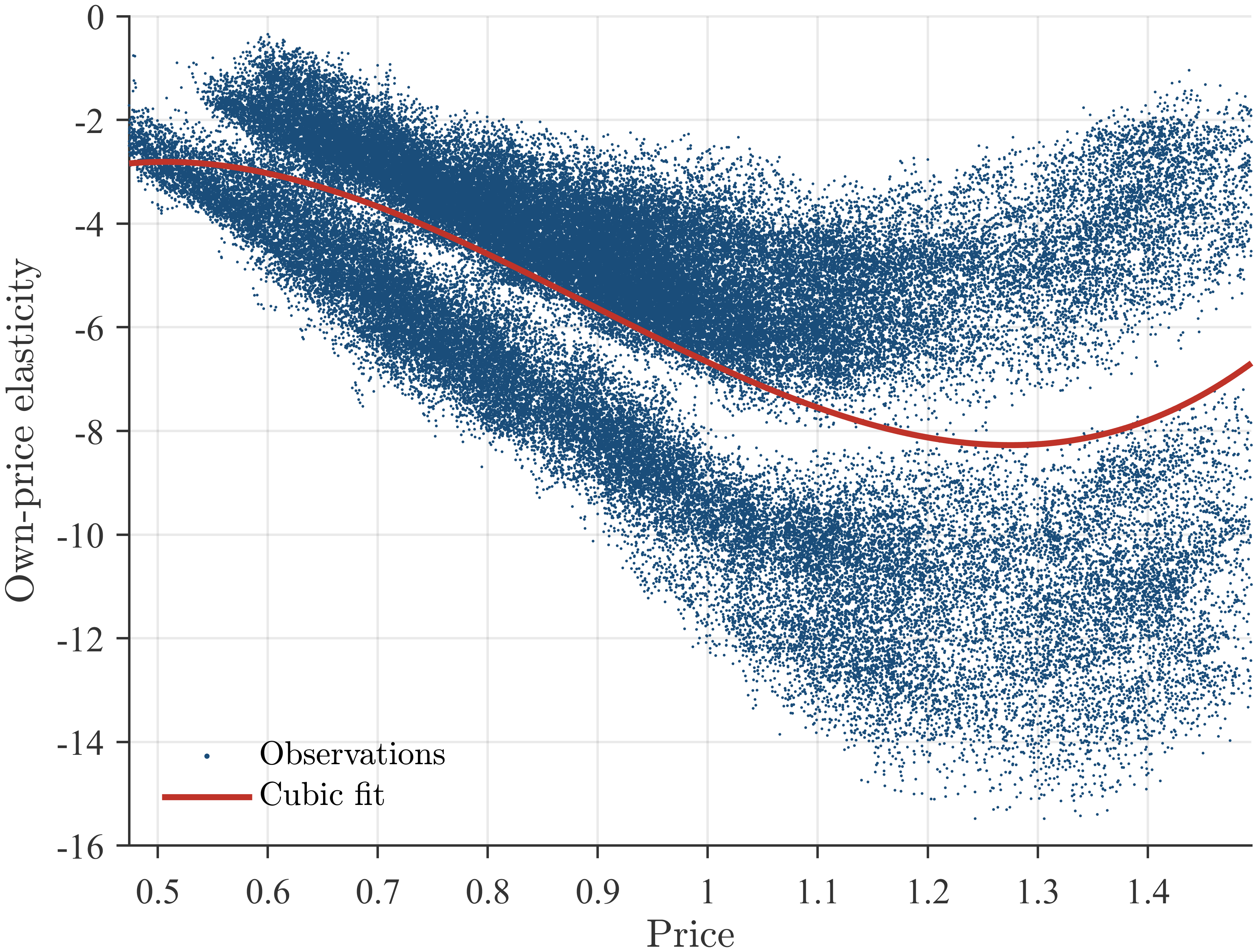}
    \par\smallskip
    SNPD -- no price interaction
\end{minipage}
\hfill
\begin{minipage}[t]{0.48\textwidth}
    \centering
    \includegraphics[width=\linewidth]{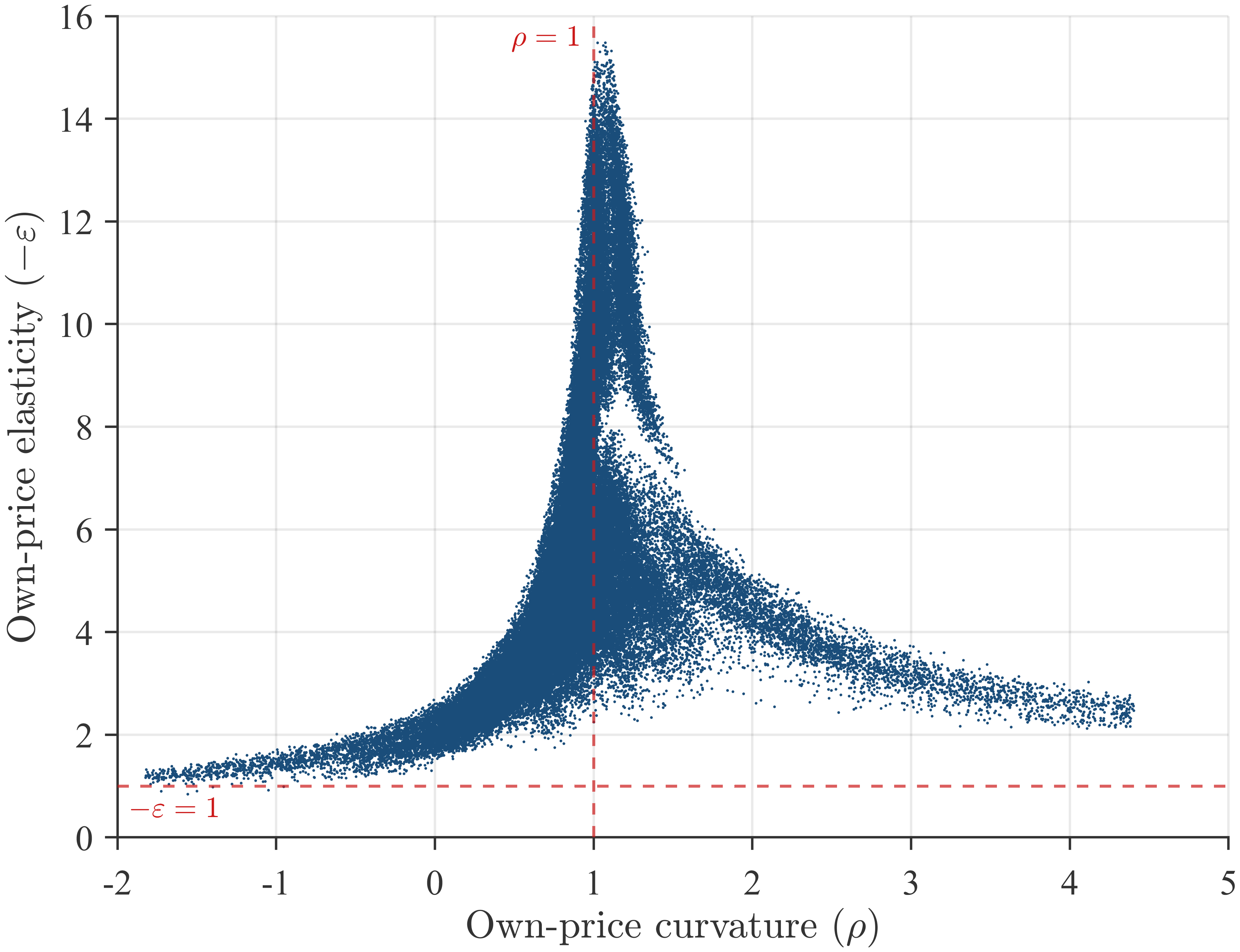}
    \par\smallskip
    SNPD -- no price interaction
\end{minipage}

\par\medskip

{\footnotesize
\textit{Notes:}
Each point represents a product-market observation. Rows report
estimates from RCNL, SNPD, and SNPD without price interaction,
respectively. Panel (a) plots own-price elasticities against prices. Panel (b) plots the magnitude of the own-price elasticity,
against curvature. In the first two rows of Panel (a), observations below the 0.5th and
above the 99.5th percentiles of price are excluded; red curves are
cubic fits across the displayed observations. In Panel (b), red dashed lines mark unit elasticity and unit
curvature. The second row of Panel (b) excludes observations below
the 0.5th and above the 99.5th percentiles of estimated curvature.
\par}

\end{minipage}}}
\end{figure}

In Panel (a), products with higher prices generally have higher own-price elasticities in absolute value under RCNL. Under both SNPD specifications, elasticities vary more at a given price, and the fitted relationship is nonmonotonic: absolute elasticities increase with price at low prices and decrease at higher prices, a shape that RCL models can also generate through heterogeneity in price sensitivity \citep*{birchall2024estimating,goldszmidt2026price}. Despite similar average elasticities, the models imply different patterns of price sensitivity across products and markets.

These estimates provide the flexibility in own-price responses called for by \citet*{birchall2024estimating} and \citet*{goldszmidt2026price}: own-price elasticities do not scale quasi-linearly with price, and products with similar prices can have different own-price elasticities. The latter show that own-price elasticities are approximately quadratic functions of own price when price and non-price random coefficients are independent and type-specific choice probabilities are small. The RCNL specification considered here allows these coefficients to be correlated through income. Nevertheless, SNPD produces greater variation in own-price elasticities at a given price.

Panel (b) examines demand curvature $\rho_{jt}$, defined in Section~\ref{Sec:Implementation}. After excluding observations below the 0.5th and above the 99.5th
percentiles of curvature in each specification, average curvature is
$1.020$ under RCNL, $1.061$ under baseline SNPD, and $0.890$ under SNPD
without the price-income interaction. Under RCNL, the curvature is concentrated
near one, whereas SNPD allows a wider range, including substantially
higher values. At similar elasticities, the models therefore imply
different demand curvature, an additional dimension of flexibility
beyond the substitution patterns in Table~\ref{tab:empirical_demand_results}. 

This comparison addresses the concern of \citet*{miravete2023pass,miravete2023elasticity} that similar elasticities can conceal different demand curvature. They show that heterogeneity in price sensitivity can raise curvature above one, while heterogeneity in tastes for non-price characteristics tends to lower it. SNPD allows flexibility in curvature without specifying a random coefficient distribution. Panel (b) illustrates this flexibility in the estimated specifications rather than establishing a general curvature restriction of RCNL. These differences matter for the counterfactual below because responses to cost shocks depend jointly on demand curvature, cross-product demand derivatives, and ownership.

\subsection{Counterfactual: a Tax on Imported Beer}

I consider a per-unit tax on imported products equal to $10\%$ of their prices. Holding product characteristics and ownership fixed,
I compute counterfactual prices under a static oligopolistic price competition. The exercise illustrates how substitution and curvature shape the price responses of both imported and domestic products to the tax.

Table~\ref{tab:counterfactual_results} reports the effects of the import
tax on prices, profits, and consumer surplus. Imported-product prices
increase by $11.60\%$ under SNPD and $10.23\%$ under RCNL. The larger
difference concerns domestic products, whose prices increase by
$6.92\%$ under SNPD, compared with $1.53\%$ under RCNL. Prices
across all inside products therefore increase by $7.94\%$ under SNPD
and $3.42\%$ under RCNL. SNPD thus implies a stronger domestic price
response to the increase in import costs. Domestic products face no direct cost increase, so their price responses reflect the indirect effects of higher import prices on pricing incentives. This mechanism also arises when nonmerging firms respond to 
price increases by merging firms. The larger domestic price response under SNPD shows that
the estimated demand specification matters for the transmission of cost shocks beyond the directly affected products. The small domestic price response under RCNL is reminiscent of the
weak responses of nonmerging firms predicted by some merger simulations under a static price competition model and RCL demand models \citep*{NevoRand,houde2012spatial,BjornerstedtVerboven2016}. Although the shock here affects costs rather than ownership, both exercises depend on how firms respond to competitors' price
changes. SNPD predicts a substantially larger domestic response.

\begin{table}[H]
\centering
\caption{Counterfactual effects of an increase in import costs}
\label{tab:counterfactual_results}
\begingroup
\fontsize{10.25}{10.5}\selectfont
\setlength{\tabcolsep}{19pt}
\renewcommand{\arraystretch}{0.94}
\begin{tabular}{@{}lcccc@{}}
\toprule
Statistic
& \multicolumn{2}{c}{RCNL}
& \multicolumn{2}{c}{SNPD} \\
\cmidrule(lr){2-3}
\cmidrule(lr){4-5}
& Estimate & 95\% CI
& Estimate & 95\% CI \\
\midrule

\multicolumn{5}{@{}l}{\textit{Price changes (\%)}} \\
\quad All inside products
& $3.423$ & $[3.344,\,3.539]$
& $7.941$ & $[6.791,\,8.778]$ \\
\quad Imported
& $10.23$ & $[9.935,\,10.66]$
& $11.60$ & $[9.578,\,11.93]$ \\
\quad Non-imported
& $1.527$ & $[1.489,\,1.579]$
& $6.920$ & $[5.793,\,8.054]$ \\
\quad AB InBev
& $1.536$ & $[1.499,\,1.585]$
& $6.864$ & $[5.806,\,8.143]$ \\
\quad Grupo Modelo
& $10.19$ & $[9.875,\,10.66]$
& $11.89$ & $[10.58,\,12.24]$ \\
\quad Heineken
& $10.28$ & $[10.01,\,10.67]$
& $11.22$ & $[6.972,\,11.54]$ \\
\quad Miller
& $1.518$ & $[1.476,\,1.573]$
& $6.986$ & $[5.531,\,8.060]$ \\

\addlinespace[4.5pt]
\multicolumn{5}{@{}l}{\textit{Profit changes (\%)}} \\
\quad All inside products
& $3.520$ & $[3.209,\,3.990]$
& $11.52$ & $[6.316,\,17.499]$ \\
\quad Imported
& $-29.85$ & $[-33.90,\,-26.06]$
& $-7.967$ & $[-13.40,\,-1.286]$ \\
\quad Non-imported
& $8.732$ & $[7.990,\,9.668]$
& $14.88$ & $[9.520,\,21.11]$ \\
\quad AB InBev
& $8.489$ & $[7.769,\,9.395]$
& $15.59$ & $[10.020,\,21.39]$ \\
\quad Grupo Modelo
& $-29.69$ & $[-33.61,\,-25.97]$
& $-9.000$ & $[-14.28,\,-2.722]$ \\
\quad Heineken
& $-30.07$ & $[-34.32,\,-26.18]$
& $-7.097$ & $[-12.72,\,0.478]$ \\
\quad Miller
& $9.015$ & $[8.247,\,9.986]$
& $14.16$ & $[8.925,\,20.69]$ \\

\addlinespace[4.5pt]
\multicolumn{5}{@{}l}{Consumer welfare measure (\$ per consumer)} \\
\quad All consumers
& $-0.050$ & $[-0.052,\,-0.048]$
& $-0.112$ & $[-0.127,\,-0.100]$ \\
\bottomrule
\end{tabular}
\endgroup

\par\smallskip
\begin{minipage}{\textwidth}
\footnotesize
\textit{Notes:}
The counterfactual increases the marginal cost of each imported product by
10\% of its baseline price. Price changes are percentage changes in price indexes weighted by
baseline market shares, $100(\sum_{j\in\mathcal J}s_j^0p_j^1/\sum_{j\in\mathcal J}s_j^0p_j^0-1)$, where $\mathcal J$ denotes
the relevant product group, $s_j^0$ denotes baseline market shares,
and $p_j^0$ and $p_j^1$ denote baseline and counterfactual prices. Profit changes compare aggregate profits within each product or
firm group. RCNL confidence intervals are percentile intervals from 999
parametric bootstrap draws based on the market-clustered covariance matrix.
SNPD confidence intervals are obtained using the \cite{hong2020numerical} bootstrap
procedure with 999 replications. The RCNL measure is computed using the log-sum formula, averaging consumer-specific surplus changes over the estimated distribution
of preferences. The SNPD welfare measure is the negative straight-line price integral of demand, equally weighted across the selected markets and expressed in the original dollar units. It is a path-specific measure.
\end{minipage}
\end{table}

Domestic profits increase by $14.88\%$ under SNPD and $8.73\%$ under
RCNL. Profits from imported products decline by $7.97\%$ and $29.85\%$,
respectively. Relative to RCNL, SNPD therefore predicts larger gains
for domestic producers and smaller losses for importers. These profit
effects reflect changes in both prices and sales, which depend on
substitution across products and firms' joint pricing decisions,
as well as own-price curvature.

The RCNL consumer-surplus change is $-\$0.050$ per consumer, while the SNPD straight-line demand integral is $-\$0.112$. The RCNL measure is computed using the log-sum formula, averaging
consumer-specific surplus changes over the estimated distribution
of preferences. The SNPD measure
instead integrates demand-weighted price changes along the straight-line path from baseline to counterfactual prices. Both indicate a consumer loss, with a larger magnitude under SNPD,
although a path-independent consumer-surplus interpretation of the latter requires additional integrability conditions.

Similar average elasticities can conceal substantial differences in counterfactual price responses. The SNPD and RCNL specifications imply different substitution patterns, demand curvature, and transmission of cost shocks across products. This application illustrates the economic relevance of the concerns raised by recent work and shows that a flexible nonparametric approach remains feasible with thousands of markets.

\section{Concluding Remarks \label{Section:CCL}}

In this paper, I develop a scalable nonparametric approach to estimating demand for differentiated products with market-level data. The approach builds on mild and economically motivated invariance restrictions, in particular permutation invariance with respect to competitors, to formulate an inverse market share model that can be written as the sum of the logit inverse term and a permutation-invariant function of market share ratios and differences in prices and characteristics relative to competitors. Leveraging a recent approximation result, I approximate this permutation-invariant component with a permutation-invariant polynomial expansion that addresses the curse of dimensionality in the number of products by expressing the function in terms of competitor-pooled rather than competitor-specific arguments. Monte Carlo simulations and an application to the U.S. beer market show that the approach performs well empirically.

Monte Carlo simulations show that the approach closely recovers substitution patterns, demand curvature, and markups across a range of data-generating processes. In the U.S. beer market, it yields average elasticities similar to those of the RCNL model but stronger substitution toward close substitutes, own-price elasticities that vary more at a given price and change nonmonotonically with price, and a wider range of demand curvature, with values both below and above one. These differences matter for counterfactuals: the price response of non-imported products to a tax on imports is much larger. Estimation takes about 30 seconds, more than a thousand times faster than the RCNL model (about 10.5 hours), illustrating that the approach is scalable to applications with thousands of markets and dozens of products.

The paper also points to several directions for future research. One is to address the two limitations discussed above: the difficulty of applying the approach to questions at the individual-consumer level and the remaining curse of dimensionality associated with incorporating many product characteristics into the inverse market share model. Another is to extend the framework to individual-level data and to dynamic and multicategory demand models.

\appendix
\begin{appendices}

\section{Proofs \label{Appendix:Proof}}

This appendix provides the proofs of the propositions stated in the main text. To ease exposition, I first suppress the notation for markets $t$.

\paragraph{Proof of Proposition \ref{prop:Main}.} Fix a market $t$ and suppress the market index. Let
$\mathbf{w}=(\mathbf{p},\mathbf{x}^{(2)})$.

\medskip
\noindent\textit{Step 1 (Market share representation).}
For each fixed $\mathbf{w}$, under Assumptions \ref{Ass1} -- \ref{AssReg},  
\citet[][Theorem 1 and Corollary 1]{monardo2019chap3} implies that there exist a constant $\mu>0$ and a continuously differentiable function $\mathbf{H}(\cdot,\mathbf{w})=\big(H_0(\cdot,\mathbf{w}),\ldots,H_J(\cdot,\mathbf{w})\big):
\mathbb R_{++}^{J+1}\to\mathbb R_{++}^{J+1}$ such that
\[
\sigma_j(\boldsymbol{\delta},\mathbf{w})=\frac{H_j(e^{\boldsymbol{\delta}},\mathbf{w})}{H_\Sigma(e^{\boldsymbol{\delta}},\mathbf{w})},
\qquad
H_\Sigma(e^{\boldsymbol{\delta}},\mathbf{w})=\sum_{m=0}^J H_m(e^{\boldsymbol{\delta}},\mathbf{w}),
\qquad j \in\mathcal{J},
\]
where each component $H_j(\cdot,\mathbf{w})$ is positive, $\mathbf{H}(\cdot,\mathbf{w})$ is homogeneous of
degree $\mu>0$, for any fixed finite values of $\{e^{\delta_k}\}_{k \ne j}$, $\lim_{e^{\delta_j} \to +\infty} H_j (e^{\boldsymbol{\delta}},\mathbf{w}) = + \infty$, and the derivative matrix $\left[\frac{\partial H_i(e^{\boldsymbol{\delta}},\mathbf{w})}{\partial \delta_j}\right]_{0\le i,j\le J}$ is symmetric positive semidefinite. In the remainder I take $\mu=1$. This is without loss for the auxiliary representation:
replace $\mathbf H$ by $\overline{\mathbf H}=H_\Sigma^{1/\mu-1}\mathbf H$.
This leaves market shares unchanged and makes $\overline{\mathbf H}$ homogeneous of degree one.
Symmetry and homogeneity give $\nabla_{\boldsymbol\delta}H_\Sigma=\mu\mathbf H$, and hence $\mathbf J_{\overline{\mathbf H}}^{\boldsymbol\delta}
=H_\Sigma^{1/\mu}\left(\mathbf J_{\boldsymbol\sigma}^{\boldsymbol\delta}
+\boldsymbol\sigma\boldsymbol\sigma'\right)$ is symmetric positive semidefinite. Moreover, $\overline H_j=H_\Sigma^{1/\mu}\sigma_j>0$ and
$\overline H_j\to+\infty$ as $\delta_j\to+\infty$, since $H_\Sigma\to+\infty$ and $\sigma_j\to1$.
Hereafter I write $\mathbf H$ for $\overline{\mathbf H}$.

\medskip
\noindent\textit{Step 2 (Inverse market share).}
The market share equations $\mathbf{s}=\sigma(\boldsymbol{\delta},\mathbf{w})$ imply
\begin{equation}\label{eq:mktshH}
\mathbf{H}(e^{\boldsymbol{\delta}},\mathbf{w})=H_\Sigma(e^{\boldsymbol{\delta}},\mathbf{w})\,\mathbf{s}.
\end{equation}
Let $C(\boldsymbol\delta,\mathbf w)=\ln H_\Sigma(e^{\boldsymbol\delta},\mathbf w)$.
By Step~1, $\nabla_{\boldsymbol\delta}C=\boldsymbol\sigma$.
Assumption~\ref{AssReg} makes $C(0,\boldsymbol\delta_{-0},\mathbf w)$ strictly convex,
since its Hessian is the positive-definite inside-good substitution matrix.
The limit condition in Assumption \ref{AssReg} gives $\partial C/\partial\delta_j=\sigma_j\to1$ as $\delta_j\to+\infty$,
with the other indices fixed, and hence the required boundary growth of $C$.
These conditions and the full index domain allow application of \citet[][Lemma~1]{monardo2019chap3}:
for every $\mathbf s\in\Delta_J^+$, there exists a unique $\boldsymbol\delta$ with $\delta_0=0$
such that $\boldsymbol\sigma(\boldsymbol\delta,\mathbf w)=\mathbf s$.

Then, for each fixed $\mathbf{w}$, I show that $\mathbf{H}(\cdot,\mathbf{w})$ is injective on $\mathbb R_{++}^{J+1}$. Suppose $\mathbf{H}(e^{\boldsymbol{\delta}},\mathbf{w})=\mathbf{H}(e^{\boldsymbol{\delta}'},\mathbf{w})$. Then, $\frac{\mathbf{H}(e^{\boldsymbol{\delta}},\mathbf{w})}{\sum_{m=0}^J H_m(e^{\boldsymbol{\delta}},\mathbf{w})}
=
\frac{\mathbf{H}(e^{\boldsymbol{\delta}'},\mathbf{w})}{\sum_{m=0}^J H_m(e^{\boldsymbol{\delta}'},\mathbf{w})}$, so $\boldsymbol{\sigma}(\boldsymbol{\delta},\mathbf{w})=\boldsymbol{\sigma}(\boldsymbol{\delta}',\mathbf{w})$. So,
$\boldsymbol{\delta}'=\boldsymbol{\delta}+c\mathbf 1$ for some $c\in\mathbb R$. Since $\mathbf{H}(\cdot,\mathbf{w})$ is
homogeneous of degree one, $\mathbf{H}(e^{\boldsymbol{\delta}'},\mathbf{w})
=
\mathbf{H}(e^{\boldsymbol{\delta}+c\mathbf 1},\mathbf{w})
=
\mathbf{H}(e^c e^{\boldsymbol{\delta}},\mathbf{w})
=
e^c \mathbf{H}(e^{\boldsymbol{\delta}},\mathbf{w})$. Because $\mathbf{H}(e^{\boldsymbol{\delta}'},\mathbf{w})=\mathbf{H}(e^{\boldsymbol{\delta}},\mathbf{w})$ and each component of $\mathbf{H}(e^{\boldsymbol{\delta}},\mathbf{w})$
is strictly positive, it follows that $e^c=1$, hence $c=0$. Therefore
$\boldsymbol{\delta}'=\boldsymbol{\delta}$, proving that $\mathbf{H}(\cdot,\mathbf{w})$ is injective.  Let $\mathbf{K}(\cdot,\mathbf{w})$
denote its inverse on its range.

I next show that the range of $\mathbf{H}(\cdot,\mathbf{w})$ is
$\mathbb R_{++}^{J+1}$. Let $\mathbf{y}\in\mathbb R_{++}^{J+1}$ be arbitrary, and define $\mathbf{s}=\mathbf{y}/\sum_{m=0}^J y_m$. Then $\mathbf{s}$ lies in the interior of the share simplex. For every such interior share
vector, \citet*[][Theorem 1 and Corollary 1]{monardo2019chap3} implies that there exists a unique normalized
index vector $\boldsymbol{\delta}$ such that $\boldsymbol{\sigma}(\boldsymbol{\delta},\mathbf{w})=\mathbf{s}$, and thus $\mathbf{H}(e^{\boldsymbol{\delta}},\mathbf{w})=H_\Sigma(e^{\boldsymbol{\delta}},\mathbf{w})\,\mathbf{s}$

Let $\lambda=H_\Sigma(e^{\boldsymbol{\delta}},\mathbf{w})>0$. Since $\mathbf{H}(\cdot,\mathbf{w})$ is
homogeneous of degree one, for every $c>0$, $\mathbf{H}(c e^{\boldsymbol{\delta}},\mathbf{w})=c \mathbf{H}(e^{\boldsymbol{\delta}},\mathbf{w})=c \lambda \mathbf{s}$. Choosing $c=\sum_{m=0}^J y_m/\lambda$ gives $\mathbf{H}(c e^{\boldsymbol{\delta}},\mathbf{w})
=
\left(\sum_{m=0}^J y_m\right)\mathbf{s}
=
\mathbf{y}$. Since $\mathbf{y}$ was arbitrary, $\operatorname{range}(\mathbf{H}(\cdot,\mathbf{w}))=\mathbb R_{++}^{J+1}$. Combined with injectivity, this implies that $\mathbf{H}(\cdot,\mathbf{w})$ is bijective on
$\mathbb R_{++}^{J+1}$, so we can write $\mathbf{K}(\cdot,\mathbf{w})=\mathbf{H}(\cdot,\mathbf{w})^{-1}$, and $\mathbf{K}(\mathbf{H}(e^{\boldsymbol{\delta}},\mathbf{w}),\mathbf{w})=e^{\boldsymbol{\delta}}$.

The inverse $\mathbf K(\cdot,\mathbf w)$ is continuously differentiable
on $\mathbb R^{J+1}_{++}$. Let $\mathbf A
=
\mathbf J_{\mathbf H}^{\boldsymbol\delta}
(e^{\boldsymbol\delta},\mathbf w)$.
Using
$\mathbf H=H_\Sigma\boldsymbol\sigma$, the symmetry of $\mathbf A$,
and the homogeneity of $\mathbf H$, we obtain $\mathbf A
=
H_\Sigma
\left(
\mathbf J_{\boldsymbol\sigma}^{\boldsymbol\delta}
+
\boldsymbol\sigma\boldsymbol\sigma'
\right)$. Furthermore, translation invariance implies that a common change in all indices does
not change market shares. Because
$\mathbf J_{\boldsymbol\sigma}^{\boldsymbol\delta}$ is positive semidefinite
and its inside-good principal block is positive definite, this common
change is the only direction in which shares remain locally unchanged.
Consequently, if $\mathbf v'
\mathbf J_{\boldsymbol\sigma}^{\boldsymbol\delta}
\mathbf v=0$, then $\mathbf v=c\mathbf1$ for some $c\in\mathbb R$. It follows that $\mathbf A$ is positive definite. Indeed, for every
$\mathbf v\neq\mathbf 0$, $
\mathbf v'\mathbf A\mathbf v
=
H_\Sigma
\left[
\mathbf v'
\mathbf J_{\boldsymbol\sigma}^{\boldsymbol\delta}
\mathbf v
+
(\mathbf v'\boldsymbol\sigma)^2
\right]
>0$. Finally, writing $\mathbf z=e^{\boldsymbol\delta}$, $\mathbf J_{\mathbf H}^{\mathbf z}(\mathbf z,\mathbf w)
=
\mathbf A\,\operatorname{diag}(\mathbf z)^{-1}$,
which is nonsingular. Since $\mathbf H(\cdot,\mathbf w)$ is a bijection,
the inverse-function theorem implies that
$\mathbf K(\cdot,\mathbf w)=\mathbf H(\cdot,\mathbf w)^{-1}$
is continuously differentiable on $\mathbb R^{J+1}_{++}$.

Because $\mathbf{H}(\cdot,\mathbf{w})$ is homogeneous of degree one, its inverse $\mathbf{K}(\cdot,\mathbf{w})$ is
 homogeneous of degree one: for every $z\in\mathbb{R}_{++}^{J+1}$
and every $\lambda>0$, $\mathbf{K}(\lambda \mathbf{z},\mathbf{w})=\lambda \mathbf{K}(\mathbf{z},\mathbf{w})$. Equivalently, each component $\log K_j(\cdot,\mathbf{w})$ is logarithmically homogeneous. Applying $\mathbf{K}(\cdot,\mathbf{w})$ to \eqref{eq:mktshH} and using its homogeneity yields $e^\delta
=
K\!\big(H_\Sigma(e^\delta,\mathbf{w})\,\mathbf{s},\mathbf{w}\big)
=
H_\Sigma(e^\delta,\mathbf{w})\,K(\mathbf{s},\mathbf{w})$. Taking logs componentwise gives
\[
\delta_j=\ln H_\Sigma(e^\delta,\mathbf{w})+\ln K_j(s,\mathbf{w}),
\qquad j\in\mathcal J.
\]
Using the normalization $\delta_0=0$ and rearranging terms, the equation for $j=0$ becomes $\ln H_\Sigma(e^\delta,\mathbf{w})=-\ln K_0(s,\mathbf{w})$, which, when substituted back, yields
\begin{equation}\label{eq:inverse_pre_ratio}
\delta_j=\ln K_j(s,\mathbf{w})-\ln K_0(\mathbf{s},\mathbf{w}),
\qquad j\in\mathcal J_0.
\end{equation}

Although $\mathbf K=\mathbf H^{-1}$ is uniquely determined by the
chosen function $\mathbf H$, \eqref{eq:inverse_pre_ratio}
depends on $\mathbf K$ only through the ratios $K_j/K_0$. Define the
auxiliary normalized function $\widetilde K_j(\mathbf s,\mathbf w)
=
\frac{s_0}{K_0(\mathbf s,\mathbf w)}
K_j(\mathbf s,\mathbf w)$ for all $j\in\mathcal{J}$. Then $\widetilde K_0(\mathbf s,\mathbf w)=s_0$, and
$
\widetilde K_j(\mathbf s,\mathbf w)/\widetilde K_0(\mathbf s,\mathbf w)
=
K_j(\mathbf s,\mathbf w)/
K_0(\mathbf s,\mathbf w)$. Therefore this normalization leaves
\eqref{eq:inverse_pre_ratio} unchanged and is without loss of
generality for the inverse market share representation. The normalized function
$\widetilde{\mathbf K}$ is used only as an auxiliary representation;
it is not asserted to be the inverse of $\mathbf H$.

Under this normalization,
\[
\delta_j=\ln \widetilde K_j(\mathbf{s},\mathbf{w})-\ln s_0,
\qquad j\in\mathcal J_0.
\]
Since $K(\cdot,\mathbf{w})$ is homogeneous of degree one, $K_0(\cdot,\mathbf{w})$ is also homogeneous
of degree one. Hence the scaling factor $s_0/K_0(\mathbf{s},\mathbf{w})$ is homogeneous of degree
zero in $s$, and each $\widetilde K_j(\cdot,\mathbf{w})$ remains homogeneous of degree one.
It follows that
\[
\ln \widetilde K_j(\mathbf{s},\mathbf{w})-\ln s_j=\bigl(\ln K_j(\mathbf s,\mathbf w)-\ln s_j\bigr)
-\bigl(\ln K_0(\mathbf s,\mathbf w)-\ln s_0\bigr)
\]
is homogeneous of degree zero in $\mathbf{s}$. Taking $s_j$ as numeraire, this term depends
on $\mathbf{s}$ only through the full collection of share ratios relative to product $j$,
namely $\{r_{jk}=s_k/s_j\}_{k\in\mathcal J\setminus\{j\}}$. Thus, the nonparametric term incorporates the subtraction of the outside good component.

Because $\boldsymbol\sigma$ is continuously differentiable jointly in
$(\boldsymbol\delta,\mathbf w)$ and its normalized inside good
Jacobian is everywhere nonsingular, the inverse-function theorem,
applied while treating $\mathbf w$ as a parameter, implies that the
normalized inverse is continuously differentiable jointly in
$(\mathbf s,\mathbf w)$. Therefore, there exists a continuously differentiable function $\widetilde g_j$ such that
\begin{equation}\label{eq:inverse_with_outside_ratio_clean}
\delta_j
=
\ln\!\left(\frac{s_j}{s_0}\right)
+
\widetilde g_j\!\left(\{r_{jk}\}_{k\in\mathcal J\setminus\{j\}},
\mathbf{w}
\right),
\qquad j\in\mathcal J.
\end{equation}

\noindent\textit{Step 3 (Relative objects and common functional form).} Fix $j\in\mathcal J_0$. Since $p_0=0$ and $\mathbf x_0^{(2)}=\mathbf 0$, the relative objects recover the levels: $p_j=-d^p_{j0}$ and $p_k=d^p_{jk}-d^p_{j0}$, while $\mathbf x_j^{(2)}=-\mathbf d_{j0}^{(2)}$ and $\mathbf x_k^{(2)}=\mathbf d_{jk}^{(2)}-\mathbf d_{j0}^{(2)}$ for $k\ne j$. Substituting them for $\mathbf w=(\mathbf p,\mathbf x^{(2)})$ in
\eqref{eq:inverse_with_outside_ratio_clean} expresses $\widetilde g_j$ as a function of the
relative objects, keeping those associated with $0$ separate.
Since these objects recover all price and characteristic levels on the normalized domain,
this change of arguments requires no translation invariance in $\mathbf p$ or $\mathbf x^{(2)}$.
Dependence on levels is retained through $d^p_{j0}$ and $\mathbf d_{j0}^{(2)}$.

Next, Assumptions \ref{Ass5}.\ref{Ass5.1} and \ref{Ass5}.\ref{Ass5.2} imply that relabeling
the inside products and their characteristics relabels the corresponding market shares.
Since this relabeling preserves $\delta_0=0$, uniqueness of the normalized inverse
implies that the inverse equations have the same functional form for every $j\in\mathcal J_0$.
Relabeling the inside competitors in $\mathcal J_0\setminus\{j\}$ leaves the correction
in \eqref{eq:inverse_with_outside_ratio_clean} unchanged. Therefore there exists a single
continuously differentiable function $G$, invariant under permutations of the inside competitors, such that
\[
\begin{aligned}
\delta_j=\ln\!\left(\frac{s_j}{s_0}\right)
+G\!\left(r_{j0},d^p_{j0},\mathbf d^{(2)}_{j0};
\{r_{jk},d^p_{jk},\mathbf d^{(2)}_{jk}\}_{k\in\mathcal J_0\setminus\{j\}}\right),
\quad j\in\mathcal J_0.
\end{aligned}
\]
Moreover, $s_j=(1+\sum_{k\ne j}r_{jk})^{-1}$ and $s_k=r_{jk}s_j$.
Together with the covariate identities above, this shows that the relative objects
determine $(\mathbf s,\mathbf w)$ and hence the value of $G$. Because the normalized inverse is unique (Step~2), $G=\delta_j-\ln(s_j/s_0)$ is uniquely determined at every admissible value of its arguments, which proves the uniqueness claim.

Suppose now that $p_k>0$ for every $k\in\mathcal J_0$.
For a fixed inside product $j$, write $v_{jk}=(r_{jk},d^p_{jk},\mathbf d^{(2)}_{jk})$.
Since $p_0=0$, $d^p_{j0}=-p_j<p_k-p_j=d^p_{jk}$, $k\in\mathcal J_0\setminus\{j\}$. Thus the outside block is the unique block with the smallest price difference.
Let $V$ be any ordering of an admissible collection of competitor blocks, and let $v_*(V)$ denote this block. Define $\widetilde G(V)=G\!\left(v_*(V);\,V\setminus\{v_*(V)\}\right)$, where the remaining blocks retain their order.
Permuting $V$ leaves the selected block unchanged and only reorders the remaining inside blocks.
Inside-competitor invariance of $G$ therefore implies that $\widetilde G$ is invariant under every permutation of the competitor blocks.
Since $v_*(V)=v_{j0}$, this representation gives the same value as the $G$ constructed above.
The strict inequality makes the selected block locally unchanged, so $\widetilde G$ is continuously differentiable on the domain obtained by permuting admissible blocks.
This proves the additional conclusion.

\medskip

\noindent\textit{Step 4 (Jacobian statement).} Fix $\mathbf w$ and let $\mathbf B(\boldsymbol\delta,\mathbf w)
=\left[\frac{\partial\sigma_i(\boldsymbol\delta,\mathbf w)}{\partial\delta_j}\right]_{1\leq i,j\leq J}$. By Assumption~\ref{AssReg}, $\mathbf B$ is symmetric positive definite.
Differentiating the normalized inverse established in Step~2 gives
$\partial\boldsymbol\delta_{-0}/\partial\mathbf s_{-0}=\mathbf B^{-1}$,
where $\boldsymbol\delta_{-0}=(\delta_1,\ldots,\delta_J)'$,
$\mathbf s_{-0}=(s_1,\ldots,s_J)'$, and $\mathbf B$ is evaluated at that inverse.
Differentiating the representation in Step~3, with $s_0=1-\sum_{j=1}^J s_j$, therefore yields
\[
\operatorname{diag}(\mathbf s_{-0})^{-1}
+\frac{1}{s_0}\mathbf1\mathbf1'
+\mathbf J_G(\mathbf s_{-0})
=\mathbf B^{-1}\succ0,
\]
where $\mathbf J_G$ is the Jacobian of the vector of corrections $(G_1,\ldots,G_J)'$
with respect to the inside shares, holding covariates fixed.
This proves the Jacobian statement.

Finally,
$\delta_j=\mathbf x_j\boldsymbol\beta-\alpha p_j+\xi_j$, so the proposition
follows.

\paragraph{Proof of Proposition \ref{PropBH}.} The argument follows \citet*[][proof of Theorem 1]{berry2014identification}. Fix $j \in \mathcal{J}_0$. Suppose that both $(\alpha,\boldsymbol{\beta},G)$ and $(\alpha',\boldsymbol{\beta}',G')$ satisfy the conditional moment restriction implied by \eqref{eq:PLIV} and Assumption \ref{Ass3}.\ref{Ass3.1}:
\begin{align}
\mathbb{E}\!\left[
\ln\!\left(\frac{s_{jt}}{s_{0t}}\right)
+ G(\mathbf{V}_{jt})
- \mathbf{x}_{jt}\boldsymbol{\beta}
+ \alpha p_{jt}
\;\Bigm|\;
\mathbf{x}_t,\mathbf{z}_t
\right]
&= 0
\quad \text{a.s.},
\label{eq:BH1}
\\[1ex]
\mathbb{E}\!\left[
\ln\!\left(\frac{s_{jt}}{s_{0t}}\right)
+ G'(\mathbf{V}_{jt})
- \mathbf{x}_{jt}\boldsymbol{\beta}'
+ \alpha' p_{jt}
\;\Bigm|\;
\mathbf{x}_t,\mathbf{z}_t
\right]
&= 0
\quad \text{a.s.}
\label{eq:BH2}
\end{align}

Subtracting \eqref{eq:BH2} from \eqref{eq:BH1} yields
\begin{equation}
\mathbb{E}\!\left[B(\mathbf{V}_{jt}) \mid \mathbf{x}_t,\mathbf{z}_t\right]
=
\mathbf{x}_{jt}\Delta\boldsymbol{\beta}
-
\Delta\alpha\,\mathbb{E}\!\left[p_{jt}\mid \mathbf{x}_t,\mathbf{z}_t\right],
\quad \text{a.s.},
\label{eq:BH3}
\end{equation}
where $B(\mathbf{V}_{jt}) = G'(\mathbf{V}_{jt}) - G(\mathbf{V}_{jt})$, $\Delta\boldsymbol{\beta} = \boldsymbol{\beta}' - \boldsymbol{\beta}$, and
$\Delta\alpha = \alpha' - \alpha$. By the separation condition in Assumption \ref{Ass3}.\ref{Ass3.3}, \eqref{eq:BH3} implies that $\Delta\alpha = 0$ and $\Delta\boldsymbol{\beta} = \mathbf{0}$. Hence $\alpha = \alpha'$ and $\boldsymbol{\beta} = \boldsymbol{\beta}'$.

Substituting $\Delta\alpha = 0$ and $\Delta\boldsymbol{\beta} = \mathbf{0}$ back into \eqref{eq:BH3} gives $\mathbb{E}\!\left[B(\mathbf{V}_{jt}) \mid \mathbf{x}_t,\mathbf{z}_t\right] = 0$ a.s. The completeness condition in Assumption \ref{Ass3}.\ref{Ass3.2} then implies that $B(\mathbf{V}_{jt}) = 0$ a.s., so $G(\mathbf{V}_{jt}) = G'(\mathbf{V}_{jt})$ a.s. on the support of $\mathbf{V}_{jt}$. Thus, $(\alpha,\boldsymbol{\beta})$ are point identified, and any two admissible functions $G$ and $G'$ that satisfy the conditional moment restriction obey
$G(\mathbf{V}_{jt})=G'(\mathbf{V}_{jt})$ almost surely. Because $G$ and $G'$ are continuous, this almost-sure equality extends to every admissible point in the support of $\mathbf{V}_{jt}$. Thus, $G$ is point identified on its support. Because this argument holds for every $j\in\mathcal{J}_0$, the inverse market share function $\boldsymbol{\sigma}^{-1}$ is point identified on the support of $(\mathbf{s}_t,\mathbf{p}_t,\mathbf{x}_t^{(2)})$. Moreover, $\xi_{jt}=\ln(s_{jt}/s_{0t})+G(\mathbf V_{jt})
-\mathbf x_{jt}\boldsymbol\beta+\alpha p_{jt}$ is identified almost surely. Thus all arguments of $\boldsymbol\sigma(\mathcal X_t)$ are identified, and the observed shares identify $\boldsymbol\sigma$ on $\mathcal X$ by continuity.

\paragraph{Proof of Proposition \ref{Prop:Approx}.} 

Let $\mathcal C_{\mathrm{sym}}(\Omega^N)$ denote the space of continuous permutation-invariant functions on $\Omega^N$. Extend $\widetilde G$ continuously from $\mathcal{\widetilde{V}}$ to $\Omega^N$ and average this extension over permutations of the competitor blocks. Since $\mathcal{\widetilde{V}}$ and $\widetilde G$ are permutation invariant, the averaged extension agrees with $\widetilde G$ on $\mathcal{\widetilde{V}}$. The result then follows from \citet*{bachmayr2024polynomial}. Their Proposition 3.2 shows that the span of the symmetrized tensor-product basis is dense in $\mathcal C_{\mathrm{sym}}(\Omega^N)$, and their Proposition 3.3 shows that this span coincides with the span of the pooled basis functions $A_{\mathbf u}$. Hence, for every $\varepsilon>0$, there exist $D$ and $\widetilde G_D\in\mathcal A_D$ such that $\sup_{\widetilde{\mathbf V}\in\mathcal{\widetilde{V}}}|\widetilde G(\widetilde{\mathbf V})-\widetilde G_D(\widetilde{\mathbf V})|<\varepsilon$.

\end{appendices}

\end{bibunit}

\newpage

\clearpage
\pagenumbering{arabic}
\setcounter{page}{1}
\setcounter{section}{0}
\setcounter{equation}{0}
\renewcommand{\theequation}{A.\arabic{equation}}

\begin{bibunit}

\begin{center}
    {\LARGE \bfseries Supplemental Appendix\par}
    \vspace{1.5em}
\end{center}

\thispagestyle{plain}

\section{Relation to Existing Models \label{Appendix:Connection}}

This supplemental appendix relates the inverse market share representation in Proposition \ref*{prop:Main} in the main text to three leading demand models and clarifies when the symmetric sieve delivers dimensionality reduction. This gain arises when the unknown component is symmetric in competitor-specific variables.

\paragraph{Semiparametric logit.}
\citet[][Appendix D.6]{compiani2018} considers the semiparametric logit model
\begin{equation}
\sigma_j(\boldsymbol{\delta}_t,\textbf{p}_t,\textbf{x}_t^{(2)})
=
\frac{\exp\!\big(\delta_{jt}+G(p_{jt},\textbf{x}_{jt}^{(2)})\big)}
{1+\sum_{k=1}^{J}\exp\!\big(\delta_{kt}+G(p_{kt},\textbf{x}_{kt}^{(2)})\big)} .
\label{eq:semi_logit_share}
\end{equation}
Inverting \eqref{eq:semi_logit_share} yields
\begin{equation}
\xi_{jt}
=
\ln\!\left(\frac{s_{jt}}{s_{0t}}\right)
-\textbf{x}_{jt}^{(1)}\boldsymbol{\beta}
-G(p_{jt},\textbf{x}_{jt}^{(2)}),
\label{eq:semi_logit_inverse}
\end{equation}
where the unknown component $G$ depends only on own-product observables. The symmetric sieve therefore provides no additional dimensionality reduction: there are no competitor-specific arguments to pool or symmetrize.

\paragraph{Universal logit model.}
The universal logit model of \citet*{mcfadden1984econometric} generalizes the standard and semiparametric logit models by allowing the utility of product $j$ to depend on both own-product variables and competitors' variables:
\begin{equation*}
\sigma_j(\boldsymbol{\delta}_t,\textbf{p}_t,\textbf{x}_t^{(2)})
=
\frac{\exp\!\Big(\delta_{jt}+G_1(p_{jt},\textbf{x}_{jt}^{(2)})
+G_2\big(\{p_{kt},\textbf{x}_{kt}^{(2)}\}_{k\in\mathcal J\setminus\{0,j\}}\big)\Big)}
{1+\sum_{k=1}^{J}\exp\!\Big(\delta_{kt}+G_1(p_{kt},\textbf{x}_{kt}^{(2)})
+G_2\big(\{p_{\ell t},\textbf{x}_{\ell t}^{(2)}\}_{\ell\in\mathcal J\setminus\{0,k\}}\big)\Big)} .
\label{eq:universal_logit_share}
\end{equation*}
Its inverse representation is
\begin{equation*}
\xi_{jt}
=
\ln\!\left(\frac{s_{jt}}{s_{0t}}\right)
-\textbf{x}_{jt}\boldsymbol{\beta}
+\alpha p_{jt}
-G_1(p_{jt},\textbf{x}_{jt}^{(2)})
-G_2\big(\{p_{kt},\textbf{x}_{kt}^{(2)}\}_{k\in\mathcal J\setminus\{0,j\}}\big) .
\label{eq:universal_logit_inverse}
\end{equation*}

The function $G_1$ captures how own-product variables affect product $j$'s market share, whereas $G_2$ captures how competitors' variables shape substitution patterns. If the market share function is permutation-invariant with respect to competitors, then $G_2$ must be permutation-invariant in its competitor-specific arguments. The symmetric sieve then approximates $G_2$ parsimoniously and imposes symmetry.

\paragraph{RCL model.} \citet*{Nevo2026} and \citet*{lu2023semi} propose two-step procedures for estimating the RCL demand. In both procedures, the first step uses inverse market share equations, and the second recovers the distribution of random coefficients nonparametrically.

In \citet*{lu2023semi}, separability between the product index and the random-coefficient component yields the inverse representation in \eqref{eq:semi_logit_inverse}. Because the unknown component again depends only on own-product variables, the symmetric sieve provides no additional dimensionality reduction.

By contrast, \citet*{Nevo2026} derive inverse market share equations of the form
\begin{equation}
\xi_{jt}
=
\ln\!\left(\frac{s_{jt}}{s_{0t}}\right)
-\textbf{x}_{jt}\boldsymbol{\beta}
+\alpha p_{jt}
-G\Big(\{s_{kt},d^p_{jk,t},\textbf{d}_{jk,t}^{(2)},1\{k=0\}\}_{k\in\mathcal J\setminus\{j\}}\Big).
\label{eq:rcl_inverse}
\end{equation}

In \eqref{eq:rcl_inverse}, $G$ is permutation invariant in the competitor blocks, which include the outside good distinguished by the indicator $1\{k=0\}$; the own share is not a separate argument \citep*[][Corollary~1]{Nevo2026}. The symmetric sieve applies once the indicator is added as a coordinate of each block. Unlike the inverse market share representation of this paper under the outside good exclusion, this representation keeps the outside good among the arguments of $G$.

\section{Alternative Approximation and Models}

This supplemental appendix first presents an alternative approximation result for permutation-invariant functions from \citet*{bachmayr2024polynomial} and then introduces an inverse market share model that can be viewed as a flexible parametric version of the nonparametric model. \citet*{chaves2025misallocation} build on both to estimate gasoline demand.

\subsection{Body-Order Approximation \label{Appendix:BodyApprox}}

The following result from \citet*{bachmayr2024polynomial} establishes that continuous permutation-invariant functions can be approximated by body-order expansions.

\begin{proposition}
Let $f: \Omega^n \to \mathbb{R}$ be a continuous, permutation-invariant function, where $\Omega$ is a compact subset of $\mathbb{R}^d$. Then $f$ can be uniformly approximated by a body-order expansion:
\[
f(\mathbf{x}_1, \ldots, \mathbf{x}_n)
\approx
\sum_{k=0}^N \sum_{1 \le j_1 < \cdots < j_k \le n}
V_k(\mathbf{x}_{j_1}, \ldots, \mathbf{x}_{j_k}),
\]
where $N\leq n$ and, for each $k=0,\ldots,N$, $V_k:\Omega^k\to\mathbb{R}$ is a permutation-invariant $k$-body function that can itself be uniformly approximated by a polynomial $V_{k,D_k}$ of degree $D_k$.

Furthermore, for any $\epsilon > 0$, there exist $N$ and $\{D_k\}_{k=1}^N$ such that
\[
f_{ND}(\mathbf{x}_1, \ldots, \mathbf{x}_n)
=
\sum_{k=0}^N \sum_{1 \le j_1 < \cdots < j_k \le n}
V_{k,D_k}(\mathbf{x}_{j_1}, \ldots, \mathbf{x}_{j_k})
\]
satisfies
\[
\sup_{(\mathbf{x}_1, \ldots, \mathbf{x}_n) \in \Omega^n}
\left|f(\mathbf{x}_1, \ldots, \mathbf{x}_n) - f_{ND}(\mathbf{x}_1, \ldots, \mathbf{x}_n)\right|
< \epsilon.
\]
\end{proposition}

\begin{proof}
    Take $N=n$, $V_n=f$, and $V_k=0$ for $k<n$. Proposition~\ref*{Prop:Approx} in the main text, which rests on Propositions~3.2 and~3.3 of \citet*{bachmayr2024polynomial}, then yields a symmetric polynomial approximation of $V_n$. Section~4 of \citet*{bachmayr2024polynomial} studies when much smaller truncation levels $N<n$ suffice, under an additional assumption on $f$ (their Assumption~4.2).
\end{proof}

This proposition shows that any continuous permutation-invariant function of $n$ variables can be approximated by a sum of permutation-invariant $k$-body components, each depending on at most $k\leq N$ variables, where $N\leq n$. These components capture interactions within subsets of variables, and each can in turn be approximated by a polynomial of degree $D_k$. Therefore, choosing the truncation level $N$ and the polynomial degrees $\{D_k\}_{k=1}^N$ appropriately yields an arbitrarily accurate polynomial approximation.

\subsection{A Flexible Parametric Inverse Market Share Model \label{Sec:FPD}}

This subsection relates the inverse market share representation in Proposition \ref*{prop:Main} in the main text to models of nested substitution and develops a flexible parametric specification of the nonparametric model.

The nested logit model partitions products into nests. Following \citet*{berry1994estimating}, its inverse market share function for a product $j$ in nest $g$ is
\begin{equation*}
\sigma_j^{-1}(\mathbf{s}_t;\boldsymbol{\mu})
=
\ln\!\left(\frac{s_{jt}}{s_{0t}}\right)
-\mu \ln\!\left(\frac{s_{jt}}{\sum_{k\in g}s_{kt}}\right).
\label{eq:nested_logit_inverse}
\end{equation*}

The nested logit restricts substitution through a prespecified nesting structure: products in the same nest must be closer substitutes than products in different nests.

One way to relax this restriction is to assign a separate nest to each product pair and a singleton nest to the outside good. The resulting inverse market share function is
\begin{equation}
\sigma_j^{-1}(\mathbf{s}_t;\boldsymbol{\mu})
=
\ln\!\left(\frac{s_{jt}}{s_{0t}}\right)
-\sum_{i\in \mathcal{J}_0\backslash\{j\}}\mu_{ij}
\ln\!\left(\frac{s_{jt}}{s_{it}+s_{jt}}\right),
\qquad j\in\mathcal{J}_0,
\label{eq:flexible_inverse_logit}
\end{equation}
where the pair-specific nesting parameter $\mu_{ij}$ governs substitution between products $i$ and $j$, larger values indicating closer substitutes.

Unlike the standard nested logit model, this specification does not require the modeler to prespecify a small number of market segments. The pair-specific terms instead fully parameterize the matrix of inverse market share derivatives, accommodating rich substitution patterns without coarse segmentation. When several product characteristics drive substitution, specifying $\mu_{ij}$ as a function of the proximity between products $i$ and $j$ maps the nesting parameters into characteristic space. Substitution can then vary with product proximity without prespecifying which dimensions of that space matter most. Consequently, \eqref{eq:flexible_inverse_logit} depends only on relative objects.

The nonparametric component $G(\mathbf{V}_{jt})$ in the main text generalizes this flexible-nesting logic. The vector $\mathbf{V}_{jt}$ collects relative objects between product $j$ and its competitors, and the symmetric sieve estimates how those objects shape substitution. The sieve thereby preserves the central intuition that substitution is governed by product-by-product comparisons without imposing a specific functional form or a small number of nests.

\section{Identification and Properties of the GMM Estimator \label{Appendix:ProofUniqueExist}}

This section addresses three distinct questions. The first part states conditions under which population moments identify the coefficients of a correctly specified polynomial family, linking identification to the specification used in estimation. The second establishes conditions for existence and uniqueness of the constrained sample GMM solution, ensuring that the vector of coefficient estimates is well defined. The third gives sufficient conditions for coefficient consistency and uniform convergence of the estimated $G$ to the true $G$ as the sample size and polynomial degree increase. These additional conditions are needed to establish convergence to the structural parameters; identification at a fixed degree and uniqueness of the sample solution do not suffice. 

\subsection{Identification of the Polynomial Specification}\label{Appendix:PolynomialIdentification}

\paragraph{Completeness without characteristics $\mathbf{x}^{(1)}_t$.} 
Suppose $J\geq2$ and $\mathbf x_t^{(1)}$ is absent. Under the outside good exclusion \eqref{eq:OutsideExclusion}, subtracting \eqref{eq:PLIV} for $j$ from the equation for $k\ne j$ gives
\[
B_{jk}(\mathbf V_{jt})
=\ln r_{jk,t}+G(\mathbf V_{kt})-G(\mathbf V_{jt})
-\mathbf d^{(2)}_{jk,t}\boldsymbol\beta+\alpha d^p_{jk,t}
=\xi_{kt}-\xi_{jt}.
\]
To see why the left-hand side is a function of $\mathbf V_{jt}$ alone, set $r_{jj,t}=1$, $d^p_{jj,t}=0$, and $\mathbf d^{(2)}_{jj,t}=\mathbf0$. For every inside product $\ell\ne k$,
\[
r_{k\ell,t}=\frac{r_{j\ell,t}}{r_{jk,t}},
\qquad
d^p_{k\ell,t}=d^p_{j\ell,t}-d^p_{jk,t},
\qquad
\mathbf d^{(2)}_{k\ell,t}
=\mathbf d^{(2)}_{j\ell,t}-\mathbf d^{(2)}_{jk,t}.
\]
These identities recover all the relative objects in $\mathbf V_{kt}$ from $\mathbf V_{jt}$, including the block for product $j$. Positive shares make the ratios well defined. At the true $(\alpha,\boldsymbol\beta,G)$, $B_{jk}$ is therefore a function of $\mathbf V_{jt}$. Exclusion gives
\[
\mathbb E[B_{jk}(\mathbf V_{jt})\mid\mathbf x_t,\mathbf z_t]
=\mathbb E[\xi_{kt}-\xi_{jt}\mid\mathbf x_t,\mathbf z_t]=0.
\]
Assumption~\ref*{Ass3}.\ref*{Ass3.2} in the main text applies to every integrable function of $\mathbf V_{jt}$, so it implies $B_{jk}(\mathbf V_{jt})=0$ almost surely. Hence $\xi_{kt}=\xi_{jt}$ almost surely for all inside $j,k$.

Under the model and Assumptions~\ref*{Ass3}.\ref*{Ass3.1} -- \ref*{Ass3}.\ref*{Ass3.2} in the main text, the absence of $\mathbf x_t^{(1)}$ rules out product-specific shocks, but allows a common market shock $\xi_{jt}=u_t$ with $\mathbb E[u_t\mid\mathbf x_t,\mathbf z_t]=0$. With $\mathbf x_t^{(1)}$ present, the shock-difference expression additionally contains $(\mathbf x_{kt}^{(1)}-\mathbf x_{jt}^{(1)})\boldsymbol\beta_1$, which need not have conditional mean zero or be recoverable from $\mathbf V_{jt}$. Including these characteristics does not guarantee completeness. If they contribute equally to all inside product indices, shocks must still be equal within each market.

\paragraph{Identification at a fixed degree.}

Proposition~\ref*{PropBH} in the main text establishes nonparametric identification within the admissible function class. The fixed-degree result below identifies the coefficients of the chosen polynomial family without requiring completeness or separation in Assumption~\ref*{Ass3} in the main text. Sample uniqueness ensures a well-defined estimate, but neither it nor fixed-degree identification establishes nonparametric identification.

Fix the degree $D$ and write $G(\mathbf V)=\boldsymbol\phi_D(\mathbf V)'\boldsymbol\gamma_D$, where $\boldsymbol\phi_D$ is a common, permutation-invariant basis. Let $\boldsymbol\theta=(\alpha,\boldsymbol\beta',\boldsymbol\gamma_D')'$ and maintain the constant normalization above. Stack the $N=JT$ product-market observations. Let $\mathbf Y$ collect $\ln(s_{jt}/s_{0t})$ and $\boldsymbol\xi$ collect $\xi_{jt}$. Define the matrices $\mathbf X\in\mathbb R^{N\times k_X}$ and $\mathbf Z\in\mathbb R^{N\times k_Z}$, with $k_Z\geq k_X$, whose corresponding rows are $\bigl(-p_{jt},\mathbf x_{jt},-\boldsymbol\phi_D(\mathbf V_{jt})'\bigr)$ and $\boldsymbol\psi_{jt}'$. Here, $\boldsymbol\psi_{jt}$ contains instrument functions of $(\mathbf x_t,\mathbf z_t)$; product-specific columns allow equation-specific moments. \eqref{eq:PLIV} then gives $\mathbf Y=\mathbf X\boldsymbol\theta+\boldsymbol\xi$, with sample moments $\mathbf g_T(\boldsymbol\theta)=N^{-1}\mathbf Z'(\mathbf Y-\mathbf X\boldsymbol\theta)$. Let $\mathbf W$ denote the GMM weighting matrix and $\Theta$ the admissible set defined by the shape restrictions described in the main text and below. Under correct specification at degree $D$, finite moments satisfying $\mathbb E[\mathbf Z'\boldsymbol\xi/N]=\mathbf0$ and $\operatorname{rank}\!\left(\mathbb E[\mathbf Z'\mathbf X/N]\right)=k_X$ identify $(\alpha,\boldsymbol\beta,\boldsymbol\gamma_D)$ jointly. This is standard linear GMM identification \citep[Section~2.2.3]{newey1994large}. If the true $G$ belongs to the chosen polynomial family, full rank identifies its coefficients, even without $\mathbf x^{(1)}$. Otherwise, approximation error need not be orthogonal to instruments, so GMM need not recover the structural coefficients.

\subsection{Existence and Uniqueness of a Solution to the Constrained GMM Problem}\label{Appendix:ConstrainedGMMUniqueness}

The following proposition instead concerns the sample minimization problem.

\begin{proposition}[Existence and uniqueness of a solution to the constrained GMM problem]\label{Prop:UniqueSol}
Suppose that $\mathbf{Z}'\mathbf{X}$ has full column rank $k_X$ and that $\mathbf{W}$ is symmetric positive definite. Then the constrained GMM minimization problem admits a unique solution in $\boldsymbol{\theta}$.
\end{proposition}

\begin{proof}
The sample moment vector is affine in $\boldsymbol{\theta}$, so, after multiplying by $(JT)^2$, the GMM criterion can be written as
$$
Q(\boldsymbol{\theta})
=
\frac{1}{2}\boldsymbol{\theta}'\mathbf{h}\boldsymbol{\theta}
+
\mathbf{f}'\boldsymbol{\theta}
+
c,
$$
where
$$
\mathbf{h}
=
2\,\mathbf{X}'\mathbf{Z}\mathbf{W}\mathbf{Z}'\mathbf{X},
\qquad
\mathbf{f}
=
-2\,\mathbf{X}'\mathbf{Z}\mathbf{W}\mathbf{Z}'\mathbf{Y},
\qquad
c
=
\mathbf{Y}'\mathbf{Z}\mathbf{W}\mathbf{Z}'\mathbf{Y}.
$$

By assumption, $\mathbf{Z}'\mathbf{X}$ has full column rank $k_X$ and $\mathbf{W}$ is symmetric positive definite. Hence $\mathbf{h}$ is positive definite \citep[][Observation 7.1.8. (b), p. 431]{horn2013matrix}. It follows that $Q$ is strictly convex and coercive.\footnote{Because $\mathbf{h}$ is symmetric positive definite, there exists $c_h>0$ such that $\boldsymbol{\theta}'\mathbf{h}\boldsymbol{\theta}\ge c_h\|\boldsymbol{\theta}\|^2$ for all $\boldsymbol{\theta}$. By Cauchy-Schwarz, $Q(\boldsymbol{\theta})
\ge
(c_h/2)\|\boldsymbol{\theta}\|^2
-\|\mathbf{f}\|\,\|\boldsymbol{\theta}\|
+c$. The right-hand side tends to $+\infty$ as $\|\boldsymbol{\theta}\|\to\infty$, so $Q$ is coercive.}

The feasible set $\Theta$ is nonempty: setting all polynomial coefficients in $G$ to zero satisfies the shape restrictions. Moreover, the restrictions defining $\Theta$ are affine or convex absolute-value inequalities. The absolute-value inequalities are closed and convex, and their standard linearization with auxiliary variables yields an equivalent system of affine constraints. Therefore, $\Theta$ is closed and convex.

Because $Q$ is continuous and coercive, it attains a minimum on the nonempty closed set $\Theta$ \citep[][Theorem 1.9, p. 11 and Corollary 3.27, p. 92]{rockafellar2009variational}. Strict convexity of $Q$ and convexity of $\Theta$ make this minimizer unique. Equivalently, the estimation problem is a strictly convex quadratic program over a nonempty closed convex feasible set and therefore has a unique solution \citep[][Section 4.2, p. 137]{boyd2004convex}.
\end{proof}

\subsection{Sup-norm Consistency of the Sieve GMM Estimator}\label{Appendix:SupNormConsistency}

Fix $J$ and maintain the conditions of Proposition~\ref*{PropBH} in the main text. Let $\mathcal V_0=\bigcup_{j=1}^J\operatorname{supp}(\mathbf V_{jt})$ be a compact subset of the model's domain, and let the sieve domain $\widetilde{\mathcal V}$ be a compact set, unchanged by relabeling inside competitors, that contains $\mathbf h(\mathcal V_0)$ and on which $\widetilde G_0$ is continuous. Write $\vartheta=(\alpha,\boldsymbol\beta,G)$, with true value $\vartheta_0$, and use the metric $d(\vartheta,\widetilde\vartheta)=|\alpha-\widetilde\alpha|+\|\boldsymbol\beta-\widetilde{\boldsymbol\beta}\|+\|G-\widetilde G\|_{\infty,\mathcal V_0}$. Let $\mathcal H$ be a set of candidate triples $(\alpha,\boldsymbol\beta,G)$, with $G$ continuous and permutation invariant, respecting the constant normalization.

For a degree sequence $D_T\to\infty$, let $\widehat Q_T(\vartheta)=\mathbf g_T(\vartheta)'\mathbf W_T\mathbf g_T(\vartheta)$ denote the sample GMM criterion, with $\mathbf g_T$ formed using the factor $1/(JT)$ as in Section~\ref*{Subsec:EmpiricalAppli} of the main text. For $\vartheta\in\mathcal H$, let $m_j(\mathbf x_t,\mathbf z_t;\vartheta)=\mathbb E[\xi_{jt}(\vartheta)\mid\mathbf x_t,\mathbf z_t]$, with $\xi_{jt}(\vartheta)$ defined in \eqref{eq:xistruct}, and define the population criterion
\begin{equation}\label{eq:PopCriterion}
Q(\vartheta)=\sum_{j\in\mathcal J_0}\mathbb E\bigl[m_j(\mathbf x_t,\mathbf z_t;\vartheta)^2\bigr].
\end{equation}
Assume:
\begin{enumerate}[label=(\roman*),leftmargin=*,itemsep=6pt]
\item $\vartheta_0\in\mathcal H$, $\mathcal H$ is compact under $d$, and the population criterion $Q$ defined in \eqref{eq:PopCriterion} is continuous and real-valued on $\mathcal H$ and uniquely minimized at $\vartheta_0$ (Condition~3.1 and Remark~3.1(4) in \citealp{chen2007large}).

\item The sets $\mathcal H_D\subseteq\mathcal H$ are nested symmetric polynomial sieves in the normalized arguments, and there exist $\vartheta_D\in\mathcal H_D$ with $d(\vartheta_D,\vartheta_0)\to0$ as $D\to\infty$ (Condition~3.2 in \citealp{chen2007large}; with (i), also Condition~3.3).

\item Each $\mathcal H_D$ is closed in the metric $d$, and hence compact by (i) (Condition~3.4 in \citealp{chen2007large}).

\item $\widehat Q_T$ converges uniformly in probability over $\mathcal H_{D_T}$ to $Q$ (Conditions~3.5(i) -- (ii) in \citealp{chen2007large}).

\item The estimator $\widehat\vartheta_T\in\mathcal H_{D_T}$ is measurable and minimizes the sample GMM criterion up to an $o_P(1)$ error (Condition~3.5(iii) in \citealp{chen2007large}).
\end{enumerate}

Then $d(\widehat\vartheta_T,\vartheta_0)=o_P(1)$. This establishes consistency of $\widehat\alpha$ and $\widehat{\boldsymbol\beta}$ and sup-norm consistency of the estimated $G$ on $\mathcal V_0$ (\citealp[Theorem~3.1]{chen2007large}, applied to $-\widehat Q_T$ and $-Q$ since Chen uses maximization). Sup-norm consistency of $G$ does not by itself establish consistency of demand elasticities, which additionally requires control of the relevant derivatives and stability of the inverse market share Jacobian.

These conditions have the following interpretation. Under the conditions of Proposition~\ref*{PropBH} in the main text, $Q(\vartheta)=0$ only if $\alpha=\alpha_0$, $\boldsymbol\beta=\boldsymbol\beta_0$, and $G=G_0$ on $\mathcal V_0$, so uniqueness in (i) follows from identification; (iv) requires the number of instrument functions to grow with $T$. Under (i), continuity, compactness, and uniqueness imply that $Q$ is uniformly above its minimum away from $\vartheta_0$. Candidates that remain a fixed distance from the truth cannot have population criterion values arbitrarily close to the minimum. Condition (ii) means that increasing the polynomial degree allows admissible candidates to approximate the true coefficients and function arbitrarily closely, with uniform approximation of $G_0$ on $\mathcal V_0$. Proposition~\ref*{Prop:Approx} in the main text establishes approximation in the unrestricted symmetric polynomial spaces. Condition (ii) additionally requires that the approximants belong to the admissible sieves $\mathcal H_D$. Under (iii), for each fixed $D$, limits under $d$ of sequences in $\mathcal H_D$ remain in $\mathcal H_D$. Condition (iv) is uniform convergence of the sample GMM criterion to its population limit over the sieve used at sample size $T$. Under (v), the minimization may be approximate, provided its error converges to zero in probability.

The constrained estimator is also consistent if the economic restrictions permit approximation of the truth: the approximants $\vartheta_{D_T}$ must satisfy all implemented constraints with probability tending to one (in particular, the true function must itself satisfy the implemented restrictions, which, like diagonal dominance, need not follow from Assumption~\ref*{AssReg} in the main text). On this event, a measurable estimator minimizing the sample GMM criterion over the constrained subset of $\mathcal H_{D_T}$ up to an $o_P(1)$ error has criterion at most $\widehat Q_T(\vartheta_{D_T})+o_P(1)$. Consistency then follows by adapting the proof of \citet[Theorem~3.1]{chen2007large} using this bound. Proposition~\ref*{Prop:Approx} in the main text establishes polynomial density, but does not ensure approximation within $\mathcal H$ or feasibility under the economic constraints.

    \section{Shape Restrictions}
\label{app:shape_restrictions}

\subsection{Local Restrictions}

In this appendix, I report the shape restrictions imposed on the
finite-dimensional sieve approximation used in estimation. I impose these restrictions market by market at the observed data points.
They regularize the estimator by enforcing local symmetry and
positive definiteness. Because I impose them only at
observed points, the restrictions need not hold globally.

Let $\mathbf{1}$ denote the $J$-dimensional vector of ones and define
$\mathbf{J}_J=\mathbf{1}\mathbf{1}'$. Let
$\mathbf{s}_t=(s_{1t},\ldots,s_{Jt})'$ collect the inside good shares and
let $s_{0t}=1-\sum_{j\in\mathcal{J}_0}s_{jt}$. For a vector-valued function
$f:\mathbb{R}^J\rightarrow\mathbb{R}^J$, let
$\mathbf{J}^{\mathbf{x}}_f(\mathbf{x}_0)$ denote its Jacobian with respect to
$\mathbf{x}$, evaluated at $\mathbf{x}_0$.

The estimated inverse market share equations are
\begin{equation}
\widehat{\xi}_{jt}
=
\log\!\left(\frac{s_{jt}}{s_{0t}}\right)
+
G_D\!\left(
\{r_{jk,t},d^p_{jk,t},\mathbf{d}^{(2)}_{jk,t}\}_{k\neq j};
\widehat{\boldsymbol{\theta}}_D
\right)
-
x_{jt}\widehat{\boldsymbol{\beta}}
+
\widehat{\alpha}p_{jt},
\qquad j\in\mathcal{J}_0.
\label{eq:shape_price_inverse}
\end{equation}
Define
\begin{align}
\Gamma^{s,p}_{jk,t}(\boldsymbol{\theta}_D)
&=
\frac{
\partial G_D\!\left(
\{r_{j\ell,t},d^p_{j\ell,t},\mathbf{d}^{(2)}_{j\ell,t}\}_{\ell\neq j};
\boldsymbol{\theta}_D
\right)}
{\partial s_{kt}},
\qquad j,k\in\mathcal{J}_0,
\label{eq:shape_price_gamma_s}\\
\Gamma^p_{jk,t}(\boldsymbol{\theta}_D)
&=
\frac{
\partial G_D\!\left(
\{r_{j\ell,t},d^p_{j\ell,t},\mathbf{d}^{(2)}_{j\ell,t}\}_{\ell\neq j};
\boldsymbol{\theta}_D
\right)}
{\partial d^p_{jk,t}},
\qquad j\neq k.
\label{eq:shape_price_gamma_p}
\end{align}
The Jacobian of the inverse-share function with respect to $\mathbf{s}_t$ is
\begin{equation}
\mathbf{A}^p_t(\boldsymbol{\theta}_D)
=
\frac{1}{s_{0t}}\mathbf{J}_J
+
\operatorname{diag}(\mathbf{s}_t)^{-1}
+
\big[
\Gamma^{s,p}_{jk,t}(\boldsymbol{\theta}_D)
\big]_{j,k=1}^J.
\label{eq:shape_price_A}
\end{equation}
The Jacobian with respect to prices is
$\mathbf{B}^p_t(\alpha,\boldsymbol{\theta}_D)$, whose entries are
\begin{equation}
\big[
\mathbf{B}^p_t(\alpha,\boldsymbol{\theta}_D)
\big]_{jk}
=
\begin{cases}
\displaystyle
\alpha-\sum_{\ell\neq j}
\Gamma^p_{j\ell,t}(\boldsymbol{\theta}_D),
& j=k, \\[0.9em]
\displaystyle
\Gamma^p_{jk,t}(\boldsymbol{\theta}_D),
& j\neq k.
\end{cases}
\label{eq:shape_price_B}
\end{equation}
Implicit differentiation gives the demand-price Jacobian:
\begin{equation}
\mathbf{J}^{\mathbf{p}}_{\boldsymbol{\sigma}}
(\boldsymbol{\delta}_t,\mathbf{p}_t,\mathbf{x}_t^{(2)})
=
-
\big[
\mathbf{A}^p_t(\widehat{\boldsymbol{\theta}}_D)
\big]^{-1}
\mathbf{B}^p_t
(\widehat{\alpha},\widehat{\boldsymbol{\theta}}_D).
\label{eq:shape_price_demand_jacobian}
\end{equation}

Let
\begin{equation}
\mathbf{C}_t^p(\boldsymbol{\theta}_D)
=
\operatorname{diag}(\mathbf{s}_t)^{-1}
+
\big[
\Gamma^{s,p}_{jk,t}(\boldsymbol{\theta}_D)
\big]_{j,k=1}^J,
\end{equation}
so that
\[
\mathbf{A}_t^p
=
\mathbf{C}_t^p
+
\frac{1}{s_{0t}}\mathbf{1}\mathbf{1}'.
\]

When these matrices are evaluated at a candidate coefficient vector rather than at $\widehat{\boldsymbol{\theta}}_D$, I write $\mathbf{A}^p_t(\boldsymbol{\theta}_D)$, $\mathbf{B}^p_t(\alpha,\boldsymbol{\theta}_D)$ and $\mathbf{C}^p_t(\boldsymbol{\theta}_D)$; at $\widehat{\boldsymbol{\theta}}_D$.

At each observed market, I require $\mathbf{C}_t^p$ to be symmetric and positive definite:
\begin{align}
\mathbf{C}_t^p(\boldsymbol{\theta}_D)
&=
\mathbf{C}_t^p(\boldsymbol{\theta}_D)',
\tag{Symmetry}
\label{eq:shape_no_price_symmetry}\\
\mathbf{C}_t^p(\boldsymbol{\theta}_D)
&\succ 0.
\tag{Positive Definiteness}
\label{eq:shape_no_price_psd}
\end{align}

I implement both conditions using linear restrictions on the sieve
coefficients. Let $K_D=\dim(\boldsymbol{\theta}_D)$ denote the number of flexible sieve
coefficients.

\paragraph{Symmetry.}
At each observed market $t$, symmetry requires
\begin{equation}
[\mathbf{C}_t^p(\boldsymbol{\theta}_D)]_{jk}
-
[\mathbf{C}_t^p(\boldsymbol{\theta}_D)]_{kj}
=
\Gamma^{s,p}_{jk,t}(\boldsymbol{\theta}_D)
-
\Gamma^{s,p}_{kj,t}(\boldsymbol{\theta}_D)
=
0,
\qquad j<k.
\label{eq:shape_symmetry_linear}
\end{equation}
The first equality follows because
$\operatorname{diag}(\mathbf{s}_t)^{-1}$ is symmetric. Since each element
of $\Gamma^{s,p}_t(\boldsymbol{\theta}_D)$ is linear in
$\boldsymbol{\theta}_D$, these conditions are linear equality restrictions.
With $T$ markets and $J$ products per market, imposing every restriction
directly produces $TJ(J-1)/2$ equality rows and requires no auxiliary
variables.

Although the pooled system contains $TJ(J-1)/2$ equations, every row
restricts the same $K_D$-dimensional coefficient vector
$\boldsymbol{\theta}_D$. Its rank is therefore at most
$\min\{K_D,TJ(J-1)/2\}$. Any rows beyond this rank are redundant and do not
further restrict the feasible coefficient set.

To avoid constructing the full system and retaining redundant rows, I use an
equivalent low-memory representation. Let $\mathbf{R}_{\mathrm{sym}}$ denote the
$TJ(J-1)/2\times K_D$ matrix of raw symmetry restrictions. The algorithm
processes one market at a time and maintains an orthonormal basis
$\mathbf{Q}_{\mathrm{sym}}$ for the accumulated row space of
$\mathbf{R}_{\mathrm{sym}}$. It then imposes
$\mathbf{Q}_{\mathrm{sym}}\boldsymbol{\theta}_D=0$.

Because the rows of $\mathbf{Q}_{\mathrm{sym}}$ span the same space as those
of $\mathbf{R}_{\mathrm{sym}}$, the two matrices impose the same restrictions
on the sieve coefficients, up to the numerical rank tolerance. Let
$r_{\mathrm{sym}}$ denote the numerical rank of
$\mathbf{R}_{\mathrm{sym}}$. The low-memory representation therefore sends
only $r_{\mathrm{sym}}$ equality rows to the solver, with
$r_{\mathrm{sym}}\leq\min\{K_D,TJ(J-1)/2\}$.

The value of $r_{\mathrm{sym}}$ depends on the sieve degree and the observed market shares, prices, and product characteristics. The retained rows are orthonormal linear combinations of the raw restrictions, not a selected subset of product-pair restrictions. Thus, the
low-memory representation reduces memory use and improves numerical
conditioning without changing the feasible coefficient set.

\paragraph{Positive definiteness.}
Together with symmetry, strict row-diagonal dominance provides a linear sufficient
condition for positive definiteness. Fix a small positive margin $\varepsilon_{\mathrm{pd}}$. In the implementation, I set $\varepsilon_{\mathrm{pd}}=10^{-6}$.
For every market $t$ and ordered
pair $j\neq k$, I introduce a nonnegative auxiliary variable $u_{jk,t}$ satisfying
\begin{align}
u_{jk,t}
&\geq
[\mathbf{C}_t^p(\boldsymbol{\theta}_D)]_{jk},
&& j\neq k,
\nonumber\\
u_{jk,t}
&\geq
-[\mathbf{C}_t^p(\boldsymbol{\theta}_D)]_{jk},
&& j\neq k,
\label{eq:shape_no_price_abs}\\
[\mathbf{C}_t^p(\boldsymbol{\theta}_D)]_{jj}
&\geq
\sum_{k\neq j}u_{jk,t}
+\varepsilon_{\mathrm{pd}},
&& j=1,\ldots,J.
\label{eq:shape_no_price_dd}
\end{align}
The first two inequalities imply $u_{jk,t}
\geq
\left|
[\mathbf{C}_t^p(\boldsymbol{\theta}_D)]_{jk}
\right|$, so the last inequality ensures
\[
[\mathbf{C}_t^p(\boldsymbol{\theta}_D)]_{jj}
\geq
\sum_{k\neq j}
\left|
[\mathbf{C}_t^p(\boldsymbol{\theta}_D)]_{jk}
\right|
+\varepsilon_{\mathrm{pd}}.
\]

This construction introduces $TJ(J-1)$ auxiliary variables and adds
$2TJ(J-1)$ inequalities to bound the absolute values, together with $TJ$ strict diagonal-dominance inequalities. The total number of linear inequality rows
is therefore $2TJ(J-1)+TJ=TJ(2J-1)$. Together with the direct symmetry
system, these conditions add $TJ(J-1)/2+TJ(2J-1)$ constraint rows. With the
low-memory symmetry representation, they instead add
$r_{\mathrm{sym}}+TJ(2J-1)$ rows.

I impose strict row-diagonal dominance on $\mathbf{C}_t^p$ rather than on
$\mathbf{A}_t^p$. The rank-one term
$s_{0t}^{-1}\mathbf{1}\mathbf{1}'$ can prevent $\mathbf{A}_t^p$ from being
row-diagonally dominant even when it is positive definite. Under symmetry, the positive-margin diagonal-dominance restrictions ensure that
$\mathbf{C}_t^p\succeq\varepsilon_{\mathrm{pd}}\mathbf{I}_J$ and hence that
$\mathbf{C}_t^p$ is positive definite. Since $s_{0t}^{-1}\mathbf{1}\mathbf{1}'$ is also positive
semidefinite, $\mathbf{A}_t^p$ is positive definite.
The restriction on
$\mathbf{C}_t^p$ is therefore a conservative but computationally convenient
sufficient condition for the corresponding restriction on
$\mathbf{A}_t^p$. The positive margin therefore guarantees that both
$\mathbf{C}_t^p$ and $\mathbf{A}_t^p$ are positive definite and nonsingular.

Finally, following \citet*{compiani2018}, one can impose one restriction under
substitutability and two further restrictions under additional conditions. First, under substitutability,
the inside good block of $\mathbf{J}^{\boldsymbol{\delta}}_{\boldsymbol{\sigma}}$ is an
$M$-matrix. Its inverse, the Jacobian of the inverse share function, is
therefore an inverse M-matrix and has nonnegative entries:
\begin{equation}
[\mathbf{A}_t^p(\boldsymbol{\theta}_D)]_{jk}
=
\frac{\partial \sigma_j^{-1}}
{\partial s_k}
(\mathbf{s}_t,\mathbf{p}_t,\mathbf{x}_t^{(2)})
\geq 0,
\qquad j,k=1,\ldots,J.
\label{eq:shape_inverse_share_monotonicity}
\tag{Share Monotonicity}
\end{equation}

Second, suppose that the own effect of each demand index dominates the sum
of its cross effects: $|
\frac{\partial \sigma_j}
{\partial \delta_j}
(\boldsymbol{\delta}_t,\mathbf{p}_t)
|
\geq
\sum_{k\neq j}
|
\frac{\partial \sigma_k}
{\partial \delta_j}
(\boldsymbol{\delta}_t,\mathbf{p}_t)
|$, for all $j=1,\ldots,J$. Then:
\begin{equation}
[\mathbf{A}_t^p(\boldsymbol{\theta}_D)]_{jj}
\geq
[\mathbf{A}_t^p(\boldsymbol{\theta}_D)]_{jk},
\qquad j\neq k.
\label{eq:shape_inverse_share_dominance}
\tag{Share-Derivative Dominance}
\end{equation}

Third, suppose that demand slopes downward in own price: $\frac{\partial \sigma_j}
{\partial p_j}
(\boldsymbol{\delta}_t,\mathbf{p}_t,\mathbf{x}_t^{(2)})<0$ for all $j=1,\ldots,J$, that substitutability holds, and that the inverse share function of every product is
weakly increasing in each competitor's price, $\Gamma^p_{jk,t}(\boldsymbol{\theta}_D)\geq 0$
for all $j\neq k$. Then, the inverse share equation for product
$j$ is weakly increasing in its own price:
\[
[\mathbf{B}_t^p(\alpha,\boldsymbol{\theta}_D)]_{jj}=\frac{\partial \sigma_j^{-1}}
{\partial p_j}
(\mathbf{s}_t,\mathbf{p}_t,\mathbf{x}_t^{(2)})\geq0,\quad j=1,\ldots,J.\label{eq:shape_own_price_monotonicity}
\tag{Own-Price Monotonicity}
\]

At the observed market states, the three implementations respectively add
$TJ^2$, $TJ(J-1)$, and $TJ$ linear inequalities and require no auxiliary
variables. They are implications of the assumptions, rather than
a complete characterization of substitutability. See \citet*{dellacherie2020inverse} for necessary and sufficient
conditions.

\subsection{Global Restrictions}

Imposing the shape restrictions separately at each observed market state can generate a large constraint system and guarantees the restrictions only at the states at which they are imposed. I therefore use a coefficient-only sufficient condition that guarantees a uniformly positive symmetric part of the inverse-share Jacobian over a prespecified share and price domain. The condition is conservative, but its size does not depend on the number of observed markets. In particular, if $K_s$ sieve coefficients depend on shares, the restriction requires at most $K_s$ auxiliary variables and $2K_s+1$ linear inequalities.

\begin{proposition}[Sufficient global shape restriction]
Suppose $J\geq2$ and $0<\underline s\leq s_i\leq\overline s$ for all $i=1,\ldots,J$. Suppose also that $s_0=1-\sum_{i=1}^J s_i\geq\underline s_0>0$ and that $\underline p\leq p_i\leq\overline p$ for all inside products, where $\underline p$ and $\overline p$ are finite. Let the flexible component for product $i$ be $G_i=\sum_u\theta_u\varphi_u^i$, where $\varphi_u^i=\prod_{r=1}^{R_u}a_{\mathbf{k}_{ur}}^i$ and $a_{\mathbf{k}}^i=\frac{1}{J-1}\sum_{n\neq i}\chi_{\mathbf{k},in}$. The pool contains inside competitors only. Each $\chi_{\mathbf{k},in}$ is a product of Chebyshev polynomials evaluated at arguments in $[-1,1]$. Under the relative-difference share transformation used in the main text, its share argument is
\[
z_{in}=\tau\!\left(\log\frac{s_n}{s_i}\right)=\frac{r_{in}-1}{r_{in}+1}=\frac{s_n-s_i}{s_n+s_i},
\qquad r_{in}=\frac{s_n}{s_i},\qquad \tau(v)=\tanh(v/2).
\]
All other arguments, including transformed price differences and product characteristics, lie in $[-1,1]$ and do not depend on shares. Write the price argument as $w_{in}=\psi((p_n-p_i)/h_p)$, where $h_p>0$ is fixed and $\psi$ is either $\zeta(v)=v/\sqrt{1+v^2}$ or $\tanh$. Let $\overline z_p=\psi((\overline p-\underline p)/h_p)$ and define
\[
\ell_b^p=\max_{|z|\leq\overline z_p}|T_b(z)|,
\qquad
\lambda_{u,p}=\prod_{r=1}^{R_u}\ell_{k_{ur,p}}^p,
\]
where $k_{ur,p}$ is the price-coordinate degree. Thus $0\leq\lambda_{u,p}\leq1$. These bounds use price variation only through $\overline p-\underline p$.

Let $A=\log(\overline s/\underline s)$ and, for each share-coordinate degree $k\geq0$, define
\[
c_k(A)
=\max_{0\leq v\leq A}(1+e^v)\left|T_k'\!\left(\tau(v)\right)\right|\tau'(v)
=\max_{1\leq q\leq \overline s/\underline s}\frac{2q}{1+q}\left|T_k'\!\left(\frac{q-1}{q+1}\right)\right|,
\qquad c_0(A)=0.
\]
For each sieve term, define its degree-specific share weight by
\[
\kappa_{u,s}=\lambda_{u,p}\sum_{r=1}^{R_u}c_{k_{ur,s}}(A),
\]
where $k_{ur,s}$ is the degree of the share-coordinate polynomial in $\chi_{\mathbf{k}_{ur},in}$. The price factors weakly reduce the coefficient weights and leave the right-hand side of the restriction unchanged. Setting $\ell_b^p=1$ recovers the restriction without price bounds.

Suppose that, for some $\eta>0$ satisfying $\eta\overline s<1$, the coefficients satisfy
\begin{equation}\label{eq:GlobalShape}
\sum_u\kappa_{u,s}|\theta_u|\leq 1-\eta\overline s.
\end{equation}
Then, uniformly over the stated domain, the matrix $\mathbf A^p(\mathbf s)$ defined in \eqref{eq:shape_price_A} satisfies
\[
\operatorname{sym}\!\bigl(\mathbf A^p(\mathbf s)\bigr)
=\frac{\mathbf A^p(\mathbf s)+\mathbf A^p(\mathbf s)'}{2}
\succeq\eta\mathbf I_J.
\]
Consequently, $\mathbf A^p(\mathbf s)$ is nonsingular throughout the domain; if it is symmetric, then $\mathbf A^p(\mathbf s)\succeq\eta\mathbf I_J$ and is uniformly positive definite.
\end{proposition}

\begin{proof}
Write $\Gamma_{ij}=\partial G_i/\partial s_j$ and $C=\operatorname{diag}(1/s_i)+\Gamma$, so that $\mathbf A^p=C+\mathbf1\mathbf1^\top/s_0$. All share derivatives hold prices and characteristics fixed. Chebyshev polynomials satisfy $|T_k(x)|\leq1$ for all $x\in[-1,1]$. Moreover, $d_k(v)=\left|T_k'\!\left(\tau(v)\right)\right|\tau'(v)$ is even because $\tau$ is odd, $\tau'(v)=(1-\tau(v)^2)/2>0$ is even, and $\left|T_k'\right|$ is even. With $q=e^v$, $(1+e^v)\tau'(v)=2q/(1+q)$, which gives the second expression for $c_k(A)$. The argument uses only that $\tau$ is odd, strictly increasing, continuously differentiable, and takes values in $[-1,1]$.

For an inside pair, let $v=\log(s_n/s_i)$. Since $|v|\leq A$,
\[
s_i\left(
\left|\frac{\partial T_k(\tau(v))}{\partial s_i}\right|
+\left|\frac{\partial T_k(\tau(v))}{\partial s_n}\right|
\right)
=(1+e^{-v})d_k(v)\leq c_k(A).
\]
Averaging over the $J-1$ inside competitors therefore yields
\[
s_i\sum_j\left|\frac{\partial a_{\mathbf{k}}^i}{\partial s_j}\right|
\leq \ell_{k_p}^p c_{k_s}(A).
\]
For a fixed derivative column, the own-product equation and the $J-1$ competitor equations each contribute at most one average of $d_{k_s}$. Hence
\[
s_j\sum_i\left|\frac{\partial a_{\mathbf{k}}^i}{\partial s_j}\right|
\leq 2\ell_{k_p}^p\max_{|v|\leq A}d_{k_s}(v)
\leq \ell_{k_p}^p c_{k_s}(A),
\]
where the last inequality follows from the evenness of $d_{k_s}$ and $1+e^v\geq2$ for $v\geq0$.

Because $|a_{\mathbf{k}}^i|\leq\ell_{k_p}^p$, the product rule retains every price factor and gives
\[
s_i\sum_j\left|\frac{\partial\varphi_u^i}{\partial s_j}\right|\leq\kappa_{u,s},
\]
with the same bound for derivative columns. Hence, writing $L_s=\sum_u\kappa_{u,s}|\theta_u|$, we obtain
\[
s_i\sum_j|\Gamma_{ij}|\leq L_s,
\qquad
s_j\sum_i|\Gamma_{ij}|\leq L_s.
\]
Restriction \eqref{eq:GlobalShape} therefore implies
\[
C_{ii}-\sum_{j\neq i}|C_{ij}|
\geq\frac{1-L_s}{s_i}
\geq\frac{\eta\overline s}{s_i}
\geq\eta.
\]
The column inequality follows identically.

Let $H=\operatorname{sym}(C)$. Averaging the row and column inequalities gives
\[
H_{ii}-\sum_{j\neq i}|H_{ij}|\geq\eta.
\]
Because $H$ is symmetric, diagonal dominance with margin $\eta$ implies $H\succeq\eta \mathbf I_J$. Finally, since $s_0^{-1}\mathbf1\mathbf1^\top$ is positive semidefinite,
\[
\operatorname{sym}(\mathbf A^p)=H+\frac{\mathbf1\mathbf1^\top}{s_0}\succeq\eta \mathbf I_J.
\]
Because $\eta>0$, $\mathbf A^p$ is nonsingular.
\end{proof}

\section{Additional Monte Carlo Results: $J=25$}
\label{app:mc_j25}

This appendix reports the results for $J=25$ from the three experiments in Section~\ref*{Section:Simul} of the main text. Tables~\ref{tab:mc_normal_j25}, \ref{tab:mc_snpd_j25}, and \ref{tab:mc_parametric_j25} report the results for the RCL model with normal coefficients, the RCL model with a lognormal price coefficient, and the parametric inverse market share model, respectively. Each experiment uses 100 Monte Carlo replications with $T\in\{100,500\}$ markets and follows the simulation and estimation procedures described in the main text.

\begin{table}[H]
\centering
\caption{RCL model with normal price and characteristic coefficients: $J=25$}
\label{tab:mc_normal_j25}
\begingroup
\fontsize{10.25}{10.5}\selectfont
\setlength{\tabcolsep}{0.8pt}
\renewcommand{\arraystretch}{0.92}
\begin{tabular*}{\textwidth}{
@{\extracolsep{\fill}}l c c
r@{\hspace{4pt}}r
r@{\hspace{4pt}}r
r@{\hspace{4pt}}r@{}}
\toprule
Statistic & $T$ & True
& \multicolumn{2}{c}{\makebox[0pt][c]{\shortstack{Correctly\\[-1pt] specified RCL}}}
& \multicolumn{2}{c}{SNPD}
& \multicolumn{2}{c}{Logit} \\
\cmidrule(lr){4-5}
\cmidrule(lr){6-7}
\cmidrule(lr){8-9}
& & & Mean & RMSE & Mean & RMSE & Mean & RMSE \\
\midrule

\multicolumn{9}{@{}l}{\textit{Price elasticities}} \\
\quad Own
& 100 & -2.484 & -2.490 & 0.054
& -2.482 & 0.054 & -2.285 & 0.213 \\
& 500 & -2.483 & -2.486 & 0.026
& -2.478 & 0.028 & -2.284 & 0.203 \\[0.5pt]

\quad Cross, all
& 100 & 0.037 & 0.037 & 0.001
& 0.039 & 0.010 & 0.017 & 0.020 \\
& 500 & 0.037 & 0.037 & 0.001
& 0.039 & 0.005 & 0.017 & 0.020 \\[0.5pt]

\quad Cross, five closest
& 100 & 0.054 & 0.054 & 0.002
& 0.050 & 0.016 & 0.017 & 0.037 \\
& 500 & 0.054 & 0.054 & 0.001
& 0.049 & 0.008 & 0.017 & 0.037 \\[0.5pt]

\quad Cross, others
& 100 & 0.032 & 0.032 & 0.001
& 0.037 & 0.010 & 0.017 & 0.015 \\
& 500 & 0.032 & 0.032 & 0.000
& 0.036 & 0.005 & 0.017 & 0.015 \\

\addlinespace[2pt]
\multicolumn{9}{@{}l}{\textit{Mean diversion ratios}} \\
\quad All
& 100 & 0.019 & 0.019 & 0.000
& 0.015 & 0.005 & 0.007 & 0.012 \\
& 500 & 0.019 & 0.019 & 0.000
& 0.015 & 0.004 & 0.007 & 0.012 \\[0.5pt]

\quad Five closest
& 100 & 0.040 & 0.039 & 0.001
& 0.032 & 0.011 & 0.012 & 0.028 \\
& 500 & 0.040 & 0.040 & 0.000
& 0.032 & 0.009 & 0.012 & 0.028 \\[0.5pt]

\quad Other products
& 100 & 0.013 & 0.013 & 0.000
& 0.011 & 0.004 & 0.006 & 0.008 \\
& 500 & 0.013 & 0.013 & 0.000
& 0.011 & 0.003 & 0.006 & 0.008 \\

\addlinespace[2pt]
\multicolumn{9}{@{}l}{\textit{Total diversion ratios}} \\
\quad Five closest
& 100 & 0.198 & 0.197 & 0.002
& 0.158 & 0.053 & 0.058 & 0.140 \\
& 500 & 0.198 & 0.198 & 0.001
& 0.158 & 0.043 & 0.058 & 0.140 \\[0.5pt]

\quad Other inside
& 100 & 0.251 & 0.251 & 0.002
& 0.204 & 0.078 & 0.109 & 0.142 \\
& 500 & 0.251 & 0.251 & 0.001
& 0.200 & 0.058 & 0.109 & 0.142 \\[0.5pt]

\quad Outside good
& 100 & 0.552 & 0.552 & 0.004
& 0.638 & 0.126 & 0.833 & 0.281 \\
& 500 & 0.552 & 0.552 & 0.002
& 0.643 & 0.099 & 0.833 & 0.281 \\

\addlinespace[2pt]
\multicolumn{9}{@{}l}{\textit{Price curvature}} \\
\quad Own
& 100 & 1.100 & 1.114 & 0.066
& 1.087 & 0.198 & 0.993 & 0.107 \\
& 500 & 1.100 & 1.106 & 0.028
& 1.093 & 0.077 & 0.993 & 0.107 \\

\addlinespace[2pt]
\multicolumn{9}{@{}l}{\textit{Markups}} \\
\quad Markup (\%)
& 100 & 42.19 & 41.98 & 1.022
& 41.74 & 1.369 & 44.05 & 2.380 \\
& 500 & 42.19 & 42.09 & 0.486
& 41.58 & 0.921 & 44.03 & 1.971 \\
\bottomrule
\end{tabular*}
\endgroup

\par\vspace{1pt}
\begin{minipage}{\textwidth}
\fontsize{8.5}{9.2}\selectfont
\textit{Notes:}
Results are based on 100 Monte Carlo replications with \(J=25\) products
per market. \(T\) denotes the number of markets. ``True'' and ``Mean'' are
averages across replications; RMSE is the square root of the average squared
difference between estimated and true statistics. Price elasticities and
own-price curvature are unweighted. Diversion ratios weight source products
by market share. Mean diversion ratios average across individual
destinations, whereas total diversion ratios sum within destination groups.
Markups are market-share-weighted price-cost margins expressed as a
percentage of price. The five closest products minimize Euclidean distance
in price-characteristic space after rescaling both variables to \([0,1]\)
within each market.
\end{minipage}
\end{table}

\begin{table}[H]
\centering
\caption{RCL model with lognormal price and normal characteristic coefficients: $J=25$}
\label{tab:mc_snpd_j25}
\begingroup
\fontsize{10.25}{10.5}\selectfont
\setlength{\tabcolsep}{0.4pt}
\renewcommand{\arraystretch}{0.92}
\begin{tabular*}{\textwidth}{
@{\extracolsep{\fill}}l c c
r@{\hspace{3pt}}r
r@{\hspace{3pt}}r
r@{\hspace{3pt}}r
r@{\hspace{3pt}}r@{}}
\toprule
Statistic & $T$ & True
& \multicolumn{2}{c}{\makebox[0pt][c]{\shortstack{Correctly\\[-1pt] specified RCL}}}
& \multicolumn{2}{c}{\makebox[0pt][c]{\shortstack{Misspecified\\[-1pt] RCL}}}
& \multicolumn{2}{c}{SNPD}
& \multicolumn{2}{c}{Logit} \\
\cmidrule(lr){4-5}
\cmidrule(lr){6-7}
\cmidrule(lr){8-9}
\cmidrule(lr){10-11}
& & & Mean & RMSE & Mean & RMSE & Mean & RMSE & Mean & RMSE \\
\midrule

\multicolumn{11}{@{}l}{\textit{Price elasticities}} \\
\quad Own
& 100 & -2.188 & -2.194 & 0.042 & -2.138 & 0.069
& -2.181 & 0.045 & -2.129 & 0.092 \\
& 500 & -2.188 & -2.191 & 0.021 & -2.139 & 0.055
& -2.176 & 0.026 & -2.128 & 0.069 \\[0.5pt]

\quad Cross, all
& 100 & 0.040 & 0.040 & 0.002 & 0.029 & 0.010
& 0.040 & 0.008 & 0.011 & 0.028 \\
& 500 & 0.040 & 0.040 & 0.001 & 0.029 & 0.010
& 0.039 & 0.003 & 0.011 & 0.028 \\[0.5pt]

\quad Cross, five closest
& 100 & 0.056 & 0.057 & 0.002 & 0.046 & 0.010
& 0.053 & 0.013 & 0.011 & 0.045 \\
& 500 & 0.056 & 0.056 & 0.001 & 0.046 & 0.010
& 0.053 & 0.006 & 0.011 & 0.045 \\[0.5pt]

\quad Cross, others
& 100 & 0.035 & 0.035 & 0.001 & 0.025 & 0.010
& 0.036 & 0.007 & 0.011 & 0.024 \\
& 500 & 0.035 & 0.035 & 0.001 & 0.025 & 0.010
& 0.036 & 0.003 & 0.011 & 0.024 \\

\addlinespace[2pt]
\multicolumn{11}{@{}l}{\textit{Mean diversion ratios}} \\
\quad All
& 100 & 0.021 & 0.021 & 0.001 & 0.017 & 0.004
& 0.019 & 0.004 & 0.005 & 0.016 \\
& 500 & 0.021 & 0.021 & 0.000 & 0.017 & 0.004
& 0.019 & 0.002 & 0.005 & 0.016 \\[0.5pt]

\quad Five closest
& 100 & 0.040 & 0.040 & 0.001 & 0.033 & 0.007
& 0.037 & 0.006 & 0.008 & 0.032 \\
& 500 & 0.040 & 0.040 & 0.001 & 0.033 & 0.007
& 0.037 & 0.004 & 0.008 & 0.032 \\[0.5pt]

\quad Other products
& 100 & 0.016 & 0.016 & 0.001 & 0.012 & 0.004
& 0.015 & 0.004 & 0.004 & 0.011 \\
& 500 & 0.016 & 0.016 & 0.000 & 0.012 & 0.004
& 0.014 & 0.002 & 0.004 & 0.011 \\

\addlinespace[2pt]
\multicolumn{11}{@{}l}{\textit{Total diversion ratios}} \\
\quad Five closest
& 100 & 0.199 & 0.201 & 0.007 & 0.166 & 0.034
& 0.186 & 0.031 & 0.039 & 0.160 \\
& 500 & 0.199 & 0.200 & 0.004 & 0.166 & 0.034
& 0.185 & 0.018 & 0.039 & 0.160 \\[0.5pt]

\quad Other inside
& 100 & 0.301 & 0.303 & 0.010 & 0.234 & 0.068
& 0.276 & 0.068 & 0.084 & 0.217 \\
& 500 & 0.301 & 0.302 & 0.006 & 0.234 & 0.067
& 0.274 & 0.039 & 0.084 & 0.217 \\[0.5pt]

\quad Outside good
& 100 & 0.499 & 0.496 & 0.017 & 0.600 & 0.102
& 0.538 & 0.093 & 0.876 & 0.377 \\
& 500 & 0.499 & 0.498 & 0.009 & 0.600 & 0.101
& 0.541 & 0.055 & 0.876 & 0.377 \\

\addlinespace[2pt]
\multicolumn{11}{@{}l}{\textit{Price curvature}} \\
\quad Own
& 100 & 1.330 & 1.352 & 0.128 & 2.263 & 1.549
& 1.336 & 0.187 & 0.995 & 0.335 \\
& 500 & 1.330 & 1.330 & 0.011 & 2.095 & 0.850
& 1.336 & 0.074 & 0.995 & 0.335 \\

\addlinespace[2pt]
\multicolumn{11}{@{}l}{\textit{Markups}} \\
\quad Markup (\%)
& 100 & 46.80 & 46.71 & 0.944 & 44.13 & 3.620
& 48.05 & 1.850 & 49.27 & 2.921 \\
& 500 & 46.80 & 46.77 & 0.417 & 43.76 & 3.072
& 47.92 & 1.327 & 49.26 & 2.573 \\
\bottomrule
\end{tabular*}
\endgroup

\par\vspace{1pt}
\begin{minipage}{\textwidth}
\fontsize{8.5}{9.2}\selectfont
\textit{Notes:}
Results are based on 100 Monte Carlo replications with \(J=25\) products
per market. \(T\) denotes the number of markets. ``True'' and ``Mean'' are
averages across replications; RMSE is the square root of the average squared
difference between estimated and true statistics. Price elasticities and
own-price curvature are unweighted. Diversion ratios weight source products
by market share. Mean diversion ratios average across individual
destinations, whereas total diversion ratios sum within destination groups.
Markups are market-share-weighted price-cost margins expressed as a
percentage of price. The five closest products minimize Euclidean distance
in price-characteristic space after rescaling both variables to \([0,1]\)
within each market. For the misspecified RCL, price-curvature statistics use
95 replications after excluding five datasets with extreme curvature values.
\end{minipage}
\end{table}

\begin{table}[H]
\centering
\caption{Parametric inverse market share model: $J=25$}
\label{tab:mc_parametric_j25}
\begingroup
\fontsize{10.25}{10.5}\selectfont
\setlength{\tabcolsep}{0.8pt}
\renewcommand{\arraystretch}{0.92}
\begin{tabular*}{\textwidth}{
@{\extracolsep{\fill}}l c c
r@{\hspace{4pt}}r
r@{\hspace{4pt}}r
r@{\hspace{4pt}}r@{}}
\toprule
Statistic & $T$ & True
& \multicolumn{2}{c}{\makebox[0pt][c]{\shortstack{Misspecified\\[-1pt] RCL}}}
& \multicolumn{2}{c}{SNPD}
& \multicolumn{2}{c}{Logit} \\
\cmidrule(lr){4-5}
\cmidrule(lr){6-7}
\cmidrule(lr){8-9}
& & & Mean & RMSE & Mean & RMSE & Mean & RMSE \\
\midrule

\multicolumn{9}{@{}l}{\textit{Price elasticities}} \\
\quad Own
& 100 & -4.293 & -4.298 & 0.064
& -4.295 & 0.064 & -4.307 & 0.065 \\
& 500 & -4.293 & -4.300 & 0.031
& -4.297 & 0.031 & -4.310 & 0.034 \\[0.5pt]

\quad Cross, all
& 100 & 0.063 & 0.053 & 0.011
& 0.062 & 0.008 & 0.034 & 0.029 \\
& 500 & 0.063 & 0.053 & 0.010
& 0.063 & 0.003 & 0.034 & 0.029 \\[0.5pt]

\quad Cross, five closest
& 100 & 0.067 & 0.054 & 0.013
& 0.071 & 0.013 & 0.034 & 0.033 \\
& 500 & 0.067 & 0.054 & 0.013
& 0.072 & 0.007 & 0.034 & 0.033 \\[0.5pt]

\quad Cross, others
& 100 & 0.062 & 0.053 & 0.010
& 0.060 & 0.007 & 0.034 & 0.028 \\
& 500 & 0.062 & 0.053 & 0.009
& 0.061 & 0.003 & 0.034 & 0.028 \\

\addlinespace[2pt]
\multicolumn{9}{@{}l}{\textit{Mean diversion ratios}} \\
\quad All
& 100 & 0.014 & 0.012 & 0.002
& 0.014 & 0.002 & 0.008 & 0.006 \\
& 500 & 0.014 & 0.012 & 0.001
& 0.014 & 0.001 & 0.008 & 0.006 \\[0.5pt]

\quad Five closest
& 100 & 0.016 & 0.014 & 0.002
& 0.017 & 0.003 & 0.009 & 0.007 \\
& 500 & 0.016 & 0.014 & 0.002
& 0.017 & 0.002 & 0.009 & 0.007 \\[0.5pt]

\quad Other products
& 100 & 0.013 & 0.012 & 0.002
& 0.013 & 0.002 & 0.008 & 0.005 \\
& 500 & 0.013 & 0.012 & 0.001
& 0.013 & 0.001 & 0.008 & 0.005 \\

\addlinespace[2pt]
\multicolumn{9}{@{}l}{\textit{Total diversion ratios}} \\
\quad Five closest
& 100 & 0.081 & 0.071 & 0.011
& 0.085 & 0.015 & 0.045 & 0.036 \\
& 500 & 0.081 & 0.071 & 0.010
& 0.087 & 0.009 & 0.045 & 0.036 \\[0.5pt]

\quad Other inside
& 100 & 0.245 & 0.222 & 0.029
& 0.239 & 0.032 & 0.144 & 0.101 \\
& 500 & 0.246 & 0.222 & 0.024
& 0.243 & 0.013 & 0.144 & 0.101 \\[0.5pt]

\quad Outside good
& 100 & 0.674 & 0.707 & 0.040
& 0.676 & 0.045 & 0.811 & 0.137 \\
& 500 & 0.674 & 0.707 & 0.034
& 0.670 & 0.018 & 0.811 & 0.137 \\

\addlinespace[2pt]
\multicolumn{9}{@{}l}{\textit{Price curvature}} \\
\quad Own
& 100 & 0.816 & 1.061 & 0.245
& 0.945 & 0.140 & 0.992 & 0.176 \\
& 500 & 0.816 & 1.061 & 0.245
& 0.946 & 0.132 & 0.992 & 0.176 \\

\addlinespace[2pt]
\multicolumn{9}{@{}l}{\textit{Markups}} \\
\quad Markup (\%)
& 100 & 25.87 & 24.79 & 1.139
& 25.85 & 0.515 & 24.48 & 1.429 \\
& 500 & 25.87 & 24.77 & 1.111
& 25.84 & 0.243 & 24.47 & 1.409 \\
\bottomrule
\end{tabular*}
\endgroup

\par\vspace{1pt}
\begin{minipage}{\textwidth}
\fontsize{8.5}{9.2}\selectfont
\textit{Notes:}
Results are based on 100 Monte Carlo replications with \(J=25\) products
per market. \(T\) denotes the number of markets. ``True'' and ``Mean'' are
averages across replications; RMSE is the square root of the average squared
difference between estimated and true statistics. Price elasticities and
own-price curvature are unweighted. Diversion ratios weight source products
by market share. Mean diversion ratios average across individual
destinations, whereas total diversion ratios sum within destination groups.
Markups are market-share-weighted price-cost margins expressed as a
percentage of price. The five closest products minimize Euclidean distance
in price-characteristic space after rescaling both variables to \([0,1]\)
within each market.
\end{minipage}
\end{table}

\section{Details on the Empirical Application \label{Appendix:Appli}}

\subsection{Demand Models}
The benchmark is the random coefficient nested logit (RCNL) model of
\citet*{miller2017understanding}. Let $y_{it}$ denote consumer $i$'s income in market $t$,
let $x^c_{jt}$ denote calories of product $j$ in market $t$, and let $d(t)$ denote the date associated with
market $t$. The deterministic component of consumer $i$'s utility from
product $j$ in market $t$ is
\[
v_{ijt}
=
\xi_j+\xi_{d(t)}-\alpha p_{jt}
+y_{it}\left(\pi_0+\pi_p p_{jt}+\pi_c x^c_{jt}\right)
+\xi_{jt},
\]
where $\xi_j$ and $\xi_{d(t)}$ are product and date fixed effects,
respectively, and $\xi_{jt}$ is the product-market demand shock. The empirical income distribution in each market generates heterogeneity in the valuation of the inside good, price, and calories. All inside products form a single nest and the outside good forms a separate singleton nest, with $\mu \in [0,1)$ denoting the corresponding nesting parameter

For the $N_t$ income draws with weights $w_{it}$ satisfying
$\sum_{i=1}^{N_t}w_{it}=1$, the aggregate inside good shares are
\[
s_{jt}
=
\sum_{i=1}^{N_t}w_{it}
\frac{
\exp\left(v_{ijt}/(1-\mu)\right)A_{it}^{-\mu}
}{
1+A_{it}^{\,1-\mu}
},
\qquad j\in\mathcal J_0,
\]
with $s_{0t}=1-\sum_{j\in\mathcal J_0}
s_{jt}$ and $A_{it}
=
\sum_{k\in\mathcal J_0}
\exp\left(\frac{v_{ikt}}{1-\mu}\right)$.

The SNPD specification is
\begin{equation}\label{eq:SNPD_MW}
\ln\left(\frac{
s_{jt}
}{
s_{0t}
}\right)
+
G_D\left(\widetilde{\mathbf V}_{jt}\right)
=
\xi_j+\xi_{d(t)}-\alpha p_{jt}
+\gamma_y\bar y_t
+\gamma_{yp}\bar y_t p_{jt}
+\gamma_{yc}\bar y_t x^c_{jt}
+\xi_{jt},
\end{equation}
where $G_D$ is the degree-$D$ permutation-invariant Chebyshev sieve, $\bar y_t$ is mean income in market $t$, and the vector $\widetilde{\mathbf V}_{jt}
=
(
\widetilde r_{jk,t},
\widetilde d^p_{jk,t},
\widetilde d^c_{jk,t}
)_{k\in\mathcal J_0\setminus\{j\}}$ collects the transformed share ratios, price differences,
and characteristic differences relative to product $j$. I approximate $G$ with $G_D$ and select the degree $D \in \{3,4,5,6\}$ by cross-validation. The selected degree is $D=3$. The cubic approximation uses only 53 nonconstant sieve terms. 

\subsection{RCNL Estimation}

I estimate the RCNL model of \citet*{miller2017understanding} with the same instruments using the PyBLP package developed by
\citet*{conlon2018}, following their recommended practices. I use \citet*{miller2017understanding}'s income draws. I solve the market share inversion problem using the SQUAREM accelerated fixed-point algorithm with a tight convergence tolerance of $10^{-12}$. I minimize the GMM objective using Knitro 13.1's Interior/Direct algorithm. I use two-step GMM with the original 12 instruments described in the main text. The weighting matrix and standard errors are clustered by market (region-month), which yields 2{,}808 clusters rather than the 39 clusters obtained by clustering by region as in \citet*{miller2017understanding}. Market-level clustering does not allow for correlation of demand shocks across months within a region.

\subsection{SNPD Estimation}

\paragraph{Constrained GMM.} I estimate the SNPD model using the specification in \eqref{eq:SNPD_MW}. The transformed variables
$\widetilde r_{jk,t}$, $\widetilde d^p_{jk,t}$, and
$\widetilde d^c_{jk,t}$ are the share ratio, price difference, and
calorie difference variables defined in the main text. The symmetric
Chebyshev sieve $G_D$ pools only inside competitors. Income is excluded from
$G_D$ and enters the linear index through its level and its interactions
with price and calories. I set $D=3$ in the empirical specification.

Before estimation, I express price per unit by dividing the price of a
144-ounce equivalent unit by 12, rescale $x^c_{jt}$ to $[0,1]$, and divide
mean income by 10. I estimate the model by constrained two-step GMM. The
instrument matrix is constructed from the excluded variables described in
the main text, their symmetric sieve transformations, and the included
exogenous regressors. I replace this matrix by an orthonormal basis for its
column space using a QR decomposition, which improves conditioning without
changing the instrument span. I construct the second-step matrix from market-clustered
moments. Because the estimated covariance matrix of the clustered moments is ill-conditioned, I regularize it using the linear Ledoit-Wolf shrinkage estimator before inversion \citep{ledoit2004well}.

I also impose the following shape restrictions. I require the share derivative matrix to be symmetric at the observed market states and impose coefficient bounds that ensure its symmetric part is positive definite over a prespecified domain of shares and prices. I also restrict the effective linear price coefficient to be negative over the observed income range and require the diagonal entries of $B_t$ to be nonnegative at the observed market states.
These restrictions are linear in the parameters and auxiliary variables, preserving the  quadratic programming formulation.

\paragraph{Selection of the sieve degree.} Sieve estimation requires selecting the polynomial degree \(D\). With endogenous regressors, standard cross-validation based on mean-squared prediction error is inappropriate because it yields a biased estimate of the structural mean-squared error \citep*{chen2025adaptive}. I instead select $D$ by minimizing a cross-validated GMM criterion, building on \citet*{rivers2002model} and \citet*{komiyama2018cross}, over \(D \in \{2,3,4,5\}\). I implement two-fold cross-validation at the market level by partitioning the markets into two nonoverlapping, approximately equal subsets, or folds. For each candidate \(D\), I estimate the model on one fold and evaluate the GMM objective on the other, then reverse their roles. I select the degree that minimizes the GMM objective averaged across the two validation folds.

\paragraph{Counterfactuals and welfare.} Counterfactual prices and market shares can be obtained by solving the corresponding demand-and-supply equilibrium, allowing the analysis of mergers, product exit, and cost or policy interventions such as taxes and tariffs. When prices do not enter $G$ and the shape restrictions hold globally, the demand model admits a representative consumer formulation and a consumer surplus function, which can be used to compute welfare changes \citep*{fosgerau2018demand}. When the estimated demand system does not admit a global consumer surplus function, a path-specific measure of welfare changes under simultaneous price movements can instead be constructed directly from demand \citep*{chipman1980compensating,auerbach1985theory,bhattacharya2018empirical,monardo2025product}.

\subsection{SNPD Inference} 

\paragraph{Numerical Bootstrap.} Binding inequality restrictions under the SNPD approach can generate
nonnormal limiting distributions, invalidating standard normal-based inference.
I therefore conduct pointwise-valid inference using the numerical bootstrap of
\citet*{hong2020numerical}, which accommodates the resulting nonregular behavior.

Inference is conditional on the selected sieve dimension \(D\), instrument set, and constraint set, with markets as the independent sampling units, and treats the degree-$D$ specification as correctly specified, so that the population moment vanishes at the parameter value being estimated. Let
\[
\widehat g_T(\boldsymbol{\theta})
=
\frac{1}{T}
\sum_{t=1}^{T}
\left[
\frac{1}{J}
\sum_{j=1}^{J}
\xi_{jt}(\boldsymbol{\theta})z_{jt}
\right]
\]
denote the sample moment vector. I use the step size $\epsilon_T=T^{-1/3}$.

For each bootstrap replication \(b\), I draw independent 
multipliers \(\{\nu_t^b\}_{t=1}^{T}\) from a standard normal distribution and define $\bar\nu^b=\frac{1}{T}\sum_{t=1}^{T}\nu_t^b$. 
The centered multiplier process is
\[
\widehat g_T^{*,b}(\boldsymbol{\theta})
=
\frac{1}{T}
\sum_{t=1}^{T}
(\nu_t^b-\bar\nu^b)
\left[
\frac{1}{J}
\sum_{j=1}^{J}
\xi_{jt}(\boldsymbol{\theta})z_{jt}
\right].
\]
The perturbed moment vector is
\[
\widehat g_T^{\epsilon,b}(\boldsymbol{\theta})
=
\widehat g_T(\boldsymbol{\theta})
+
\epsilon_T\sqrt{T}\,
\widehat g_T^{*,b}(\boldsymbol{\theta}).
\]

I estimate the final GMM weighting matrix \(\widehat W_T\) from the
original sample and hold it fixed across bootstrap replications. For each replication \(b\), the numerical-bootstrap estimate
\(\widehat{\boldsymbol{\theta}}^{*,b}\) solves
\[
\widehat{\boldsymbol{\theta}}^{*,b}
=
\arg\min_{\boldsymbol{\theta}\in\Theta_D}
\widehat g_T^{\epsilon,b}(\boldsymbol{\theta})'
\widehat W_T
\widehat g_T^{\epsilon,b}(\boldsymbol{\theta}),
\]
where \(\Theta_D\) is the same constrained sieve parameter space used for the
original estimator. 

For a smooth scalar object \(f(\boldsymbol{\theta})\), define the scaled
numerical-bootstrap statistic
\[
R_b(f)
=
\epsilon_T^{-1}
\left\{
f(\widehat{\boldsymbol{\theta}}^{*,b})
-
f(\widehat{\boldsymbol{\theta}})
\right\}.
\]
The conditional distribution of \(R_b(f)\) approximates the distribution of
\[
\sqrt{T}
\left\{
f(\widehat{\boldsymbol{\theta}})
-
f(\boldsymbol{\theta}_0)
\right\}.
\]

I estimate the standard error of
\(f(\widehat{\boldsymbol{\theta}})\) by
$\operatorname{sd}_b\{R_b(f)\}/\sqrt{T}$. Because the limiting distribution
may be nonnormal, this standard deviation provides only a heuristic guide.

Let \(q_\tau\) denote the \(\tau\)-quantile of
\(\{R_b(f)\}_{b=1}^{B}\). The corresponding \(95\%\) confidence
interval is
\[
\left[
f(\widehat{\boldsymbol{\theta}})
-
\frac{q_{0.975}}{\sqrt{T}},
\quad
f(\widehat{\boldsymbol{\theta}})
-
\frac{q_{0.025}}{\sqrt{T}}
\right].
\]

\paragraph{Local QP Approximation.} I also implement a local QP approximation to
\citet*{hong2020numerical}'s numerical bootstrap. The procedure estimates the
original model once and replaces each bootstrap re-estimation with a smaller
local quadratic program.

For fixed sieve dimension \(D\), let $\Theta_D
=
\left\{
\boldsymbol{\theta}:
R\boldsymbol{\theta}=r,\;
A\boldsymbol{\theta}\leq b
\right\}$ denote the constrained parameter space, and let \(\widehat G_T\) denote the
sample-moment Jacobian evaluated at \(\widehat{\boldsymbol{\theta}}\). The final
GMM weighting matrix \(\widehat W_T\) is estimated once and held fixed during
inference.

Using the same centered market multipliers as above, define
\[
\zeta_T^b
=
\frac{1}{\sqrt{T}}
\sum_{t=1}^{T}
(\nu_t^b-\bar\nu^b)
\left\{
m_t(\widehat{\boldsymbol{\theta}})
-
\widehat g_T(\widehat{\boldsymbol{\theta}})
\right\},
\qquad
m_t(\boldsymbol{\theta})
=
\frac{1}{J}
\sum_{j=1}^{J}
\xi_{jt}(\boldsymbol{\theta})z_{jt}.
\]
Following \citet*{hong2020numerical}, I approximate the
unknown tangent cone by
\[
\widehat{\mathcal C}_{T,\epsilon}
=
\epsilon_T^{-1}
\left(
\Theta_D-\widehat{\boldsymbol{\theta}}
\right)
=
\left\{
h:
Rh=0,\quad
Ah\leq
\frac{b-A\widehat{\boldsymbol{\theta}}}{\epsilon_T}
\right\},
\qquad
\epsilon_T=T^{-1/3}.
\]
Binding and nearly binding constraints remain relevant locally, while
constraints with large slack become asymptotically irrelevant.

Each multiplier draw $b$ solves the convex quadratic program
\[
h^b
=
\arg\min_{h\in\widehat{\mathcal C}_{T,\epsilon}}
\left(
\zeta_T^b+\widehat G_Th
\right)'
\widehat W_T
\left(
\zeta_T^b+\widehat G_Th
\right).
\]

Under regular conditions, Theorem 4.3 of
\citet*{hong2020numerical} implies that the conditional distribution of \(h^b\)
consistently approximates the pointwise distribution of
$\sqrt{T}
\left(
\widehat{\boldsymbol{\theta}}
-
\boldsymbol{\theta}_0
\right)$. For a smooth scalar object \(f(\boldsymbol{\theta})\), I use
\[
R_b^{\mathrm{HL-QP}}(f)
=
\nabla_{\boldsymbol{\theta}} f(\widehat{\boldsymbol{\theta}})'h^b.
\]
Standard errors and confidence intervals are constructed exactly as in the
numerical bootstrap above, replacing \(R_b(f)\) with
\(R_b^{\mathrm{HL-QP}}(f)\). 

The local QP procedure allows the set of binding inequalities to vary across draws without repeatedly estimating the original model, making it substantially cheaper than the direct numerical bootstrap. Its validity is pointwise and conditional on the selected sieve dimension, instruments, and constraint set, while its finite-sample accuracy depends on the local quadratic approximation. I therefore use the direct numerical bootstrap as the preferred pointwise inference procedure and rely on the local QP approximation when direct implementation becomes computationally prohibitive.

\end{bibunit}


\begin{thebibliography}{71}
\newcommand{\enquote}[1]{``#1''}
\expandafter\ifx\csname natexlab\endcsname\relax\def\natexlab#1{#1}\fi

\bibitem[\protect\citeauthoryear{Ai and Chen}{Ai and
  Chen}{2003}]{ai2003efficient}
\textsc{Ai, C. and X.~Chen} (2003): \enquote{Efficient estimation of models
  with conditional moment restrictions containing unknown functions,}
  \emph{Econometrica}, 71, 1795--1843.

\bibitem[\protect\citeauthoryear{Al{\'e}-Chilet, Chen, Li, and
  Reynaert}{Al{\'e}-Chilet et~al.}{2026}]{ale2026colluding}
\textsc{Al{\'e}-Chilet, J., C.~Chen, J.~Li, and M.~Reynaert} (2026):
  \enquote{Colluding against environmental regulation,} \emph{Review of
  Economic Studies}, 93, 35--71.

\bibitem[\protect\citeauthoryear{Anderson, de~Palma, and Thisse}{Anderson
  et~al.}{1992}]{anderson1992discrete}
\textsc{Anderson, S.~P., A.~de~Palma, and J.~F. Thisse} (1992): \emph{Discrete
  choice theory of product differentiation}, Cambridge, MA: MIT press.

\bibitem[\protect\citeauthoryear{Bachmayr, Dusson, Ortner, and Thomas}{Bachmayr
  et~al.}{2024}]{bachmayr2024polynomial}
\textsc{Bachmayr, M., G.~Dusson, C.~Ortner, and J.~Thomas} (2024):
  \enquote{Polynomial approximation of symmetric functions,} \emph{Mathematics
  of Computation}, 93, 811--839.

\bibitem[\protect\citeauthoryear{Backus, Conlon, and Sinkinson}{Backus
  et~al.}{2021}]{backus2021common}
\textsc{Backus, M., C.~Conlon, and M.~Sinkinson} (2021): \enquote{Common
  ownership and competition in the ready-to-eat cereal industry,} .

\bibitem[\protect\citeauthoryear{Bakhitov}{Bakhitov}{2025}]{Bakhitov2025}
\textsc{Bakhitov, E.} (2025): \enquote{On machine learning instrumental
  variable estimators,} \emph{Economics Letters}, 256, 112601.

\bibitem[\protect\citeauthoryear{Barnett and Serletis}{Barnett and
  Serletis}{2008}]{barnettserletis2008}
\textsc{Barnett, W.~A. and A.~Serletis} (2008): \enquote{Consumer preferences
  and demand systems,} \emph{Journal of Econometrics}, 147, 210--224.

\bibitem[\protect\citeauthoryear{Berry and Haile}{Berry and
  Haile}{2014}]{berry2014identification}
\textsc{Berry, S. and P.~Haile} (2014): \enquote{Identification in
  differentiated products markets using market level data,}
  \emph{Econometrica}, 82, 1749--1797.

\bibitem[\protect\citeauthoryear{Berry and Haile}{Berry and
  Haile}{2016}]{berryhaile2016}
---\hspace{-.1pt}---\hspace{-.1pt}--- (2016): \enquote{Identification in
  differentiated products markets,} \emph{Annual review of Economics}, 8,
  27--52.

\bibitem[\protect\citeauthoryear{Berry, Levinsohn, and Pakes}{Berry
  et~al.}{1995}]{BLP1995}
\textsc{Berry, S., J.~Levinsohn, and A.~Pakes} (1995): \enquote{Automobile
  prices in market equilibrium,} \emph{Econometrica}, 63, 841--890.

\bibitem[\protect\citeauthoryear{Berry, Levinsohn, and Pakes}{Berry
  et~al.}{2004}]{berry2004differentiated}
---\hspace{-.1pt}---\hspace{-.1pt}--- (2004): \enquote{Differentiated products
  demand systems from a combination of micro and macro data: The new car
  market,} \emph{Journal of political Economy}, 112, 68--105.

\bibitem[\protect\citeauthoryear{Berry}{Berry}{1994}]{berry1994estimating}
\textsc{Berry, S.~T.} (1994): \enquote{Estimating discrete-choice models of
  product differentiation,} \emph{The RAND Journal of Economics}, 25, 242--262.

\bibitem[\protect\citeauthoryear{Berry and Haile}{Berry and
  Haile}{2021}]{berryhaile2021}
\textsc{Berry, S.~T. and P.~A. Haile} (2021): \enquote{Foundations of demand
  estimation,} in \emph{Handbook of industrial organization}, Elsevier, vol.~4,
  1--62.

\bibitem[\protect\citeauthoryear{Birchall, Mohapatra, and Verboven}{Birchall
  et~al.}{2024}]{birchall2024estimating}
\textsc{Birchall, C., D.~Mohapatra, and F.~Verboven} (2024):
  \enquote{Estimating substitution patterns and demand curvature in
  discrete-choice models of product differentiation,} \emph{Review of Economics
  and Statistics}, 1--40.

\bibitem[\protect\citeauthoryear{Bj\"ornerstedt and Verboven}{Bj\"ornerstedt
  and Verboven}{2016}]{BjornerstedtVerboven2016}
\textsc{Bj\"ornerstedt, J. and F.~Verboven} (2016): \enquote{Does merger
  simulation work? Evidence from the Swedish analgesics market,} \emph{American
  Economic Journal: Applied Economics}, 8, 125--64.

\bibitem[\protect\citeauthoryear{Blundell, Horowitz, and Parey}{Blundell
  et~al.}{2012}]{blundell2012measuring}
\textsc{Blundell, R., J.~L. Horowitz, and M.~Parey} (2012): \enquote{Measuring
  the price responsiveness of gasoline demand: Economic shape restrictions and
  nonparametric demand estimation,} \emph{Quantitative Economics}, 3, 29--51.

\bibitem[\protect\citeauthoryear{Borusyak, Chen, Hull, and Lei}{Borusyak
  et~al.}{2026}]{borusyak2026nonparametric}
\textsc{Borusyak, K., J.~Chen, P.~Hull, and L.~Lei} (2026):
  \enquote{Nonparametric identification of demand without exogenous product
  characteristics,} .

\bibitem[\protect\citeauthoryear{Brand and Smith}{Brand and
  Smith}{2025}]{brand2025quasi}
\textsc{Brand, J. and A.~N. Smith} (2025): \enquote{A Quasi-Bayes Approach to
  Nonparametric Demand Estimation with Economic Constraints,} \emph{Available
  at SSRN 5100826}.

\bibitem[\protect\citeauthoryear{Bronnenberg, Kruger, and Mela}{Bronnenberg
  et~al.}{2008}]{bronnenberg2008database}
\textsc{Bronnenberg, B.~J., M.~W. Kruger, and C.~F. Mela} (2008):
  \enquote{Database paper - The IRI marketing data set,} \emph{Marketing
  science}, 27, 745--748.

\bibitem[\protect\citeauthoryear{Chen}{Chen}{2007}]{Chen2007}
\textsc{Chen, X.} (2007): \enquote{Large Sample Sieve Estimation of
  Semi-Nonparametric Models,} in \emph{Handbook of Econometrics}, ed. by J.~J.
  Heckman and E.~E. Leamer, Elsevier, vol.~6B.

\bibitem[\protect\citeauthoryear{Chen and Pouzo}{Chen and
  Pouzo}{2012}]{chen2012estimation}
\textsc{Chen, X. and D.~Pouzo} (2012): \enquote{Estimation of nonparametric
  conditional moment models with possibly nonsmooth generalized residuals,}
  \emph{Econometrica}, 80, 277--321.

\bibitem[\protect\citeauthoryear{Chetverikov, Santos, and Shaikh}{Chetverikov
  et~al.}{2018}]{chetverikov2018}
\textsc{Chetverikov, D., A.~Santos, and A.~M. Shaikh} (2018): \enquote{The
  econometrics of shape restrictions,} \emph{Annual Review of Economics}, 10,
  31--63.

\bibitem[\protect\citeauthoryear{Chetverikov and Wilhelm}{Chetverikov and
  Wilhelm}{2017}]{chetverikov2017nonparametric}
\textsc{Chetverikov, D. and D.~Wilhelm} (2017): \enquote{Nonparametric
  instrumental variable estimation under monotonicity,} \emph{Econometrica},
  85, 1303--1320.

\bibitem[\protect\citeauthoryear{Compiani}{Compiani}{2022}]{compiani2018}
\textsc{Compiani, G.} (2022): \enquote{Market counterfactuals and the
  specification of multiproduct demand: A nonparametric approach,}
  \emph{Quantitative Economics}, 13, 545--591.

\bibitem[\protect\citeauthoryear{Conlon and Gortmaker}{Conlon and
  Gortmaker}{2020}]{conlon2018}
\textsc{Conlon, C. and J.~Gortmaker} (2020): \enquote{Best practices for
  differentiated products demand estimation with pyblp,} \emph{The RAND Journal
  of Economics}, 51, 1108--1161.

\bibitem[\protect\citeauthoryear{Conlon and Gortmaker}{Conlon and
  Gortmaker}{2025}]{conlon2023incorporating}
---\hspace{-.1pt}---\hspace{-.1pt}--- (2025): \enquote{Incorporating micro data
  into differentiated products demand estimation with PyBLP,} \emph{Journal of
  Econometrics}, 105926.

\bibitem[\protect\citeauthoryear{Daly and Zachary}{Daly and
  Zachary}{1979}]{dalyzachary}
\textsc{Daly, A. and S.~Zachary} (1979): \enquote{Improved multiple choice
  models,} in \emph{Determinants of Travel Choice}, ed. by D. Hensher and Q.
  Dalvi (eds), London: Teakfield, 337--357.

\bibitem[\protect\citeauthoryear{Davis and Schiraldi}{Davis and
  Schiraldi}{2014}]{davis2014flexible}
\textsc{Davis, P. and P.~Schiraldi} (2014): \enquote{The flexible coefficient
  multinomial logit (FC-MNL) model of demand for differentiated products,}
  \emph{The RAND Journal of Economics}, 45, 32--63.

\bibitem[\protect\citeauthoryear{D{\"o}pper, MacKay, Miller, and
  Stiebale}{D{\"o}pper et~al.}{2025}]{dopper2025rising}
\textsc{D{\"o}pper, H., A.~MacKay, N.~H. Miller, and J.~Stiebale} (2025):
  \enquote{Rising markups and the role of consumer preferences,} \emph{Journal
  of Political Economy}, 133, 2462--2505.

\bibitem[\protect\citeauthoryear{Fally and Ligon}{Fally and
  Ligon}{2026}]{fally2026consumer}
\textsc{Fally, T. and E.~Ligon} (2026): \enquote{Consumer demand with price
  aggregators and low-rank cross-price effects,} .

\bibitem[\protect\citeauthoryear{Florens, Johannes, and Van~Bellegem}{Florens
  et~al.}{2012}]{florens2012instrumental}
\textsc{Florens, J.-P., J.~Johannes, and S.~Van~Bellegem} (2012):
  \enquote{Instrumental regression in partially linear models,} \emph{The
  Econometrics Journal}, 15, 304--324.

\bibitem[\protect\citeauthoryear{Fosgerau, Monardo, and de~Palma}{Fosgerau
  et~al.}{2024}]{fosgerau2018demand}
\textsc{Fosgerau, M., J.~Monardo, and A.~de~Palma} (2024): \enquote{The inverse
  product differentiation logit model,} \emph{American Economic Journal:
  Microeconomics}, 16, 329--370.

\bibitem[\protect\citeauthoryear{Gandhi and Houde}{Gandhi and
  Houde}{2026}]{gandhihoude}
\textsc{Gandhi, A. and J.-F. Houde} (2026): \enquote{Measuring Substitution
  Patterns in Differentiated-Products Industries,} \emph{Available on author's
  website}.

\bibitem[\protect\citeauthoryear{Goldszmidt, Henry~de Frahan, List, Mogstad,
  Muir, and Torgovitsky}{Goldszmidt et~al.}{2026}]{goldszmidt2026price}
\textsc{Goldszmidt, A., L.~Henry~de Frahan, J.~A. List, M.~Mogstad, I.~Muir,
  and A.~Torgovitsky} (2026): \enquote{Price Elasticities in Logit Demand
  Models,} .

\bibitem[\protect\citeauthoryear{Hastie, Tibshirani, and Friedman}{Hastie
  et~al.}{2009}]{hastie09elements}
\textsc{Hastie, T., R.~Tibshirani, and J.~Friedman} (2009): \emph{The Elements
  of Statistical Learning: Data Mining, Inference, and Prediction}, Springer
  Series in Statistics, Springer New York, 2 ed.

\bibitem[\protect\citeauthoryear{Hong and Li}{Hong and
  Li}{2020}]{hong2020numerical}
\textsc{Hong, H. and J.~Li} (2020): \enquote{THE Numerical Bootstrap,}
  \emph{The Annals of Statistics}, 48, 397--412.

\bibitem[\protect\citeauthoryear{Horowitz}{Horowitz}{2011}]{horowitz2011applied}
\textsc{Horowitz, J.~L.} (2011): \enquote{Applied nonparametric instrumental
  variables estimation,} \emph{Econometrica}, 79, 347--394.

\bibitem[\protect\citeauthoryear{Horowitz}{Horowitz}{2012}]{horowitz2012semiparametric}
---\hspace{-.1pt}---\hspace{-.1pt}--- (2012): \emph{Semiparametric methods in
  econometrics}, vol. 131, Springer Science \& Business Media.

\bibitem[\protect\citeauthoryear{Horta\c{c}su, Lieber, Monardo, and
  de~Paula}{Horta\c{c}su et~al.}{2024}]{Hortacsu2020}
\textsc{Horta\c{c}su, A., J.~Lieber, J.~Monardo, and A.~de~Paula} (2024):
  \enquote{Estimating Nesting Structures,} \emph{Working paper}.

\bibitem[\protect\citeauthoryear{Houde}{Houde}{2012}]{houde2012spatial}
\textsc{Houde, J.-F.} (2012): \enquote{Spatial differentiation and vertical
  mergers in retail markets for gasoline,} \emph{American Economic Review},
  102, 2147--2182.

\bibitem[\protect\citeauthoryear{Ichimura and Todd}{Ichimura and
  Todd}{2007}]{ichimura2007implementing}
\textsc{Ichimura, H. and P.~E. Todd} (2007): \enquote{Implementing
  nonparametric and semiparametric estimators,} \emph{Handbook of
  econometrics}, 6, 5369--5468.

\bibitem[\protect\citeauthoryear{Knittel and Metaxoglou}{Knittel and
  Metaxoglou}{2014}]{knittel2014estimation}
\textsc{Knittel, C.~R. and K.~Metaxoglou} (2014): \enquote{Estimation of
  Random-Coefficient Demand Models: Two Empiricists' Perspective,} \emph{Review
  of Economics and Statistics}, 96, 34--59.

\bibitem[\protect\citeauthoryear{Komiyama and Shimao}{Komiyama and
  Shimao}{2018}]{komiyama2018cross}
\textsc{Komiyama, J. and H.~Shimao} (2018): \enquote{Cross validation based
  model selection via generalized method of moments,} \emph{arXiv preprint
  arXiv:1807.06993}.

\bibitem[\protect\citeauthoryear{Koning and Ridder}{Koning and
  Ridder}{2003}]{koning2003}
\textsc{Koning, R.~H. and G.~Ridder} (2003): \enquote{Discrete choice and
  stochastic utility maximization,} \emph{The Econometrics Journal}, 6, 1--27.

\bibitem[\protect\citeauthoryear{Ledoit and Wolf}{Ledoit and
  Wolf}{2004}]{ledoit2004well}
\textsc{Ledoit, O. and M.~Wolf} (2004): \enquote{A well-conditioned estimator
  for large-dimensional covariance matrices,} \emph{Journal of multivariate
  analysis}, 88, 365--411.

\bibitem[\protect\citeauthoryear{Lu, Shi, and Tao}{Lu
  et~al.}{2023}]{lu2023semi}
\textsc{Lu, Z., X.~Shi, and J.~Tao} (2023): \enquote{Semi-nonparametric
  estimation of random coefficients logit model for aggregate demand,}
  \emph{Journal of Econometrics}, 235, 2245--2265.

\bibitem[\protect\citeauthoryear{Matzkin}{Matzkin}{1994}]{matzkin1994}
\textsc{Matzkin, R.~L.} (1994): \enquote{Restrictions of economic theory in
  nonparametric methods,} \emph{Handbook of econometrics}, 4, 2523--2558.

\bibitem[\protect\citeauthoryear{McFadden and Train}{McFadden and
  Train}{2000}]{mcfadden2000mixed}
\textsc{McFadden, D. and K.~Train} (2000): \enquote{Mixed MNL models for
  discrete response,} \emph{Journal of applied Econometrics}, 15, 447--470.

\bibitem[\protect\citeauthoryear{McFadden}{McFadden}{1984}]{mcfadden1984econometric}
\textsc{McFadden, D.~L.} (1984): \enquote{Econometric analysis of qualitative
  response models,} \emph{Handbook of econometrics}, 2, 1395--1457.

\bibitem[\protect\citeauthoryear{Miller and Weinberg}{Miller and
  Weinberg}{2017}]{miller2017understanding}
\textsc{Miller, N.~H. and M.~C. Weinberg} (2017): \enquote{Understanding the
  price effects of the MillerCoors joint venture,} \emph{Econometrica}, 85,
  1763--1791.

\bibitem[\protect\citeauthoryear{Miravete, Seim, and Thurk}{Miravete
  et~al.}{2023}]{miravete2023pass}
\textsc{Miravete, E., K.~Seim, and J.~Thurk} (2023): \enquote{Pass-through and
  tax incidence in differentiated product markets,} \emph{International Journal
  of Industrial Organization}, 102985.

\bibitem[\protect\citeauthoryear{Miravete, Seim, and Thurk}{Miravete
  et~al.}{2026}]{miravete2023elasticity}
---\hspace{-.1pt}---\hspace{-.1pt}--- (2026): \enquote{Elasticity and Curvature
  of Discrete Choice Demand Models,} .

\bibitem[\protect\citeauthoryear{Monardo}{Monardo}{2026}]{monardo2019chap3}
\textsc{Monardo, J.} (2026): \enquote{Functional Form and Shape Restrictions in
  Discrete Choice Models,} \emph{Available at SSRN 4681913}.

\bibitem[\protect\citeauthoryear{{MOSEK ApS}}{{MOSEK ApS}}{2024}]{mosek}
\textsc{{MOSEK ApS}} (2024): \enquote{MOSEK Optimizer,}
  \url{https://www.mosek.com/}, version 10.2.1.

\bibitem[\protect\citeauthoryear{Mr{\'a}zov{\'a} and Neary}{Mr{\'a}zov{\'a} and
  Neary}{2017}]{mrazova2017}
\textsc{Mr{\'a}zov{\'a}, M. and J.~P. Neary} (2017): \enquote{Not so demanding:
  Demand structure and firm behavior,} \emph{American Economic Review}, 107,
  3835--74.

\bibitem[\protect\citeauthoryear{Nevo}{Nevo}{2000}]{NevoRand}
\textsc{Nevo, A.} (2000): \enquote{Mergers with Differentiated Products: The
  Case of the Ready-to-Eat Cereal Industry,} \emph{The RAND Journal of
  Economics}, 31, 395--421.

\bibitem[\protect\citeauthoryear{Nevo, Tao, and Gandhi}{Nevo
  et~al.}{2026}]{Nevo2026}
\textsc{Nevo, A., J.~Tao, and A.~Gandhi} (2026): \enquote{Flexible Estimation
  of Differentiated Products Demand Models Using Aggregate Data,}
  \emph{Available on author's website}.

\bibitem[\protect\citeauthoryear{Newey and Powell}{Newey and
  Powell}{2003}]{newey2003instrumental}
\textsc{Newey, W.~K. and J.~L. Powell} (2003): \enquote{Instrumental variable
  estimation of nonparametric models,} \emph{Econometrica}, 71, 1565--1578.

\bibitem[\protect\citeauthoryear{Pakes and McGuire}{Pakes and
  McGuire}{1994}]{pakes1994computing}
\textsc{Pakes, A. and P.~McGuire} (1994): \enquote{Computing Markov-perfect
  Nash equilibria: Numerical implications of a dynamic differentiated product
  model,} \emph{The Rand Journal of Economics}, 25, 555.

\bibitem[\protect\citeauthoryear{Petrin}{Petrin}{2002}]{petrin2002}
\textsc{Petrin, A.} (2002): \enquote{Quantifying the benefits of new products:
  The case of the minivan,} \emph{Journal of political Economy}, 110, 705--729.

\bibitem[\protect\citeauthoryear{Pinkse and Slade}{Pinkse and
  Slade}{2004}]{pinkse2004mergers}
\textsc{Pinkse, J. and M.~E. Slade} (2004): \enquote{Mergers, brand
  competition, and the price of a pint,} \emph{European Economic Review}, 48,
  617--643.

\bibitem[\protect\citeauthoryear{Pinkse, Slade, and Brett}{Pinkse
  et~al.}{2002}]{pinkse2002spatial}
\textsc{Pinkse, J., M.~E. Slade, and C.~Brett} (2002): \enquote{Spatial price
  competition: a semiparametric approach,} \emph{Econometrica}, 70, 1111--1153.

\bibitem[\protect\citeauthoryear{Rivers and Vuong}{Rivers and
  Vuong}{2002}]{rivers2002model}
\textsc{Rivers, D. and Q.~Vuong} (2002): \enquote{Model selection tests for
  nonlinear dynamic models,} \emph{The Econometrics Journal}, 5, 1--39.

\bibitem[\protect\citeauthoryear{Salani{\'e} and Wolak}{Salani{\'e} and
  Wolak}{2022}]{salanie2022}
\textsc{Salani{\'e}, B. and F.~A. Wolak} (2022): \enquote{Fast, detail-free,
  and approximately correct: Estimating mixed demand systems,} .

\bibitem[\protect\citeauthoryear{Shen, Wang, and Yu}{Shen
  et~al.}{2014}]{shen2014approximations}
\textsc{Shen, J., L.-L. Wang, and H.~Yu} (2014): \enquote{Approximations by
  orthonormal mapped Chebyshev functions for higher-dimensional problems in
  unbounded domains,} \emph{Journal of computational and applied mathematics},
  265, 264--275.

\bibitem[\protect\citeauthoryear{Singh, Liu, and Yoganarasimhan}{Singh
  et~al.}{2026}]{singh2023choice}
\textsc{Singh, A., Y.~Liu, and H.~Yoganarasimhan} (2026): \enquote{Choice
  Models and Permutation Invariance: Deep Demand Estimation in Differentiated
  Products Markets,} \emph{Marketing Science}, published online July 28, 2026.

\bibitem[\protect\citeauthoryear{Slade}{Slade}{2004}]{slade2004}
\textsc{Slade, M.~E.} (2004): \enquote{Market power and joint dominance in UK
  brewing,} \emph{The Journal of Industrial Economics}, 52, 133--163.

\bibitem[\protect\citeauthoryear{Starc and Wollmann}{Starc and
  Wollmann}{2025}]{starc2025does}
\textsc{Starc, A. and T.~G. Wollmann} (2025): \enquote{Does entry remedy
  collusion? Evidence from the generic prescription drug cartel,}
  \emph{American Economic Review}, 115, 1400--1438.

\bibitem[\protect\citeauthoryear{Wang}{Wang}{2023}]{wang2021sieve}
\textsc{Wang, A.} (2023): \enquote{Sieve BLP: A semi-nonparametric model of
  demand for differentiated products,} \emph{Journal of Econometrics}, 235,
  325--351.

\bibitem[\protect\citeauthoryear{Wang and Huang}{Wang and
  Huang}{2025}]{wang2025scalable}
\textsc{Wang, J. and Y.~Huang} (2025): \enquote{Scalable just-in-time price
  elasticity estimation,} \emph{Available at SSRN 3557359}.

\bibitem[\protect\citeauthoryear{Zaheer, Kottur, Ravanbakhsh, Poczos,
  Salakhutdinov, and Smola}{Zaheer et~al.}{2017}]{zaheer2017}
\textsc{Zaheer, M., S.~Kottur, S.~Ravanbakhsh, B.~Poczos, R.~R. Salakhutdinov,
  and A.~J. Smola} (2017): \enquote{Deep sets,} \emph{Advances in neural
  information processing systems}, 30.

\end{thebibliography}

\begin{thebibliography}{25}
\newcommand{\enquote}[1]{``#1''}
\expandafter\ifx\csname natexlab\endcsname\relax\def\natexlab#1{#1}\fi

\bibitem[\protect\citeauthoryear{Auerbach}{Auerbach}{1985}]{auerbach1985theory}
\textsc{Auerbach, A.~J.} (1985): \enquote{The theory of excess burden and
  optimal taxation,} in \emph{Handbook of public economics}, Elsevier, vol.~1,
  61--127.

\bibitem[\protect\citeauthoryear{Bachmayr, Dusson, Ortner, and Thomas}{Bachmayr
  et~al.}{2024}]{bachmayr2024polynomial}
\textsc{Bachmayr, M., G.~Dusson, C.~Ortner, and J.~Thomas} (2024):
  \enquote{Polynomial approximation of symmetric functions,} \emph{Mathematics
  of Computation}, 93, 811--839.

\bibitem[\protect\citeauthoryear{Berry}{Berry}{1994}]{berry1994estimating}
\textsc{Berry, S.~T.} (1994): \enquote{Estimating discrete-choice models of
  product differentiation,} \emph{The RAND Journal of Economics}, 25, 242--262.

\bibitem[\protect\citeauthoryear{Bhattacharya}{Bhattacharya}{2018}]{bhattacharya2018empirical}
\textsc{Bhattacharya, D.} (2018): \enquote{Empirical welfare analysis for
  discrete choice: Some general results,} \emph{Quantitative Economics}, 9,
  571--615.

\bibitem[\protect\citeauthoryear{Boyd and Vandenberghe}{Boyd and
  Vandenberghe}{2004}]{boyd2004convex}
\textsc{Boyd, S. and L.~Vandenberghe} (2004): \emph{Convex optimization},
  Cambridge, U.K.: Cambridge University Press.

\bibitem[\protect\citeauthoryear{Chaves, Clark, Duarte, and Houde}{Chaves
  et~al.}{2025}]{chaves2025misallocation}
\textsc{Chaves, D., R.~Clark, M.~Duarte, and J.-F. Houde} (2025):
  \enquote{Misallocation in Differentiated Product Markets: Evidence from
  Collusion under Uniform Pricing,} .

\bibitem[\protect\citeauthoryear{Chen}{Chen}{2007}]{chen2007large}
\textsc{Chen, X.} (2007): \enquote{Large sample sieve estimation of
  semi-nonparametric models,} \emph{Handbook of econometrics}, 6, 5549--5632.

\bibitem[\protect\citeauthoryear{Chen, Christensen, and Kankanala}{Chen
  et~al.}{2025}]{chen2025adaptive}
\textsc{Chen, X., T.~Christensen, and S.~Kankanala} (2025): \enquote{Adaptive
  estimation and uniform confidence bands for nonparametric structural
  functions and elasticities,} \emph{Review of Economic Studies}, 92, 162--196.

\bibitem[\protect\citeauthoryear{Chipman and Moore}{Chipman and
  Moore}{1980}]{chipman1980compensating}
\textsc{Chipman, J.~S. and J.~C. Moore} (1980): \enquote{Compensating
  variation, consumer's surplus, and welfare,} \emph{The American Economic
  Review}, 70, 933--949.

\bibitem[\protect\citeauthoryear{Compiani}{Compiani}{2022}]{compiani2018}
\textsc{Compiani, G.} (2022): \enquote{Market counterfactuals and the
  specification of multiproduct demand: A nonparametric approach,}
  \emph{Quantitative Economics}, 13, 545--591.

\bibitem[\protect\citeauthoryear{Conlon and Gortmaker}{Conlon and
  Gortmaker}{2020}]{conlon2018}
\textsc{Conlon, C. and J.~Gortmaker} (2020): \enquote{Best practices for
  differentiated products demand estimation with pyblp,} \emph{The RAND Journal
  of Economics}, 51, 1108--1161.

\bibitem[\protect\citeauthoryear{Dellacherie, Mart{\'\i}nez, and
  San~Mart{\'\i}n}{Dellacherie et~al.}{2020}]{dellacherie2020inverse}
\textsc{Dellacherie, C., S.~Mart{\'\i}nez, and J.~San~Mart{\'\i}n} (2020):
  \enquote{Inverse M-matrix, a new characterization,} \emph{Linear Algebra and
  its Applications}, 595, 182--191.

\bibitem[\protect\citeauthoryear{Fosgerau, Monardo, and de~Palma}{Fosgerau
  et~al.}{2024}]{fosgerau2018demand}
\textsc{Fosgerau, M., J.~Monardo, and A.~de~Palma} (2024): \enquote{The inverse
  product differentiation logit model,} \emph{American Economic Journal:
  Microeconomics}, 16, 329--370.

\bibitem[\protect\citeauthoryear{Hong and Li}{Hong and
  Li}{2020}]{hong2020numerical}
\textsc{Hong, H. and J.~Li} (2020): \enquote{THE Numerical Bootstrap,}
  \emph{The Annals of Statistics}, 48, 397--412.

\bibitem[\protect\citeauthoryear{Horn and Johnson}{Horn and
  Johnson}{2013}]{horn2013matrix}
\textsc{Horn, R. and C.~Johnson} (2013): \emph{Matrix Analysis}, Matrix
  Analysis, Cambridge University Press.

\bibitem[\protect\citeauthoryear{Komiyama and Shimao}{Komiyama and
  Shimao}{2018}]{komiyama2018cross}
\textsc{Komiyama, J. and H.~Shimao} (2018): \enquote{Cross validation based
  model selection via generalized method of moments,} \emph{arXiv preprint
  arXiv:1807.06993}.

\bibitem[\protect\citeauthoryear{Ledoit and Wolf}{Ledoit and
  Wolf}{2004}]{ledoit2004well}
\textsc{Ledoit, O. and M.~Wolf} (2004): \enquote{A well-conditioned estimator
  for large-dimensional covariance matrices,} \emph{Journal of multivariate
  analysis}, 88, 365--411.

\bibitem[\protect\citeauthoryear{Lu, Shi, and Tao}{Lu
  et~al.}{2023}]{lu2023semi}
\textsc{Lu, Z., X.~Shi, and J.~Tao} (2023): \enquote{Semi-nonparametric
  estimation of random coefficients logit model for aggregate demand,}
  \emph{Journal of Econometrics}, 235, 2245--2265.

\bibitem[\protect\citeauthoryear{McFadden}{McFadden}{1984}]{mcfadden1984econometric}
\textsc{McFadden, D.~L.} (1984): \enquote{Econometric analysis of qualitative
  response models,} \emph{Handbook of econometrics}, 2, 1395--1457.

\bibitem[\protect\citeauthoryear{Miller and Weinberg}{Miller and
  Weinberg}{2017}]{miller2017understanding}
\textsc{Miller, N.~H. and M.~C. Weinberg} (2017): \enquote{Understanding the
  price effects of the MillerCoors joint venture,} \emph{Econometrica}, 85,
  1763--1791.

\bibitem[\protect\citeauthoryear{Monardo}{Monardo}{2026}]{monardo2025product}
\textsc{Monardo, J.} (2026): \enquote{Product-Dependent Price Coefficients in
  Discrete Choice Models,} \emph{Available at SSRN 4624771}.

\bibitem[\protect\citeauthoryear{Nevo, Tao, and Gandhi}{Nevo
  et~al.}{2026}]{Nevo2026}
\textsc{Nevo, A., J.~Tao, and A.~Gandhi} (2026): \enquote{Flexible Estimation
  of Differentiated Products Demand Models Using Aggregate Data,}
  \emph{Available on author's website}.

\bibitem[\protect\citeauthoryear{Newey and McFadden}{Newey and
  McFadden}{1994}]{newey1994large}
\textsc{Newey, W.~K. and D.~McFadden} (1994): \enquote{Large sample estimation
  and hypothesis testing,} \emph{Handbook of econometrics}, 4, 2111--2245.

\bibitem[\protect\citeauthoryear{Rivers and Vuong}{Rivers and
  Vuong}{2002}]{rivers2002model}
\textsc{Rivers, D. and Q.~Vuong} (2002): \enquote{Model selection tests for
  nonlinear dynamic models,} \emph{The Econometrics Journal}, 5, 1--39.

\bibitem[\protect\citeauthoryear{Rockafellar and Wets}{Rockafellar and
  Wets}{2009}]{rockafellar2009variational}
\textsc{Rockafellar, R.~T. and R.~J.-B. Wets} (2009): \emph{Variational
  analysis}, vol. 317, Springer Science \& Business Media.

\end{thebibliography}
\end{document}